\documentclass{article} 
\usepackage{tem,times}
\usepackage{array}
\usepackage{amsmath}
\usepackage{amssymb}
\usepackage{amsthm}
\usepackage{mathtools}
\usepackage{booktabs}
\usepackage{microtype}
\usepackage{algorithm}
\usepackage{algpseudocode}
\usepackage[hypertexnames=false]{hyperref}
\usepackage{url}
\usepackage[nameinlink,capitalize,noabbrev]{cleveref}

\usepackage{amsmath,amsfonts,bm}

\def\eqref#1{equation~\ref{#1}}

\def\1{\bm{1}}

\DeclareMathAlphabet{\mathsfit}{\encodingdefault}{\sfdefault}{m}{sl}
\SetMathAlphabet{\mathsfit}{bold}{\encodingdefault}{\sfdefault}{bx}{n}

\newtheorem{theorem}{Theorem}
\newtheorem{proposition}[theorem]{Proposition}
\newtheorem{lemma}[theorem]{Lemma}
\newtheorem{corollary}[theorem]{Corollary}
\theoremstyle{remark}
\newtheorem{remark}[theorem]{Remark}

\newcommand{\F}{\mathbb F_2}
\newcommand{\I}{\mathcal I}
\newcommand{\N}{\mathcal N}
\newcommand{\Prb}{\mathbb P}

\title{Learning Latent Algebraic Structure from Ambiguous Set Observations}

\author{Cheng Peng}

\begin{document}

\maketitle

\begin{abstract}
We study when statistically learnable latent structure can
also be recovered efficiently, and how membership queries
change the answer. An unknown support
$A\subseteq\mathbb F_2^n$ has small additive doubling and
is observed through a fixed set $B$ satisfying
$|A\triangle B|\le\eta|A|$. We seek one linear subspace $V$ such that every compatible
support $A$ is covered by few $V$-cosets and satisfies
$|V|\le|A|$.
For every $\eta<1$, polynomially many uniform samples
suffice statistically, with cost polynomial in the doubling
constant and proportional to $(1-\eta)^{-1}$; this radius
dependence is sharp. Under a specified hardness assumption for learning parities
with noise (search-LPN), however, no polynomial-time
sample-only learner achieves even constant covering cost,
including when the latent support is unique.
At fixed structural parameters and the same constant
covering budget, adding exact membership queries to $B$
permits polynomial-time recovery. The general query learner constructs a short
structural list and uses fresh samples to select one common
output through a majority-coverage rule. Persistent structured
cores make this candidate construction possible.
At doubling one, a complementary distinction appears at
$\eta=1/3$: coarse recovery remains polynomial time, while
exact recovery requires exponentially many accesses in the
worst case when latent cardinality is unknown.

\end{abstract}

\section{Introduction}
When observations leave several latent explanations, a learner
can seek a structural description valid for all of them.
Two questions then arise: do passive samples contain enough
information to learn such a description, and can it be
computed efficiently? We study these questions in a discrete
algebraic model and compare passive sampling with membership
access to the observed support.

Ambiguity already arises when every observed point is genuine.
Let $U,X\le\mathbb F_2^n$ have trivial intersection.
The observation $B=U$ may come from either $A_0=U$ or
$A_1=U\oplus X$ after deletions. If $\dim X=t$, both
explanations are compatible at deletion radius $1-2^{-t}$.
Access to $B$ cannot reveal which unseen directions were
present. Our target is a compact algebraic description
shared by all compatible explanations.

Our latent support is a nonempty set $A\subseteq G:=\mathbb F_2^n$ with small additive doubling, $|A+A|\le K|A|$. The observed support $B\subseteq G$ satisfies
\begin{equation}
\label{eq:v3-corruption-model}
|A\triangle B|\le\eta|A|,
\qquad 0\le\eta<1.
\end{equation}
A concrete example is a union of $s$ distinct cosets of a subspace $V_0$. Such a set has doubling at most $s$ and consists of shared linear directions together with $s$ offsets. Small doubling allows more general supports, including
sets that occupy only part of each covering coset.

We seek a linear subspace $V$ whose cosets cover every
compatible support:
\[
A\subseteq\bigcup_{j=1}^{q}(z_j+V),
\qquad |V|\le|A|.
\]
The size constraint rules out the ambient space as a trivial answer. After quotienting by $V$, the support occupies at most $q$ positions; each coset may contain only part of the support. The learned $V$ must work simultaneously for all supports compatible with the same $B$, while the offsets may depend on the support. We call this a \emph{common structural hypothesis}.

Why should one hypothesis serve several latent explanations? Deletions and insertions consume the same budget in \eqref{eq:v3-corruption-model}, which implies
\begin{equation}
\label{eq:v3-majority}
|B|\le(1+\eta)|A\cap B|,
\qquad
\frac{|A\cap B|}{|B|}\ge\frac1{1+\eta}>\frac12.
\end{equation}
Every explanation therefore regards a strict majority of $B$ as genuine. More strongly, any set covering half of $B$ intersects each compatible $A$ in at least $(1-\eta)|A|/2$ points. For a union of a few parallel cosets, small doubling turns this overlap into a cover of all of $A$. This majority-coverage principle underlies both of our positive learning results.

\paragraph{Samples suffice statistically.}
Clean polynomial Freiman--Ruzsa (PFR) structure supplies a bounded-coset cover for each latent support~\citep{gowers2025marton}. We search over such descriptions, retain those covering an empirical majority, and choose one with minimum subspace dimension. Uniform convergence makes the coverage test reliable; minimum dimension enforces the size constraint for every compatible support. Polynomially many samples suffice for every $\eta<1$, with cost polynomial in $K$ and proportional to $(1-\eta)^{-1}$. The exhaustive search takes exponential time.

\paragraph{Computational hardness despite unique identification.}
Could a different sample-only algorithm find a coarse cover
efficiently? Under search-LPN hardness at a specified
constant noise rate, no polynomial-time learner achieves
even constant covering cost, already at $K=1$ and a fixed
positive radius below $1/3$, where the latent affine support
is unique. The lower bound allows polynomial-size output
lists and knowledge of the latent cardinality.

\paragraph{Queries provide a constructive route.}
Adding exact membership queries to $B$ lets us generate a short family of candidates containing a serving entry for every compatible support. Fresh samples then apply the majority-coverage test and minimum-dimension rule to select one common output. The general learner has polynomial sample complexity and XP query/time complexity. Under moderate corruption, a single core suffices and the bound is FPT in $K$. For fixed structural parameters the query learner is polynomial time; its cover bound is larger than the statistical bound. Section~\ref{sec:v3-main-results} makes the access comparison at a common fixed quality budget.

Writing $\varepsilon=1-\eta$ and $p(K)$ for a clean PFR covering bound, the results are:
\begin{center}
\small
\setlength{\tabcolsep}{5pt}
\renewcommand{\arraystretch}{1.15}
\begin{tabular}{@{}>{\raggedright\arraybackslash}p{0.14\linewidth}>{\raggedright\arraybackslash}p{0.14\linewidth}>{\raggedright\arraybackslash}p{0.66\linewidth}@{}}
\toprule
\textbf{Access} & \textbf{Regime} & \textbf{Guarantee}\\
\midrule
Samples & All $\eta<1$
& \textbf{Statistical sufficiency:} common $V$ with cost $2Kp(K)/\varepsilon$; polynomial samples, exhaustive time.\\
Samples & Fixed $K=1$, $0<\eta<1/3$
& \textbf{Computational barrier:} no polynomial-time constant-cost learner under the stated LPN assumption.\\
\midrule
Samples + MQ & $\eta\le c_{\rm rob}/\sqrt K$
& \textbf{FPT computation:} one common $V$; FPT in $K$.\\
Samples + MQ & All $\eta<1$
& \textbf{XP computation:} one common $V$; polynomial samples, XP queries and time.\\
\bottomrule
\end{tabular}
\end{center}

Exact recovery sharpens this contrast. At $K=1$ and $\eta=1/3$, recovering the unique latent support can require exponentially many accesses when its cardinality is unknown, whereas a coarse cover is polynomial-time computable. Section~\ref{sec:v3-exact} develops this boundary.

\paragraph{Technical contribution.}
Candidate generation requires reusable access to a structured
core of a potentially exponentially large support.
The downstream PFR routine must query and sample the same
realized set, with size and doubling guarantees for the entire
core. We construct such a core at the sharp $\sqrt{\alpha}$
retained-mass scale~\citep{reiherschoen2024bsg}.
Weighted proposals make useful candidates accessible;
shared sketches and independent validation and cleaning
produce a finite record defining one deterministic subset.
The record supports membership queries and exact uniform
sampling conditional on finite-block success.
Crucially for residual peeling, the extractor bounds the
probability of returning an invalid core even without an
energy promise, while high energy guarantees success with
high probability.
Earlier implicit BSG algorithms provide randomized tests
and sampling with set-sandwich guarantees~\citep{tulsiani2014quadratic}.
Section~\ref{sec:v3-compiler} develops the fixed-set interface
used here; Appendix~\ref{app:related-work} compares the
related constructions.

\paragraph{Related work.}
The common-target viewpoint connects to latent identifiability under partial observation~\citep{yao2024multiview,syrota2025metric}. Robust and list-decodable learning also address ambiguity~\citep{charikar2017untrusted,raghavendra2020listsubspace,bakshi2021listsubspace}; our insertion--deletion budget yields the majority property above and permits a common output. Finite-field subspace learning already captures LPN in a nested-mixture setting~\citep{chen2021mixture}, and noisy-parity assumptions underlie earlier query-access comparisons~\citep{feldman2009faulty}. Here the reduction targets coarse quotient covers of a fixed corrupted support. On the algorithmic side, we combine persistent core extraction with clean algorithmic PFR~\citep{arunachalam2026classicalpfr,castrosilva2026algorithmicpfr}. Appendix~\ref{app:related-work} compares the targets and oracle guarantees in detail.

\section{Problem Formulation and Main Results}
\label{sec:v3-main-results}

For fixed $K,\eta$ and observation $B$, define the identification set
\begin{equation}
\label{eq:v3-identification-set}
\mathcal I_{K,\eta}(B)
:=\{A\subseteq G:A\neq\varnothing,\ |A+A|\le K|A|,\ |A\triangle B|\le\eta|A|\}.
\end{equation}
For a subspace $V\le G$, let $\pi_V:G\to G/V$ be the quotient map and $\mathcal N_V(A):=|\pi_V(A)|$ the number of $V$-cosets meeting $A$, and set
\begin{equation}
\label{eq:v3-cost}
q_A(V):=
\begin{cases}
\mathcal N_V(A),&|V|\le|A|,\\
+\infty,&|V|>|A|.
\end{cases}
\end{equation}
Thus $q_A(V)\le q$ implies
\begin{equation}
\label{eq:v3-dim-sandwich}
\log_2|A|-\log_2q\le\dim V\le\log_2|A|.
\end{equation}
Equivalently, a linear map $W$ with $\ker W=V$ maps $A$
to at most $q$ distinct values, since
$|W(A)|=\mathcal N_V(A)$.
A common-output guarantee requires one realized $V$ to satisfy
the cost bound for all $A\in\mathcal I_{K,\eta}(B)$ on the same
success event.

The objective has a direct interpretation for an affine support
$A=a+H$. In this case,
\[
\mathcal N_V(a+H)
=2^{\dim H-\dim(H\cap V)}.
\]
Thus $q_A(V)\le q$ requires $V$ to contain all but at most
$\log_2 q$ dimensions of the latent direction space $H$, while
$|V|\le|H|$ controls its total dimension. At cost one, these
conditions force $V=H$. Constant-cost recovery therefore
relaxes exact direction recovery by only a constant number
of dimensions. For general small-doubling supports, the coset
cover extends this objective to sets with several offsets
and incomplete occupancy within each coset.

For nonempty finite $S$, let $U_S$ be uniform on $S$. We compare
\[
\mathsf{Samp}(S)=\{\mathsf{Sample}_S\},
\qquad
\mathsf{SampMQ}(S)=\{\mathsf{Sample}_S,\mathsf{MQ}_S\},
\]
where sample calls are independent and exact, and $\mathsf{MQ}_S(x)=\mathbf1[x\in S]$. The learner accesses the fixed observed support $B$; the hidden $A$ is never directly queryable.

Membership access lets the learner test additive relations at
points constructed from its samples. For example, given
$u,v,w\in B$, it can determine whether $u+v+w\in B$; this test
is used by the weighted proposal mechanism in
Section~\ref{sec:v3-compiler}.
The answer concerns one fixed observed support, so repeated
queries at the same point always agree. The computational
resource is the ability to choose these test points.

Fix an integer-valued nondecreasing polynomial $p(K)$ dominating the clean PFR covering bound~\citep{gowers2025marton}. The query algorithm uses the covering bound $P_{\rm PFR}$ of the specific algorithmic routine in Appendix~\ref{app:size-oblivious-pfr}. To express statistical sharpness, write
\[
Q_{\rm com}(B;K,\eta):=\min_{V\le G}\max_{A\in\mathcal I_{K,\eta}(B)}q_A(V)
\]
when the identification set is nonempty. We use FPT for bounds of the form $f(\kappa)\operatorname{poly}(n,\log(1/\rho))$ with parameter-independent polynomial degree; XP permits that degree to depend on the structural parameters $\kappa$.

\begin{theorem}[Statistical common structural learning and sharp ambiguity cost]
\label{thm:v3-statistical}
Let $\varepsilon=1-\eta$ and assume $\mathcal I_{K,\eta}(B)\neq\varnothing$.
Given $K,\eta,\rho$ and access only to $\mathsf{Samp}(B)$, a randomized learner outputs a basis of $\widehat V$ such that
\begin{equation}
\label{eq:v3-statistical-main}
\Pr\!\left[
\forall A\in\mathcal I_{K,\eta}(B):
q_A(\widehat V)\le\frac{2Kp(K)}{\varepsilon}
\right]\ge1-\rho.
\end{equation}
A sufficient sample bound is
\begin{equation}
\label{eq:v3-vc-samples}
O\!\left(
\frac{(n+1)p(K)\log(2p(K))\log(2/\varepsilon)+\log(2/\rho)}{\varepsilon^2}
\right).
\end{equation}
An explicit exhaustive implementation is finite but exponential in $n$.

Moreover, for every integer $t\ge1$, there are $K=1$, $\eta_t=1-2^{-t}$, and an observation $B$ for which
\[
Q_{\rm com}(B;1,\eta_t)=(1-\eta_t)^{-1}.
\]
Thus the worst-case $\varepsilon^{-1}$ dependence of the common structural objective is sharp up to a constant factor already at doubling one.
\end{theorem}

\begin{theorem}[Conditional sample-only computational barrier]
\label{thm:v3-lpn}
Fix constants $Q\ge1$ and $0<\eta_0<1/3$.  Put
\(
k=1+\lfloor\log_2Q\rfloor,
\
\theta=\frac{\eta_0}{8k}.
\)

Assume search-LPN with a uniformly random secret and noise rate $\theta$ has no randomized polynomial-time algorithm with constant success probability.
Then there is no randomized polynomial-time learner which, for every valid $K=1$ instance $(A,B)$ with
$|A\triangle B|\le\eta_0|A|$, receives only independent samples from $U_B$ and with constant success probability outputs a polynomial-size list of explicitly represented subspaces containing some $V$ with $q_A(V)\le Q$.
The conclusion still holds if the learner is given $|A|$ or $\dim A$.  Since $\eta_0<1/3$, every valid observation in this theorem has at most one compatible doubling-one support.
\end{theorem}

\begin{theorem}[Query-assisted common structural learning]
\label{thm:v3-query-common}
There are absolute constants and computable nondecreasing functions with the following properties.

\emph{Moderate corruption.} If
$\mathcal I_{K,c_{\rm rob}/\sqrt K}(B)\neq\varnothing$, then from
$\mathsf{SampMQ}(B)$ a randomized learner outputs one $V$ such that, with probability at least $1-\rho$,
\[
\forall A\in\mathcal I_{K,c_{\rm rob}/\sqrt K}(B):
q_A(V)\le P_{\rm rob}(K),
\]
with sample, query, and time complexity
$F_{\rm rob}(K)\operatorname{poly}(n,\log(1/\rho))$.

\emph{Arbitrary corruption.}  Let $\varepsilon=1-\eta$ and assume
$\mathcal I_{K,\eta}(B)\neq\varnothing$.  From $\mathsf{SampMQ}(B)$ a randomized learner outputs one $V$ such that, with probability at least $1-\rho$,
\begin{equation}
\label{eq:v3-general-common-quality}
\forall A\in\mathcal I_{K,\eta}(B):
q_A(V)
\le
C K^{5/2}\varepsilon^{-5/2}
P_{\rm PFR}\!\left(C'K^4\varepsilon^{-4}\right).
\end{equation}
The number of uniform samples is polynomial in
$K,\varepsilon^{-1},n,\log(1/\rho)$, while membership-query complexity and ordinary running time are at most
\begin{equation}
\label{eq:v3-general-common-xp}
G(K,\varepsilon^{-1})
\bigl(1+n+\log(1/\rho)\bigr)^{O(\sqrt{K/\varepsilon})}.
\end{equation}
\end{theorem}

\paragraph{Comparison at a common quality budget.}
Fix $K=1$ and a constant $\eta_*\in(0,1/3)$. Choose a constant $Q_*$ dominating the cover bounds of Theorems~\ref{thm:v3-statistical} and~\ref{thm:v3-query-common}. At this same quality, samples suffice statistically; under search-LPN hardness at noise $\eta_*/(8(1+\lfloor\log_2Q_*\rfloor))$, Theorem~\ref{thm:v3-lpn} excludes polynomial-time sample-only learning. Membership access permits polynomial-time learning because the structural parameters are fixed. This comparison holds where the latent support is uniquely identified.

\section{Statistical Common Structure from Samples Alone}
\label{sec:v3-statistical-proof}

The selection rule has two parts: majority coverage controls the quotient cost, and minimum dimension enforces the size constraint. Let $p=p(K)$ and consider descriptions $(V,z_1,\ldots,z_p)$ with induced sets
\[
Z=\bigcup_{j=1}^{p}(z_j+V),
\qquad
\mathcal Z_p:=\{Z:V\le G,\ z_j\in G\}.
\]
At the population level, retain descriptions with $U_B(Z)\ge1/2$ and choose one minimizing $\dim V$. Every compatible $A$ has a clean PFR witness $(V_A,Z_A)$ with $A\subseteq Z_A$ and $|V_A|\le|A|$. By \eqref{eq:v3-majority}, $U_B(Z_A)\ge1/(1+\eta)>1/2$, so every witness is feasible. The selected $V$ consequently satisfies $|V|\le|A|$ for all compatible supports.

Now fix any retained description. Writing $D=A\setminus B$ and $O=B\setminus A$, its clean overlap obeys
\begin{equation}
\label{eq:v3-majority-capture-proof}
|A\cap B\cap Z|
\ge\frac{|B|}{2}-|O|
=\frac{|A|-|D|-|O|}{2}
\ge\frac{\varepsilon}{2}|A|.
\end{equation}
One of the $p$ cosets contains a slice $C\subseteq A$ of size at least $\varepsilon|A|/(2p)$. Since $C+C\subseteq V$ and $|A+C|\le K|A|$, Ruzsa covering gives
\begin{equation}
\label{eq:v3-stat-intro}
\mathcal N_V(A)\le\frac{|A+C|}{|C|}\le\frac{2Kp}{\varepsilon}.
\end{equation}
This argument applies to every compatible $A$ for the same selected description.

The sample learner retains descriptions with empirical mass at least $1/2+\varepsilon/8$. On the uniform event
\begin{equation}
\label{eq:v3-uniform-convergence}
\sup_{Z\in\mathcal Z_p}|\widehat P(Z)-U_B(Z)|\le\varepsilon/8,
\end{equation}
every clean witness passes because $1/(1+\eta)\ge1/2+\varepsilon/4$, and every survivor has population mass at least $1/2$. The preceding proof then applies unchanged. Affine subspaces have VC dimension at most $n+1$, and unions of $p$ affine subspaces have VC dimension $O((n+1)p\log(2p))$. Parallel unions form a subclass, giving the sample bound in \eqref{eq:v3-vc-samples}. Appendix~\ref{app:statistical-common} supplies the full argument.

Sharpness has a simple example. Take $G=\mathbb F_2^t$, $B=\{0\}$, and $\varepsilon=2^{-t}$. Both $A_0=\{0\}$ and $A_1=G$ are compatible at $K=1$. The size constraint for $A_0$ forces a common output to be $V=\{0\}$, whose cost on $A_1$ is $\varepsilon^{-1}$. Conversely this $V$ costs at most $|G|$ on any subset, proving equality. Appendix~\ref{app:ambiguity} also exhibits directional ambiguity among equally large latent supports.

Samples suffice, but the selection rule is computationally exhaustive. Can passive samples support an efficient learner?

\section{A Conditional Computational Barrier from LPN}
\label{sec:v3-lpn-proof}

Search-LPN asks to recover $s$ from examples $(x,\langle s,x\rangle\oplus e)$, with uniform $x$ and independent $e\sim\mathrm{Bernoulli}(\theta)$. The subspace $\mathbb F_2^d\times\{0\}$ covers any single
parity graph with two cosets, independently of $s$. We use several blocks so that every cost-$Q$ output reveals a short secret list, then couple noisy examples to one fixed support.

\paragraph{Even a coarse cover reveals a short secret list.}
Let $s\in\mathbb F_2^d$ be the LPN secret, put $\ell=\lfloor\log_2Q\rfloor$ and $k=\ell+1$, and define
\begin{equation}
\label{eq:v3-blocked-support}
A_s=\left\{(x_1,\ldots,x_k,
\langle s,x_1\rangle,\ldots,\langle s,x_k\rangle):x_i\in\mathbb F_2^d\right\}.
\end{equation}
Its dimension is $D=kd$. If $q_{A_s}(V)\le Q$, then
\[
\dim(A_s\cap V)\ge D-\ell,
\qquad D-\ell\le\dim V\le D.
\]
Hence $\dim V^\perp\le k+\ell$ and $\dim(A_s^\perp\cap V^\perp)\ge k-\ell=1$. Moreover,
\[
A_s^\perp=\{(a_1s,\ldots,a_ks,a_1,\ldots,a_k):a\in\mathbb F_2^k\}.
\]
Enumerating the at most $2^{k+\ell}\le2Q^2$ vectors of $V^\perp$ therefore gives at most $O(kQ^2)$ candidate secrets, including $s$.

\paragraph{Realizing the samples through one fixed support.}
For each complete tuple $x\in\mathbb F_2^D$, independently draw and freeze a $k$-bit noise vector $E(x)$ with Bernoulli$(\theta)$ coordinates. Let $B_E$ be the resulting noisy graph. At $\theta=\eta_0/(8k)$, a tuple changes with probability at most $\eta_0/8$, so
\[
\Pr_E[|A_s\triangle B_E|>\eta_0|A_s|]
\le\exp(-\Omega(\eta_0^2 2^D)).
\]
We draw the table without conditioning on the corruption
bound and include invalid tables in the failure probability. Couple fresh grouped LPN examples to the experiment that first draws $E$ and then samples uniformly from $B_E$. Until a complete input tuple repeats, the transcripts coincide.
For at most $T$ requests, the collision probability bounds their
total variation distance by $\binom{T}{2}2^{-D}$.
A polynomial-time structural learner would thus return a serving subspace with constant probability. Discard entries outside the prescribed dimension range, decode the remaining list as above, and validate the resulting secret candidates on fresh LPN samples. The true secret has mismatch rate $\theta$ and every false secret has rate $1/2$, yielding a polynomial-time search-LPN algorithm. Appendix~\ref{app:lpn-hardness} gives the amplification and error accounting.

The hard instances have a uniquely identified latent support, and the construction preserves hardness when its cardinality is supplied.

\section{Query-Assisted Common Learning}
\label{sec:v3-query-learning}

Membership queries let us replace exhaustive search by a short candidate family. Once the family contains a serving entry for every compatible support, fresh samples apply the same majority-coverage and minimum-dimension rule as in Section~\ref{sec:v3-statistical-proof}. The main task is to construct those candidates when the observation itself can have large doubling.

\paragraph{Extracting enough structured candidates.}
For a finite set $S$, its additive energy
\[
E(S):=|\{(a,b,c,d)\in S^4:a+b=c+d\}|
\]
counts repeated additive relations. Write $C_A=A\cap B$. If a residual $R\subseteq B$ retains at least half of $C_A$ for some compatible $A$, then
\[
E(R)\ge\frac{|A\cap R|^4}{|A+A|}
\ge\frac{|C_A|^4}{16K|A|}.
\]
Using $|R|\le|B|\le(1+\eta)|C_A|$ and $|C_A|\ge\varepsilon|A|$ gives
\begin{equation}
\label{eq:v3-residual-energy}
\frac{E(R)}{|R|^3}\ge\frac{\varepsilon}{16K(1+\eta)^3}
\ge\frac{\varepsilon}{128K}=: \alpha_\star.
\end{equation}
At this energy level, the persistent BSG procedure of Section~\ref{sec:v3-compiler} extracts a core of relative size
$\gamma_\star=c_{\rm BSG}\sqrt{\alpha_\star}=\Theta(\sqrt{\varepsilon/K})$.

We repeatedly remove such cores. The residual-density test permits continuation only while the residual has constant density in $B$. Thus each valid peel removes an $\Omega(\sqrt{\varepsilon/K})$ fraction of $B$, and $M_\star=O(\sqrt{K/\varepsilon})$ rounds suffice. On one joint good event, each normal stopping route---a small residual, compiler failure under its completeness guarantee, or reaching the peel cap---implies
\begin{equation}
\label{eq:v3-terminal-residual}
\forall A\in\mathcal I_{K,\eta}(B):
|A\cap R_{\rm fin}|<\tfrac12|A\cap B|.
\end{equation}
Rejection exhaustion and failed PFR calls are charged as algorithmic failures. Appendix~\ref{app:list-implementation} specifies the finite tests and budgets.

The removed cores are disjoint, so every compatible $A$ has a serving core with
\[
|A\cap Y_j|\ge\frac{\varepsilon}{2M_\star}|A|.
\]
\paragraph{From observed cores to latent covers.}
Each extracted core has doubling at most
$L_\star=C_{\rm BSG}\alpha_\star^{-4}$.
Apply the clean algorithmic PFR routine to obtain a
subspace $H_j$ with
$|H_j|\le|Y_j|$ and
$\mathcal N_{H_j}(Y_j)\le P_\star$, where
$P_\star=P_{\rm PFR}(L_\star)$.

Fix a compatible support $A$ and one of its serving
cores $Y_j$, and put $\delta=\varepsilon/(2M_\star)$.
Among the at most $P_\star$ occupied $H_j$-cosets,
one contains a slice $C\subseteq A\cap Y_j$ with
$|C|\ge\delta|A|/P_\star$.
Points of this slice differ only along $H_j$, so
$C+C\subseteq H_j$. Small doubling now transfers
the core cover to the whole latent support:
\[
|A+C|\le K|A|,
\qquad
\mathcal N_{H_j}(A)
\le\frac{|A+C|}{|C|}
\le\frac{KP_\star}{\delta},
\]
where the middle inequality follows from Ruzsa covering.
The clean slice is used only in the analysis; the
learner never needs to identify which core points
belong to $A$.

It remains to enforce the output-size constraint.
The available comparison is
\(
|H_j|\le|Y_j|\le|B|<2|A|,
\)
and it holds for every compatible support.
If $H_j\ne\{0\}$, remove one vector from its returned
basis to form an index-two subspace $V_j$;
otherwise set $V_j=H_j$.
This fixed rule guarantees $|V_j|\le|A|$ simultaneously
for all compatible $A$, without knowing their
cardinalities. Each $H_j$-coset splits into at most
two $V_j$-cosets, so every support served by $Y_j$
satisfies
$q_A(V_j)\le 2KP_\star/\delta
=4KM_\star P_\star/\varepsilon$.
Applying this construction to every extracted core
therefore gives subspaces satisfying
\begin{equation}
\label{eq:v3-list-intermediate}
m=O(\sqrt{K/\varepsilon}),
\qquad
\forall A\in\mathcal I_{K,\eta}(B)\ \exists i(A):
q_A(V_{i(A)})\le Q_{\rm list},
\end{equation}
where we choose the integer budget
\begin{equation}
\label{eq:v3-qlist}
Q_{\rm list}:=\left\lceil
C_{\rm list}K^{3/2}\varepsilon^{-3/2}
P_{\rm PFR}(C_{\rm list,L}K^4\varepsilon^{-4})
\right\rceil.
\end{equation}
Direct sample complexity is polynomial, but evaluating the
implicit sets has a recursive cost. Membership in
\[
R_j=B\setminus\bigcup_{i<j}Y_i
\]
requires evaluating earlier core predicates. Each core
predicate in turn calls the residual oracle from which
it was constructed. If $T_j$ bounds the cost of the
$j$th core predicate, the implementation admits a bound
of the form
$T_j\le a(1+\sum_{i<j}T_i)$, where $a$ is polynomial in
$n$ and $\log(1/\rho)$ for fixed structural parameters.
Over $M_\star=O(\sqrt{K/\varepsilon})$ rounds, this gives
the parameter-dependent exponent in
\eqref{eq:v3-general-common-xp}.
The final majority filter adds polynomial work.
Appendix~\ref{app:list-implementation} gives the resource
bounds.

\paragraph{Selecting one common output.}
The serving entry in \eqref{eq:v3-list-intermediate} may depend on $A$. The majority-coverage rule turns this family into one common hypothesis.

\begin{lemma}[Fresh-sample conversion of a universal list]
\label{lem:v3-list-common}
Let $Q\in\mathbb N_{\ge1}$. Suppose a fixed explicit list $V_1,\ldots,V_m$ satisfies
\[
\forall A\in\mathcal I_{K,\eta}(B)\ \exists i:
q_A(V_i)\le Q.
\]
Then a postprocessor using only fresh samples from $B$ and no new membership queries outputs one list entry $V_\star$ such that, with probability at least $1-\zeta$,
\[
\forall A\in\mathcal I_{K,\eta}(B):
q_A(V_\star)\le\frac{2KQ}{\varepsilon}.
\]
Its extra sample and time complexity is polynomial in
$m,Q,\varepsilon^{-1},\log(1/\zeta)$ and $n$; using VC uniform convergence, the sample count can be made independent of $n$ apart from the representation of the returned subspaces.
\end{lemma}

Condition on the past that fixes the candidate list.
For each $V_i$, group fresh samples by their $V_i$-coset and
let $Z_i$ be the union of the at most $Q$ most frequent
observed cosets. Retain indices with
$\widehat P(Z_i)\ge1/2+\varepsilon/8$ and return a retained
entry of minimum dimension, using fixed tie-breaking.
If none survives, return failure. Computing these scores
requires only coset identification and frequency counts.

Consider one event on which empirical mass is accurate to
$\varepsilon/8$ for every union of at most $Q$ cosets of
every listed subspace. Fix any compatible $A$ and its
serving entry $V_{i(A)}$. Some union $Z_A$ of at most $Q$
of its cosets contains $A$. Since the top-frequency union
maximizes empirical coverage for that entry,
\[
\widehat P(Z_{i(A)})
\ge \widehat P(Z_A)
\ge \frac{1}{1+\eta}-\frac{\varepsilon}{8}
\ge \frac12+\frac{\varepsilon}{8}.
\]
Every compatible support therefore has a surviving serving
entry. Minimum dimension gives, simultaneously,
\[
|V_\star|\le |V_{i(A)}|\le |A|
\qquad\text{for every }A\in\mathcal I_{K,\eta}(B).
\]
The selected union has true mass at least $1/2$, so
\eqref{eq:v3-majority-capture-proof} and the covering argument
of Section~\ref{sec:v3-statistical-proof} give the common
cost bound.

The concentration event ranges over coset unions of the
fixed candidates. For each $V_i$, this class has VC dimension
at most $Q$, so
\[
O\!\left(
\varepsilon^{-2}
\left[Q\log(2/\varepsilon)+\log(2m/\zeta)\right]
\right)
\]
fresh samples suffice. The same event handles every
compatible $A$. Applying the lemma with
\eqref{eq:v3-qlist} proves the general query theorem;
Appendix~\ref{app:list-to-common} gives the full accounting.

Under moderate corruption, one core already captures clean mass from every compatible support. Its retained fraction is $\Theta(K^{-1/2})$, larger by a constant factor than the allowed contamination. The size-based purity bound of Appendix~\ref{app:structural-lifting} then supplies a common lifting guarantee directly. With one core there is no residual recursion, yielding the FPT bound; Appendix~\ref{app:moderate-common} gives the shared success event and resource analysis.

\section{A Persistent Structured-Core Compiler}
\label{sec:v3-compiler}

After extracting a core, the query learner passes it to
the PFR routine, which needs membership queries and uniform
samples from that same set. The key requirement is a size and doubling guarantee for
the entire set defined by the returned record, valid
throughout all later adaptive calls. Peeling can also
reach residuals below the energy threshold, so the
extraction routine must control invalid outputs on those
inputs as well.

\begin{theorem}[Persistent structured-core compiler]
\label{thm:v3-compiler}
There are absolute constants $c_{\rm BSG},C_{\rm BSG}>0$ and a randomized algorithm which, given $\mathsf{SampMQ}(S)$ for a nonempty $S\subseteq\mathbb F_2^n$ and parameters $\alpha,\rho_{\rm snd},\rho_{\rm cmp}\in(0,1)$, always terminates and returns either \textsc{Fail} or a finite record $D$.  Every nonfailure $D$ defines one deterministic set $Y_D\subseteq S$.  Without any energy promise,
\[
\Pr\!\left[
D\neq\textsc{Fail}
\ \text{and}\
\bigl(|Y_D|<c_{\rm BSG}\sqrt\alpha|S|
\ \text{or}\
|Y_D+Y_D|>C_{\rm BSG}\alpha^{-4}|Y_D|\bigr)
\right]
\le\rho_{\rm snd}.
\]
If $E(S)\ge\alpha|S|^3$, then
$\Pr[D=\textsc{Fail}]\le\rho_{\rm cmp}$.
Once $D$ is fixed, membership in $Y_D$ is deterministic and has controlled worst-case query cost.  Moreover, for any fixed nonfailure record satisfying the retained-size guarantee $|Y_D|\ge c_{\rm BSG}\sqrt\alpha|S|$, any prescribed number $t$ of downstream samples, and failure budget $\zeta$, independent fixed-length rejection blocks succeed with probability at least $1-\zeta$; conditioned on their success, the returned sample vector has law exactly $U_{Y_D}^{\otimes t}$.
\end{theorem}

The construction separates finding a promising core from certifying the set that will be returned. In each of two multiplicity regimes, a counting argument assigns noticeable mass to useful candidates under an exactly sampleable weighted distribution.

For the high-multiplicity branch, put $r(d)=|S\cap(S+d)|$. Good fibres have mass at least $\Omega(\alpha)$ under $\mu_E(d)=r(d)^2/E(S)$. This law has a simple implementation: draw $U,V,W\sim U_S$, set $D=U+V$, and accept when $W+D\in S$. Then
\[
\Pr[D=d\text{ and accept}]=\frac{r(d)^2}{|S|^3},
\qquad
\Pr[D=d\mid\text{accept}]=\frac{r(d)^2}{E(S)}.
\]
Thus the proposal samples by squared multiplicity without estimating that multiplicity first.

For low multiplicities, a visible-layer argument gives
$O(1+\log(1/\sqrt\alpha))$ thresholds to search.
At a useful layer, sampling centres with probability
proportional to squared relation degree puts constant
mass on good candidates.

Shared sketches make the resulting predicates reliable
throughout their later use. For example, define the
translation density
$p(d)=\Pr_{X\sim U_S}[X+d\in S]$.
One block $X_1,\ldots,X_M\sim U_S$ defines
\[
\widehat p(d)
=\frac1M\sum_{j=1}^{M}\mathbf1_S(X_j+d).
\]
Hoeffding's inequality and a union bound over the $2^n$
translations give
\[
\Pr\!\left[
\sup_{d\in\mathbb F_2^n}|\widehat p(d)-p(d)|>\tau
\right]\le\delta
\quad\text{for}\quad
M=O\!\left(\tau^{-2}
\bigl(n+\log(1/\delta)\bigr)\right).
\]
The record stores the sample block and evaluates this
estimate only where needed. Its accuracy already holds at
every translation, including those chosen adaptively later.
The relation sketches and sampled codegree graph use such
simultaneous control and are fixed before candidate
validation.

Fresh samples test candidate density and average bad-pair rate. An independent cleaning sample turns the latter into a pointwise condition on every accepted element, yielding a deterministic small-doubling criterion. The stored sketches, candidate, and cleaning data define $Y_D$; the eventual membership predicate uses no fresh randomness. This gives soundness on arbitrary inputs. Under high energy, the weighted proposal bounds ensure that a good candidate is found with high probability.

For a valid record, the retained-size bound gives acceptance probability at least $c_{\rm BSG}\sqrt\alpha$ when rejecting samples from $S$. Independent blocks of length
\(
\left\lceil(c_{\rm BSG}\sqrt\alpha)^{-1}\log(t/\zeta)\right\rceil
\)
supply the exact product law conditional on all blocks succeeding. The same compiled set therefore supports both downstream access modes. Appendix~\ref{app:persistent-bsg} gives the construction and worst-case bounds. The comparison with earlier implicit BSG access procedures~\citep{tulsiani2014quadratic} is detailed in Appendix~\ref{app:related-work}.

\section{Exact Recovery at the Critical Radius}
\label{sec:v3-exact}

The common-cover objective also changes the cost of recovery when the latent support is unique. At doubling one, every support is affine, and exact reconstruction has the following boundary.

\begin{proposition}[Exact recovery at doubling one]
\label{prop:v3-exact-boundary}
Suppose $K=1$.
\begin{enumerate}
\item If $0\le\eta<1/3$ and $\mathcal I_{1,\eta}(B)\neq\varnothing$, then its unique affine support can be recovered from $\mathsf{SampMQ}(B)$ in
$\operatorname{poly}(n,(1-3\eta)^{-1},\log(1/\rho))$ time and oracle calls.
\item At $\eta=1/3$, every nonempty identification set is still a singleton, but any randomized classical algorithm that exactly recovers its member with probability at least $2/3$ on every instance requires $\Omega(2^n)$ total sample and membership-query calls in the worst case when the latent cardinality is not supplied.
\item For every $\eta>1/3$, there are observations with two distinct compatible affine supports, so exact recovery is not information-theoretically possible on every instance.
\end{enumerate}
In contrast, the coarse common structural output of Theorem~\ref{thm:v3-query-common} is polynomial-time computable at the fixed critical parameters $K=1$, $\eta=1/3$.
\end{proposition}

Below $1/3$, use the common learner at the fixed radius $1/3$ to obtain a constant-cost subspace. Fresh samples propose a containing affine flat $W$ of size $O(|A|)$. Containment can be tested because the coverage gap between containing and noncontaining flats is
\[
\frac1{1+\eta}-\frac1{2(1-\eta)}
=\frac{1-3\eta}{2(1-\eta^2)}>0.
\]
Inside a certified flat, write $p=|A|/|W|$, $f=\mathbf1_{B\cap W}$, and $\mu=\mathbb E_W f$. The threshold $\mu/2$ separates Fourier characters defining the affine support from all others by margin $(1-3\eta)p/2$. Goldreich--Levin search~\citep{goldreich1989hardcore} and coefficient estimation recover the affine equations.

At $\eta=1/3$, two observed supports differing at a single
point can have distinct unique latent explanations:
nested subspaces of sizes $a$ and $2a$.
Distinguishing the resulting oracles requires exponentially
many accesses in the worst case. Appendix~\ref{app:exact-k1} proves these claims and studies how candidate lists and knowledge of latent cardinality change the exact-recovery boundary.

\paragraph{Discussion.}
\label{sec:v3-discussion}

Even when the latent support is uniquely identified,
efficient recovery can depend on how the learner accesses
the observation. The statistical and query learners use the same
majority-coverage rule. The statistical learner searches
exhaustively over coset covers; the query learner constructs
a short family to which the rule can be applied.
The specified search-LPN assumption rules out an efficient
passive learner at the matched covering budget.

The next questions are whether the general query learner
can match the statistical covering bound and improve its
XP dependence to FPT. Recursive residual predicates cause
the current exponent, making their evaluation cost a
concrete target. Appendix~\ref{app:further-directions}
discusses these questions and extensions beyond exact
uniform samples and exact membership access.

\clearpage
\subsection*{AI use statement}

LLM-based tools were used during the development and preparation of this manuscript for mathematical discussion and checking, literature discovery, and writing assistance. In the technical development, they were used to discuss formulations and possible refinements of bounds and proof arguments, and to check algebraic manipulations, quantifier order, parameter dependence, failure-probability accounting, and assumptions in compositions of randomized algorithms. They were also used to improve proof exposition and theorem statements. All suggestions were independently assessed by the author, and all final mathematical claims and proofs were manually checked; Lean was additionally used during verification.

LLMs were also used to identify and compare potentially relevant literature and to assist with organization, drafting, wording, LaTeX formatting, and notation consistency. Retained references were checked against the original sources. The work contains no synthetic datasets or empirical experiments. The author takes responsibility for all mathematical statements, proofs, citations, and AI-assisted text in the final manuscript.

\newpage

\begin{center}
    {\large\bfseries Appendix Contents}
\end{center}

\vspace{0.4em}
\noindent
The detailed proofs and constructions underlying
Theorems~\ref{thm:v3-statistical},
\ref{thm:v3-lpn}, and~\ref{thm:v3-query-common}
appear in Appendices~\ref{app:statistical-common},
\ref{app:lpn-hardness}, and~\ref{app:list-implementation}--%
\ref{app:moderate-common}, respectively.

\medskip
{\small
\renewcommand{\arraystretch}{1.18}
\begin{tabular}{
    @{}
    >{\raggedright\arraybackslash}p{0.34\linewidth}
    >{\raggedright\arraybackslash}p{0.56\linewidth}
    r
    @{}
}
\toprule
\textbf{Appendix} & \textbf{Purpose} & \textbf{Page} \\
\midrule

\hyperref[app:probability-tools]{\textbf{A. Probability and Exact-Sampling Tools}}
&
Conditional concentration, finite-family testing, adaptive failure accounting,
and exact finite rejection-sampling primitives used throughout the proofs.
&
\pageref{app:probability-tools}
\\

\hyperref[app:persistent-bsg]{\textbf{B. Persistent BSG Compiler}}
&
Full deterministic and randomized construction of the persistent BSG core,
including searchable sifting, validation, exact sampling, and finite resource
bounds.
&
\pageref{app:persistent-bsg}
\\

\hyperref[app:size-oblivious-pfr]{\textbf{C. Size-Oblivious Algorithmic PFR}}
&
Algorithmic PFR for an implicit structured set of unknown cardinality:
localization, dense modeling, restricted-homomorphism search, validation,
covering, and resource accounting.
&
\pageref{app:size-oblivious-pfr}
\\

\hyperref[app:structural-lifting]{\textbf{D. Structural Lifting from an Observed Core}}
&
Lifting a core cover to a latent-support cover when their
overlap is sufficiently large, with subspace-size control.
&
\pageref{app:structural-lifting}
\\

\hyperref[app:list-implementation]{\textbf{E. Intermediate List Construction}}
&
Arbitrary-corruption learner based on persistent-core peeling, PFR on each
core, residual certificates, finite residual access, and terminal invariants.
&
\pageref{app:list-implementation}
\\

\hyperref[app:list-to-common]{\textbf{F. From a List to One Common Hypothesis}}
&
Fresh-sample selection and majority-coverage arguments that convert the
intermediate candidate list into one hypothesis valid for all compatible
latent supports.
&
\pageref{app:list-to-common}
\\

\hyperref[app:moderate-common]{\textbf{G. Common Learning under Moderate Corruption}}
&
Single-core analysis in the moderate-corruption regime and the proof that one
persistent core contains enough clean mass for simultaneous structural
lifting.
&
\pageref{app:moderate-common}
\\

\hyperref[app:statistical-common]{\textbf{H. Sample-Only Common Learning}}
&
Sample-only learning by majority coverage and minimum-dimension
selection, with uniform-convergence sample bounds.
&
\pageref{app:statistical-common}
\\

\hyperref[app:lpn-hardness]{\textbf{I. Conditional Sample-Only Hardness}}
&
Search-LPN reduction for constant-cost coarse recovery,
allowing polynomial-size output lists and known latent cardinality.
&
\pageref{app:lpn-hardness}
\\

\hyperref[app:exact-k1]{\textbf{J. Exact Affine-Support Recovery at \(K=1\)}}
&
Exact recovery below $\eta=1/3$, exponential access cost at
the boundary with unknown latent cardinality, and refinements
using candidate lists or known cardinality.
&
\pageref{app:exact-k1}
\\

\hyperref[app:ambiguity]{\textbf{K. Ambiguity Beyond Latent Cardinality}}
&
Sharp common-cover costs for a prescribed family of equally
large latent subspaces with independent unseen directions.
&
\pageref{app:ambiguity}
\\

\hyperref[app:one-core-barrier]{\textbf{L. One-Core Size-Only Mechanism Barrier}}
&
A limitation of the single-core size-based certification mechanism, explaining
the intrinsic \(K^{-1/2}\) scale of that proof architecture.
&
\pageref{app:one-core-barrier}
\\

\hyperref[app:related-work]{\textbf{M. Detailed Relation to Prior Work}}
&
Extended comparisons with latent identifiability, robust and list-decodable
learning, persistent-noise query models, algorithmic additive combinatorics,
and exact affine recovery.
&
\pageref{app:related-work}
\\

\hyperref[app:further-directions]{\textbf{N. Algorithmic and Access-Model Questions}}
&
Open algorithmic questions on matching statistical covering quality,
representing recursive residuals more efficiently, and weakening the required
sample/membership access.
&
\pageref{app:further-directions}
\\

\bottomrule
\end{tabular}
}

\newpage
\appendix

\section{Probability and Exact-Sampling Tools}
\label{app:probability-tools}
\label{app:probability-tools-v2}

The randomized constructions reveal information sequentially.  At each
stage, the object being tested may depend on the full past transcript, but
the block used to test it is fresh.  We therefore state all probability
tools conditionally.  Throughout, \(\mathcal F\) denotes the sigma-field
generated by the past; quantities described as fixed may be
\(\mathcal F\)-measurable.

\paragraph{Arithmetic convention for public parameters.}
Public scalar parameters such as \(\alpha\), \(\rho\), and the deterministic
thresholds derived from them are available to the algorithm as exact public
parameters, and arithmetic/comparisons involving only these public scalars are
counted at unit cost.  The displayed ordinary running-time bounds are stated
under this exact-public-parameter arithmetic convention; we do not claim a
separate bit-complexity bound for arbitrary real encodings of the public
parameters.  Data-dependent quantities are stored discretely: empirical
probabilities are represented by their integer counts and sample sizes, and
derived thresholds are reconstructed symbolically from those counts and the
public parameters.  In particular, the record-length bounds below count only
data-dependent information stored in the record, not a second encoding of the
public parameters.

\begin{lemma}[Conditional concentration and finite-family control]
\label{lem:conditional-concentration-v25}
\label{lem:conditional-concentration-v2}
\label{lem:uniform-finite-family-v2}
Conditional on \(\mathcal F\), let \(X_1,\ldots,X_m\) be independent
Bernoulli variables with common \(\mathcal F\)-measurable mean \(p\), and
put \(X=\sum_iX_i\).  Then, almost surely,
\begin{align}
\Pr\!\left[X\le(1-\delta)mp\mid\mathcal F\right]
&\le \exp\!\left(-\frac{\delta^2mp}{2}\right)
&& (0<\delta<1),
\label{eq:app-conditional-lower-chernoff-v2}\\
\Pr\!\left[X\ge(1+\delta)mp\mid\mathcal F\right]
&\le \exp\!\left(-\frac{\delta^2mp}{2+\delta}\right)
&& (\delta>0),
\label{eq:app-conditional-upper-chernoff-v2}\\
\Pr\!\left[\left|\frac Xm-p\right|\ge\varepsilon\,\middle|\,\mathcal F\right]
&\le 2e^{-2m\varepsilon^2}.
\label{eq:app-conditional-hoeffding-v2}
\end{align}
Moreover, if \(\mathcal H\) is an \(\mathcal F\)-measurable family of at
most \(M\) functions \(f:\Omega\to[0,1]\), and
\(Z_1,\ldots,Z_m\) are fresh independent samples from a fixed conditional
distribution \(\mu\), then
\begin{equation}
\Pr\!\left[
\sup_{f\in\mathcal H}
\left|\frac1m\sum_{i=1}^m f(Z_i)-\mathbb E_\mu f\right|>\varepsilon
\,\middle|\,\mathcal F
\right]
\le 2M e^{-2m\varepsilon^2}.
\label{eq:app-uniform-finite-family-v2}
\end{equation}
Hence simultaneous error at most \(\varepsilon\) holds with conditional
probability at least \(1-\rho\) whenever
\begin{equation}
m\ge \frac{1}{2\varepsilon^2}\log\frac{2M}{\rho}.
\label{eq:app-uniform-family-size-v2}
\end{equation}
\end{lemma}

\begin{proof}
Condition on \(\mathcal F\).  The ordinary Chernoff and Hoeffding bounds
apply to the fresh block.  The finite-family statement follows by applying
Hoeffding to each \(f\in\mathcal H\) and taking a union bound.  Once the
sample block is stored, every empirical value is a deterministic function
of that same block; this is the basic persistence mechanism used below.
\end{proof}

\begin{lemma}[Sequential conditional failure accounting]
\label{lem:sequential-conditional-failure-v2}
\label{cor:adaptive-slot-accounting-v2}
Let
\(\mathcal F_0\subseteq\cdots\subseteq\mathcal F_R\) be a filtration.
For each \(j\), let \(\mathcal B_j\in\mathcal F_j\) satisfy
\[
\Pr[\mathcal B_j\mid\mathcal F_{j-1}]\le\rho_j
\quad\text{almost surely}.
\]
Then
\begin{equation}
\Pr\!\left[\bigcup_{j=1}^R\mathcal B_j\right]
\le\sum_{j=1}^R\rho_j.
\label{eq:app-sequential-failure-sum-v2}
\end{equation}
In particular, an algorithm with at most \(R\) adaptively reached test
slots, each using a fresh block and having conditional error at most
\(\rho\), incurs total error at most \(R\rho\).
\end{lemma}

\begin{proof}
Let \(\mathcal H_j=\mathcal B_j\cap\bigcap_{i<j}\mathcal B_i^c\) be the
event that the first failure occurs at stage \(j\).  The \(\mathcal H_j\)
are disjoint and the tower property gives
\[
\Pr[\mathcal H_j]
=\mathbb E\!\left[
\mathbf 1_{\cap_{i<j}\mathcal B_i^c}
\Pr(\mathcal B_j\mid\mathcal F_{j-1})
\right]
\le\rho_j.
\]
Summing proves the claim.  Unreached slots may be padded with the empty bad
event.
\end{proof}

\begin{lemma}[Finite gap tests]
\label{lem:finite-gap-tests-v25}
\label{lem:multiplicative-gap-test-v2}
\label{lem:additive-gap-test-v2}
\label{cor:persistent-bsg-gap-constants-v2}
The following two conditional tests are valid after any adaptively generated
past.

\begin{enumerate}
\item[(a)] \emph{Multiplicative gap.}
Fix \(0<a<\theta<b\), let \(q\in(0,1]\) with \(bq\le1\), and accept a
Bernoulli mean \(p\) when the empirical mean from \(m\) fresh samples is at
least \(\theta q\).  Then
\begin{align}
p\ge bq
&\Longrightarrow
\Pr[\mathrm{reject}\mid\mathcal F]
\le \exp\!\left(-\frac{(b-\theta)^2}{2b}mq\right),
\label{eq:app-multiplicative-gap-completeness-v2}\\
p\le aq
&\Longrightarrow
\Pr[\mathrm{accept}\mid\mathcal F]
\le \exp\!\left(-\frac{(\theta-a)^2}{a+\theta}mq\right).
\label{eq:app-multiplicative-gap-soundness-v2}
\end{align}
Thus both errors are at most \(\rho\) if
\begin{equation}
m\ge C(a,\theta,b)q^{-1}\log\frac{2}{\rho},
\quad
C(a,\theta,b):=
\max\left\{\frac{2b}{(b-\theta)^2},
\frac{a+\theta}{(\theta-a)^2}\right\}.
\label{eq:app-multiplicative-gap-size-v2}
\end{equation}

\item[(b)] \emph{Additive gap.}
Fix \(0\le a<\theta<b\le1\), and accept when the empirical mean is at
least \(\theta\).  Then
\begin{align}
p\ge b
&\Longrightarrow
\Pr[\mathrm{reject}\mid\mathcal F]\le e^{-2m(b-\theta)^2},
\label{eq:app-additive-gap-completeness-v2}\\
p\le a
&\Longrightarrow
\Pr[\mathrm{accept}\mid\mathcal F]\le e^{-2m(\theta-a)^2}.
\label{eq:app-additive-gap-soundness-v2}
\end{align}
It suffices that
\begin{equation}
m\ge
\frac{1}{2\min\{(b-\theta)^2,(\theta-a)^2\}}
\log\frac{2}{\rho}.
\label{eq:app-additive-gap-size-v2}
\end{equation}
\end{enumerate}
For the persistent BSG validation below it is enough to take
\begin{align}
m_{\rm deg,H}
&=\left\lceil128h^{-1}\log\frac{2}{\zeta}\right\rceil,
&
m_{\rm deg,L}(\delta_0)
&=\left\lceil256\delta_0^{-1}\log\frac{2}{\zeta}\right\rceil,
\label{eq:app-bsg-degree-constants-v2}\\
m_{\rm pair}
&=\left\lceil512\log\frac{2}{\zeta}\right\rceil.
\label{eq:app-bsg-pair-constant-v2}
\end{align}
\end{lemma}

\begin{proof}
Part (a) is the lower- and upper-tail Chernoff bound applied at the two gap
endpoints; part (b) is Hoeffding.  For the displayed BSG constants use
\((a,\theta,b)=(3/4,7/8,1)\) in the high degree test,
\((3/8,7/16,1/2)\) in the low degree test, and apply the additive test to
\(1-\mathsf{Bad}\) with gaps of size \(1/32\).
\end{proof}

\begin{lemma}[Fixed-block exact rejection sampling]
\label{lem:fixed-block-exact-sampling-v25}
\label{lem:fixed-block-first-success-v2}
\label{cor:adaptive-exact-rejection-v2}
Conditional on \(\mathcal F\), suppose a trial takes values in
\(\Omega\cup\{\bot\}\) and satisfies
\[
\Pr[V\in A\mid\mathcal F]=p\,\mu(A)
\]
for an \(\mathcal F\)-measurable acceptance probability \(p\) and
conditional law \(\mu\).  In a block of \(T\) independent trials, return
the first non-\(\bot\) value, or \(\mathsf{FAIL}\) if all trials reject.
Then
\begin{equation}
\Pr[\mathsf{FAIL}\mid\mathcal F]=(1-p)^T\le e^{-pT},
\label{eq:app-fixed-block-failure-v2}
\end{equation}
and conditioned on success the returned value has law \(\mu\).
Independent blocks have the product of their conditional output laws when
conditioned on their individual successes.

In particular, if an adaptively fixed \(Y\subseteq S\) satisfies
\(|Y|\ge\beta|S|\), then \(t\) independent rejection blocks from \(U_S\)
of length
\begin{equation}
T=\left\lceil\beta^{-1}\log\frac{t}{\rho}\right\rceil
\label{eq:app-adaptive-rejection-length-v2}
\end{equation}
all succeed with conditional probability at least \(1-\rho\), and,
conditioned on success, their outputs have law
\begin{equation}
U_Y^{\otimes t}.
\label{eq:app-adaptive-rejection-product-v2}
\end{equation}
\end{lemma}

\begin{proof}
The probability that the first accepted value lies in \(A\) is
\[
p\mu(A)\sum_{j=1}^T(1-p)^{j-1};
\]
the geometric factor is independent of \(A\), so conditioning on success
leaves the law \(\mu\).  For set rejection, one trial accepts with
probability \(|Y|/|S|\ge\beta\); the chosen \(T\) makes each block fail
with probability at most \(\rho/t\).  A union bound and conditional
independence give the product statement.
\end{proof}

\section{Persistent BSG Compiler: Full Construction and Proof}
\label{app:persistent-bsg}

The structural argument follows the high-/low-multiplicity
strategy of \citet{reiherschoen2024bsg}.
We formulate weighted certificates that can be sampled
through the oracle and introduce margins for the sampled
tests. We then construct a finite record defining one
fixed subset and provide exact sampling from that set,
proving Theorem~\ref{thm:v3-compiler}.

\subsection{Deterministic searchable sifting}
\label{sec:deterministic-sifting-v2}
\label{subsec:sifting-criteria-v2}
\label{subsec:high-branch-v2}
\label{subsec:low-branch-v2}
\label{subsec:weighted-sifting-theorem-v2}

Let \(S\subseteq\mathbb F_2^n\) be nonempty, write \(N:=|S|\), and define
\[
r(d):=|S\cap(S+d)|,
\qquad
p(d):=\frac{r(d)}N.
\]
Then
\begin{equation}
\sum_dp(d)=N,
\qquad
\sum_dp(d)^2=\frac{E(S)}{N^2}.
\label{eq:sifting-moments}
\end{equation}
Fix \(\alpha\in(0,1]\) and put \(h:=\sqrt\alpha\).  For every nonempty
\(Y\subseteq S\) and \(s\in Y+Y\), fix the lexicographically first ordered
representation \(s=y_1(s)+y_2(s)\).

\begin{lemma}[Two structural criteria]
\label{lem:structural-criteria-v25}
\label{lem:high-criterion-v2}
\label{lem:low-criterion-v2}
Each of the following deterministic conditions forces a large
small-doubling subset.

\begin{enumerate}
\item[(H)] Let \(Q\subseteq\mathbb F_2^n\) and \(Y\subseteq X\subseteq S\)
satisfy
\begin{equation}
|Q|\le2^{13}h^{-2}N,
\qquad
|X|\ge\frac34hN,
\qquad
|Y|\ge\frac12|X|,
\label{eq:high-criterion-size-data}
\end{equation}
and suppose every \(y\in Y\) has at most \(3|X|/8\) partners
\(w\in X\) with \(y+w\notin Q\).  Then
\[
|Y|\ge\frac38hN,
\qquad
|Y+Y|\le\frac{2^{33}}9h^{-6}|Y|.
\]

\item[(L)] Let \(Q\subseteq\mathbb F_2^n\), and suppose a symmetric graph
\(\mathcal G\subseteq S^2\), numbers \(\delta_0,\tau,\kappa>0\), and
sets \(Y\subseteq X\subseteq S\) satisfy
\begin{align}
\delta_0&\ge\frac78h,
&
\delta_0\tau^2&\ge2^{-16}h^2,
&
|Q|&\le\frac{2\delta_0N}{h\tau},
\label{eq:low-criterion-layer-data}\\
|X|&\ge\frac{3\delta_0N}{8},
&
|Y|&\ge\frac12|X|,
&
\kappa&\ge2^{-12}\delta_0^2.
\label{eq:low-criterion-local-data}
\end{align}
Assume every graph edge has \(Q\)-codegree at least \(\kappa N\), and
every \(y\in Y\) has at most \(3|X|/8\) non-neighbours in \(X\).  Then
\[
|Y|\ge\frac{21}{128}hN,
\qquad
|Y+Y|\le\frac{2^{69}}9h^{-8}|Y|.
\]
\end{enumerate}
\end{lemma}

\begin{proof}
For (H), fix \(s\in Y+Y\) with canonical endpoints \(y_1,y_2\).  At
least \(|X|/4\) points \(w\in X\) satisfy both \(y_1+w\in Q\) and
\(y_2+w\in Q\).  The map
\((s,w)\mapsto(y_1+w,w+y_2)\in Q^2\) is injective, hence
\[
|Y+Y|\frac{|X|}{4}\le|Q|^2.
\]
Substitution gives the stated doubling bound and the size lower bound is
immediate.

For (L), the same endpoint argument leaves at least \(|X|/4\) common graph
neighbours \(w\).  For each such \(w\), each endpoint--\(w\) graph edge
has at least \(\kappa N\) witnesses in \(S\).  The resulting map
\((s,w,z_1,z_2)\mapsto(q_1,q_2,q_3,q_4)\in Q^4\) is injective, so
\begin{equation}
|Y+Y|\frac{|X|}{4}(\kappa N)^2\le|Q|^4.
\label{eq:low-four-path-count-v2}
\end{equation}
Using \eqref{eq:low-criterion-layer-data}--\eqref{eq:low-criterion-local-data}
yields the stated \(h^{-8}\) bound; the size estimate follows directly.
\end{proof}

For \(Q\subseteq\mathbb F_2^n\), let
\[
R_Q:=\{(x,y)\in S^2:x+y\in Q\},
\qquad
\delta(Q):=\frac{|R_Q|}{N^2},
\]
and define the \(Q\)-codegree
\[
c_Q(x,y):=|\{z\in S:x+z\in Q,\ y+z\in Q\}|.
\]
Set
\[
\mathcal H_h:=\{d:p(d)\ge h\},
\qquad
\mathcal L_h:=\{d:p(d)<h\}.
\]
If \(E(S)\ge h^2N^3\), then
\begin{equation}
\sum_{d\in\mathcal H_h}p(d)^2\ge\frac{h^2N}{32}
\quad\text{or}\quad
\sum_{d\in\mathcal L_h}p(d)^2\ge\frac{31h^2N}{32}.
\label{eq:sifting-energy-dichotomy-v2}
\end{equation}

\begin{lemma}[Visible low-multiplicity layer]
\label{lem:visible-low-layer-v25}
\label{lem:visible-layer-v2}
\label{lem:visible-low-layer-v2}
Assume the second alternative in
\eqref{eq:sifting-energy-dichotomy-v2}.  Define
\[
\mathcal T_h:=
\left\{2^{-j}:0\le j\le\left\lceil\log_2\frac{2^8}{h}\right\rceil\right\}.
\]
Then some \(\tau\in\mathcal T_h\) has the following property.  For
\[
Q_\tau:=\{d\in\mathcal L_h:p(d)\ge h\tau\},
\qquad
\delta_\tau:=\delta(Q_\tau),
\]
one has
\begin{align}
\delta_\tau&\ge\frac{15}{16}h,
&
\delta_\tau\tau^2&\ge\frac{h^2}{2^{15}},
&
h\tau&\ge\frac{h^2}{2^8},
&
|Q_\tau|&\le\frac{\delta_\tau N}{h\tau}.
\label{eq:visible-low-layer-data-v2}
\end{align}
Every \(d\in Q_\tau\) also satisfies \(h\tau\le p(d)<h\).
\end{lemma}

\begin{proof}
For \(d\in\mathcal L_h\), put \(x_d=p(d)/h\), and for \(t\in(0,1]\)
write \(M(t)=\sum_{x_d\ge t}x_d\).  Then
\[
\sum_dx_d^2\ge\frac{31}{32}N,
\qquad
\sum_dx_d\le\frac Nh,
\qquad
\int_0^1M(t)\,dt=\sum_dx_d^2.
\]
If no \(t\) simultaneously satisfied
\(M(t)\ge15N/16\) and \(M(t)t^2\ge hN/2^{13}\), then
\[
M(t)\le\frac{15}{16}N+
\min\left\{\frac Nh,\frac{hN/2^{13}}{t^2}\right\}.
\]
Integrating and using
\(\int_0^1\min\{A,Z/t^2\}\,dt\le2\sqrt{AZ}\) contradicts the
\(31N/32\) second-moment lower bound.  Choose dyadic \(\tau\) with
\(t/2<\tau\le t\).  The quadratic bound loses at most a factor four.
At this chosen \(t\), the inequality \(M(t)t^2\ge hN/2^{13}\), together
with \(M(t)\le N/h\), gives
\[
 t^2\ge \frac{h^2}{2^{13}},
 \qquad
 \tau>\frac t2\ge \frac{h}{2^{15/2}}>\frac{h}{2^8}.
\]
Thus the finite dyadic grid reaches the needed scale and
\(h\tau>h^2/2^8\).  Since
\[
\delta_\tau=\frac hN M(\tau),
\]
the first two displayed bounds follow, while
\(|Q_\tau|h\tau\le\delta_\tau N\) gives the cardinality estimate.
\end{proof}

\begin{lemma}[Searchable weighted certificates]
\label{lem:searchable-weighted-certificates-v25}
\label{lem:weak-pair-identity-v2}
\label{prop:weighted-high-branch-v2}
\label{lem:weighted-relation-sifting-v2}
\label{prop:weighted-low-branch-v2}
Under the corresponding alternative in
\eqref{eq:sifting-energy-dichotomy-v2}, the structural criteria above have
noticeable mass under an exactly sampleable weighted law.

\begin{enumerate}
\item[(H)] Set
\[
\vartheta_{\rm H}:=\frac{h^2}{2^{12}},
\qquad
Q_{\rm H}:=\{q:p(q)\ge\vartheta_{\rm H}\},
\qquad
\mu_E(d):=\frac{r(d)^2}{E(S)}.
\]
Then
\begin{equation}
|Q_{\rm H}|\le2^{12}h^{-2}N,
\qquad
\sum_{q\notin Q_{\rm H}}r(q)^2\le\vartheta_{\rm H}N^3.
\label{eq:high-popular-estimates-v2}
\end{equation}
If the first alternative in \eqref{eq:sifting-energy-dichotomy-v2} holds,
then \(\mu_E\) assigns probability at least \(7h^2/2^8\) to differences
\(d\) for which \(|X_d|\ge hN\) and the average fraction of pairs
\((x,y)\in X_d^2\) with \(x+y\notin Q_{\rm H}\) is at most \(1/16\).
Cleaning vertices of weak degree above \(|X_d|/4\) leaves at least
\(3|X_d|/4\) points and satisfies criterion (H).

\item[(L)] Let \(\tau,Q_\tau,\delta_\tau\) be supplied by
Lemma~\ref{lem:visible-low-layer-v25}, put \(R_\tau=R_{Q_\tau}\), and let
\[
\Gamma_{R_\tau}(a):=\{x\in S:(a,x)\in R_\tau\},
\quad
d(a):=|\Gamma_{R_\tau}(a)|,
\quad
Z:=\sum_ad(a)^2.
\]
For
\[
\mu_{R_\tau}(a):=\frac{d(a)^2}{Z},
\qquad
\kappa_\tau:=\frac{\delta_\tau^2}{2^{10}},
\]
and any fixed symmetric graph \(\mathcal G_\tau\) satisfying
\begin{equation}
c_{Q_\tau}(x,y)\ge2\kappa_\tau N
\Longrightarrow (x,y)\in\mathcal G_\tau
\Longrightarrow c_{Q_\tau}(x,y)\ge\kappa_\tau N,
\label{eq:low-graph-sandwich-v2}
\end{equation}
the law \(\mu_{R_\tau}\) assigns probability at least \(15/32\) to
centres \(a\) with
\[
d(a)\ge\frac{\delta_\tau N}{2}
\]
and at most a \(1/16\) fraction of graph nonedges inside
\(\Gamma_{R_\tau}(a)^2\).  Cleaning high-nonedge-degree vertices leaves
at least three quarters of the neighbourhood and satisfies criterion (L).
\end{enumerate}
\end{lemma}

\begin{proof}
For (H), define
\[
B_Q(d):=|\{(x,y)\in X_d^2:x+y\notin Q\}|.
\]
Double counting ordered pairs of representations gives the identity
\[
\sum_dB_Q(d)=\sum_{q\notin Q}r(q)^2.
\]
The high-energy alternative places at least \(h^2N^3/32\) squared-
representation mass on \(\{d:p(d)\ge h\}\).  Differences with
\(B_{Q_{\rm H}}(d)>r(d)^2/16\) carry less than
\(16\sum_dB_{Q_{\rm H}}(d)\le h^2N^3/2^8\).  Thus at least
\(7h^2N^3/2^8\) good mass remains; because \(E(S)\le N^3\), its
\(\mu_E\)-probability is at least \(7h^2/2^8\).  Markov's inequality on
weak degrees gives the cleaning claim.

For (L), Cauchy--Schwarz gives \(Z\ge\delta_\tau^2N^3\).  Centres with
degree below \(\delta_\tau N/2\) carry at most half of \(Z\).  If
\(B_{\mathcal G}(a)\) denotes the ordered number of graph nonedges inside
\(\Gamma(a)^2\), then
\[
\sum_aB_{\mathcal G}(a)
=\sum_{(x,y)\notin\mathcal G}c_{Q_\tau}(x,y)
<2\kappa_\tau N^3.
\]
Hence centres with
\(B_{\mathcal G}(a)>d(a)^2/16\) carry less than
\(32\kappa_\tau N^3=\delta_\tau^2N^3/32\le Z/32\).  Removing the two bad
classes leaves \(\mu_{R_\tau}\)-mass at least
\(1-1/2-1/32=15/32\), and Markov gives the cleaning claim.
\end{proof}

\begin{theorem}[Weighted deterministic sifting]
\label{thm:weighted-deterministic-sifting-v2}
Let \(S\subseteq\mathbb F_2^n\) be nonempty and suppose
\(E(S)\ge\alpha|S|^3\).  Then at least one of the following holds.
\begin{enumerate}
\item[(H)] Under \(\mu_E\), a set of differences of probability at least
\(7\alpha/2^8\) yields, after cleaning, a set \(Y\subseteq S\) with
\[
|Y|\ge\frac38\sqrt\alpha\,|S|,
\qquad
|Y+Y|\le\frac{2^{33}}9\alpha^{-3}|Y|.
\]
\item[(L)] Some \(\tau\in\mathcal T_h\) satisfies
\eqref{eq:visible-low-layer-data-v2}; for every graph obeying
\eqref{eq:low-graph-sandwich-v2}, the squared-degree law assigns mass at
least \(15/32\) to centres whose cleaning yields \(Y\subseteq S\) with
\[
|Y|\ge\frac{21}{128}\sqrt\alpha\,|S|,
\qquad
|Y+Y|\le\frac{2^{69}}9\alpha^{-4}|Y|.
\]
\end{enumerate}
In particular one branch contains a set \(Y\) with
\[
|Y|\ge\frac{21}{128}\sqrt\alpha\,|S|,
\qquad
|Y+Y|\le\frac{2^{69}}9\alpha^{-4}|Y|.
\]
\end{theorem}

\begin{proof}
Apply the dichotomy \eqref{eq:sifting-energy-dichotomy-v2}.  In the high
case use Lemma~\ref{lem:searchable-weighted-certificates-v25}(H) and
Lemma~\ref{lem:structural-criteria-v25}(H).  In the low case use
Lemmas~\ref{lem:visible-low-layer-v25},
\ref{lem:searchable-weighted-certificates-v25}(L), and
\ref{lem:structural-criteria-v25}(L).  Since \(\alpha\le1\), the high-
branch \(\alpha^{-3}\) bound is also dominated by the displayed
\(\alpha^{-4}\) guarantee.
\end{proof}

The weighted bounds make the structural alternatives searchable: each
places useful candidates under a distribution that the oracle algorithm can
sample exactly.

\subsection{Persistent randomized BSG: construction and proof}
\label{sec:persistent-bsg-v2}

Let \(S\subseteq\mathbb F_2^n\) be nonempty, write \(N:=|S|\), and
fix \(\alpha\in(0,1]\).  As in the preceding section, put
\(h:=\sqrt\alpha\).

The randomized algorithm compiles the criteria of
Section~\ref{sec:deterministic-sifting-v2} into one reusable predicate.
All data defining the output set are sampled before it is returned.
Fresh validation tests decide whether to accept a candidate; their
randomness is separate from the data defining membership.

The algorithm estimates the representation function once, constructs
one high-popularity set and finitely many low-multiplicity layers,
proposes candidates from the weighted distributions of
Section~\ref{sec:deterministic-sifting-v2}, and tests them using fresh
samples.  An independent cleaning sample defines the final subset.

Set
\[
c_{\rm BSG}:=\frac{21}{128},\qquad
C_{\rm BSG}:=\frac{2^{69}}9,\qquad
C_{\rm BSG,H}:=\frac{2^{33}}9.
\]

\begin{theorem}[Persistent randomized BSG, detailed form]
\label{thm:persistent-bsg-v2}
There is a randomized algorithm with the following properties.
It receives deterministic membership access to a nonempty set
\(S\subseteq\mathbb F_2^n\), independent exact uniform samples from
\(S\), and parameters
\(\alpha,\rho_{\rm snd},\rho_{\rm cmp}\in(0,1)\).
It always terminates and returns either \(\mathsf{FAIL}\) or a finite
record \(D\).

Every nonfailure record \(D\) defines one deterministic set
\(Y_D\subseteq S\).  For an arbitrary input set \(S\), without an
energy assumption,
\begin{equation}
\Pr\!\left[
D\neq\mathsf{FAIL}
\ \text{and}\
\left(
|Y_D|<c_{\rm BSG}\sqrt\alpha\,|S|
\ \text{or}\
|Y_D+Y_D|>C_{\rm BSG}\alpha^{-4}|Y_D|
\right)
\right]
\leq \rho_{\rm snd}.
\label{eq:persistent-bsg-soundness-v2}
\end{equation}
If the output has a high-branch tag, then on the same soundness event
it satisfies the stronger estimate
\[
|Y_D+Y_D|\leq C_{\rm BSG,H}\alpha^{-3}|Y_D|.
\]

If, in addition,
\begin{equation}
E(S)\geq\alpha|S|^3,
\label{eq:persistent-bsg-energy-v2}
\end{equation}
then
\begin{equation}
\Pr[D=\mathsf{FAIL}]\leq\rho_{\rm cmp}.
\label{eq:persistent-bsg-completeness-v2}
\end{equation}

Once \(D\) is fixed, membership in \(Y_D\) is deterministic, and
every later query refers to this same set.  If
\[
\Lambda:=
1+n+\log\frac1\alpha+
\log\frac1{\rho_{\rm snd}}+
\log\frac1{\rho_{\rm cmp}},
\]
then the compilation uses
\[
O(\alpha^{-2}\Lambda^4)
\quad\text{samples from \(S\), and}\quad
O(\alpha^{-4}\Lambda^5)
\quad\text{membership queries to \(S\)}.
\]
One membership query to \(Y_D\) uses
\(O(\alpha^{-4}\Lambda^3)\) membership queries to \(S\), and the
record length is \(O(n\alpha^{-2}\Lambda)\) bits.  All these bounds
are worst-case over the random choices of the algorithm.
\end{theorem}

\subsubsection{The construction}
\label{subsec:persistent-bsg-algorithm-v2}

The exact block lengths and error allocations are given in
Appendix~\ref{app:persistent-bsg-details-v2}.  The following
description records the logical order of the algorithm.

\begin{algorithm}[t]
\caption{\textsc{PersistentBSG}}
\label{alg:persistent-bsg-v2}
\begin{algorithmic}[1]
\Require membership access and uniform samples from \(S\);
         \(\alpha,\rho_{\rm snd},\rho_{\rm cmp}\in(0,1)\)
\State draw one representation sample block and define the high set
       \(\widehat Q_{\rm H}\) and every dyadic layer
       \(\widehat Q_\tau\)
\State estimate the relation density of every dyadic layer and retain
       the layers passing the density test
\State for each retained layer, draw one codegree sample and define a
       fixed graph \(\mathcal G_\tau\)
\State generate prescribed finite lists of high and low candidates
       from the exact weighted proposal laws
\For{the candidates in a fixed deterministic order}
    \State test the density of the associated set \(X\)
    \State using fresh exact samples from \(X\), test its average
           bad-pair rate
    \If{both tests pass}
        \State draw an independent cleaning sample from \(X\)
        \If{all required finite rejection blocks succeed}
            \State \Return the finite record defining the cleaned set
        \EndIf
    \EndIf
\EndFor
\State \Return \(\mathsf{FAIL}\)
\end{algorithmic}
\end{algorithm}

All proposal loops and all rejection samplers have deterministic
finite lengths.  If a rejection block is exhausted, the current
candidate is discarded.  Exhaustion can cause the algorithm to miss a
good candidate, but it cannot cause an invalid set to be returned.

\subsubsection{Approximating the relevant relations}
\label{subsec:persistent-bsg-approximations-v2}

For \(q\in\mathbb F_2^n\), recall that
\[
p(q)=\Pr_{Z\sim U_S}[Z+q\in S].
\]
Draw independent \(Z_1,\ldots,Z_{m_{\rm rep}}\sim U_S\) and define
\[
\widehat p(q):=
\frac1{m_{\rm rep}}\sum_{i=1}^{m_{\rm rep}}
\mathbf 1_S(Z_i+q).
\]
The stored sample makes \(q\mapsto\widehat p(q)\) a fixed
deterministic function on the whole group.

Let
\[
\vartheta_{\rm H}:=\frac{h^2}{2^{12}},
\qquad
\widehat Q_{\rm H}:=
\{q:\widehat p(q)\geq 3\vartheta_{\rm H}/4\}.
\]
For the dyadic family
\[
\mathcal T_h:=
\left\{
2^{-j}:0\leq j\leq
\left\lceil\log_2\frac{2^8}{h}\right\rceil
\right\},
\]
put
\[
\widehat Q_\tau:=
\{q:\widehat p(q)\geq 3h\tau/4\}
\qquad(\tau\in\mathcal T_h).
\]

For a fixed layer \(\widehat Q_\tau\), write
\[
\delta_\tau^{\rm true}:=
\Pr_{A,X\sim U_S}[A+X\in\widehat Q_\tau].
\]
A fresh sample gives an estimate \(\widetilde\delta_\tau\), and we
use the lower confidence value
\[
\delta_{\tau,0}:=
\max\{0,\widetilde\delta_\tau-h/2^7\}.
\]
We retain the layer when
\begin{equation}
\delta_{\tau,0}\geq\frac78h
\quad\text{and}\quad
\delta_{\tau,0}\tau^2\geq\frac{h^2}{2^{16}}.
\label{eq:density-test-v2}
\end{equation}

\begin{lemma}[Representation and layer approximation]
\label{lem:representation-layer-approx-v2}
With the block sizes of
Appendix~\ref{subsec:app-global-sketches-v2}, the following event has
probability at least
\(1-\rho_{\rm rep}-\rho_{\rm den}\).

First,
\[
|\widehat p(q)-p(q)|\leq\frac{h^2}{2^{16}}
\qquad\text{for every }q\in\mathbb F_2^n.
\]
On this event,
\begin{align}
\{q:p(q)\geq\vartheta_{\rm H}\}
&\subseteq \widehat Q_{\rm H}
\subseteq
\left\{q:p(q)\geq\frac23\vartheta_{\rm H}\right\},
&
|\widehat Q_{\rm H}|
&\leq 2^{13}h^{-2}N,
\label{eq:high-layer-sandwich-v2}\\
\{q:p(q)\geq h\tau\}
&\subseteq \widehat Q_\tau
\subseteq
\left\{q:p(q)\geq\frac23h\tau\right\}
\qquad(\tau\in\mathcal T_h).
\label{eq:low-layer-sandwich-v2}
\end{align}
Moreover,
\begin{equation}
\sum_{q\notin\widehat Q_{\rm H}}r(q)^2
\leq \vartheta_{\rm H}N^3.
\label{eq:compiled-unpopular-energy-v2}
\end{equation}

Second, all relation-density estimates are simultaneously accurate:
\[
|\widetilde\delta_\tau-\delta_\tau^{\rm true}|
\leq\frac{h}{2^7}
\qquad(\tau\in\mathcal T_h).
\]
Every retained layer satisfies
\begin{equation}
\delta_{\tau,0}
\leq\delta_\tau^{\rm true}
\leq\frac{57}{56}\delta_{\tau,0},
\qquad
|\widehat Q_\tau|
\leq
\frac{2\delta_{\tau,0}N}{h\tau}.
\label{eq:certified-layer-data-v2}
\end{equation}

Finally, under the low-energy alternative of
Theorem~\ref{thm:weighted-deterministic-sifting-v2}, the threshold
supplied by Lemma~\ref{lem:visible-low-layer-v25} is retained by the
density test.
\end{lemma}

\begin{proof}
The uniform estimate for \(\widehat p\) follows from
Lemma~\ref{lem:conditional-concentration-v25}, applied to the \(2^n\)
translation predicates.
Since the error is \(\vartheta_{\rm H}/16\), the high inclusions
follow immediately.  The cardinality estimate follows from
\[
\frac23\vartheta_{\rm H}|\widehat Q_{\rm H}|
\leq
\sum_{q\in\widehat Q_{\rm H}}p(q)
\leq N.
\]
Outside \(\widehat Q_{\rm H}\) one has \(p(q)<\vartheta_{\rm H}\),
which gives \eqref{eq:compiled-unpopular-energy-v2} after summing
\(p(q)^2\leq\vartheta_{\rm H}p(q)\).

For an arbitrary searched low threshold, the ceiling in the dyadic grid only
guarantees
\[
h\tau>\frac{h^2}{2^9}.
\]
Hence the same representation error is strictly smaller than
\(2^{-7}h\tau\), which is still well inside the gap between the
thresholds \(h\tau\), \(3h\tau/4\), and \(2h\tau/3\).  This proves
\eqref{eq:low-layer-sandwich-v2}.  The density estimates follow from the same lemma, now applied to the
\(J_h\) searched layers.

If a layer is retained, then
\(\delta_{\tau,0}\leq\delta_\tau^{\rm true}\leq
\delta_{\tau,0}+h/64\).  Since
\(\delta_{\tau,0}\geq7h/8\), this is at most
\((57/56)\delta_{\tau,0}\).  Every element of
\(\widehat Q_\tau\) has \(p(q)\geq(2/3)h\tau\), and hence
\[
|\widehat Q_\tau|
\leq
\frac{3\delta_\tau^{\rm true}N}{2h\tau}
<
\frac{2\delta_{\tau,0}N}{h\tau}.
\]

For the visible threshold of
Lemma~\ref{lem:visible-low-layer-v25}, the ideal layer is contained in
\(\widehat Q_\tau\), so
\[
\delta_\tau^{\rm true}\geq\frac{15}{16}h,
\qquad
\delta_\tau^{\rm true}\tau^2\geq\frac{h^2}{2^{15}}.
\]
Subtracting the two-sided density error leaves
\(\delta_{\tau,0}>7h/8\).  The relative loss is at most \(1/60\),
so the quadratic condition in \eqref{eq:density-test-v2} also holds.
\end{proof}

Fix a retained layer and abbreviate
\[
Q:=\widehat Q_\tau,\qquad
\delta_0:=\delta_{\tau,0},\qquad
\kappa_0:=\frac{\delta_0^2}{2^{12}}.
\]
Draw independent \(W_1,\ldots,W_{m_{\rm code}}\sim U_S\) and set
\[
\widehat\gamma_Q(x,y):=
\frac1{m_{\rm code}}\sum_{i=1}^{m_{\rm code}}
\mathbf 1_Q(x+W_i)\mathbf 1_Q(y+W_i).
\]
The corresponding fixed graph is
\[
\mathcal G_\tau:=
\{(x,y)\in S^2:\widehat\gamma_Q(x,y)\geq3\kappa_0/2\}.
\]

\begin{lemma}[Codegree graph]
\label{lem:codegree-graph-v2}
With the block size of
Appendix~\ref{subsec:app-global-sketches-v2}, except with probability
\(\rho_{\rm code}\), simultaneously for every \(x,y\in S\),
\[
\left|
\widehat\gamma_Q(x,y)-\frac{c_Q(x,y)}N
\right|
\leq\frac{\kappa_0}{4}.
\]
On this event,
\begin{equation}
c_Q(x,y)\geq2\kappa_0N
\ \Longrightarrow\
(x,y)\in\mathcal G_\tau
\ \Longrightarrow\
c_Q(x,y)\geq\kappa_0N.
\label{eq:codegree-graph-sandwich-v2}
\end{equation}
\end{lemma}

\begin{proof}
Condition on the previously sampled data, which fixes
\(Q,\delta_0,\kappa_0\).  The codegree sample is fresh, and each
summand is Bernoulli with mean \(c_Q(x,y)/N\).  Lemma~\ref{lem:conditional-concentration-v25}, applied to at most
\(2^{2n}\) ordered pairs, gives the uniform estimate.  The two implications follow by comparing the thresholds
\(2\kappa_0\), \(3\kappa_0/2\), and the error \(\kappa_0/4\).
\end{proof}

The weighted laws in Lemma~\ref{lem:searchable-weighted-certificates-v25}
admit simple exact sampling rules.

\begin{proposition}[Exact weighted proposals]
\label{prop:exact-weighted-proposals-v2}
The following one-trial procedures have the stated laws.

\begin{enumerate}
\item Draw \(U,V,W\sim U_S\), put \(D:=U+V\), and accept \(D\) when
\(W+D\in S\).  The trial accepts with probability \(E(S)/N^3\), and,
conditioned on acceptance,
\[
\Pr[D=d]=\frac{r(d)^2}{E(S)}.
\]

\item Fix \(Q\subseteq\mathbb F_2^n\).  Draw \(A,X,Y\sim U_S\) and
accept \(A\) when \(A+X,A+Y\in Q\).  If
\(d_Q(a):=|\{x\in S:a+x\in Q\}|\) and
\(Z_Q:=\sum_a d_Q(a)^2\), then the trial accepts with probability
\(Z_Q/N^3\), and, conditioned on acceptance,
\[
\Pr[A=a]=\frac{d_Q(a)^2}{Z_Q}.
\]
If the relation density of \(Q\) is \(\delta\), then
\(Z_Q/N^3\geq\delta^2\).
\end{enumerate}
\end{proposition}

\begin{proof}
For the first sampler,
\[
\Pr[D=d\ \text{and accept}]
=
\frac{r(d)}{N^2}\frac{r(d)}N
=
\frac{r(d)^2}{N^3}.
\]
Summing over \(d\) gives the acceptance probability, and normalizing
gives the conditional law.

For the second sampler, the probability of drawing \(A=a\) and
accepting is \(d_Q(a)^2/N^3\).  The final lower bound follows from
Cauchy--Schwarz:
\[
Z_Q\geq\frac1N\left(\sum_a d_Q(a)\right)^2
=\delta^2N^3.
\]
\end{proof}

Each candidate request uses a fixed number of independent trials and
returns the first accepted value.  By
Lemma~\ref{lem:fixed-block-exact-sampling-v25}, conditioned on success
the output has exactly the law in
Proposition~\ref{prop:exact-weighted-proposals-v2}.  The precise
block lengths are recorded in
Appendix~\ref{subsec:app-candidate-blocks-v2}.

\subsubsection{Proposing, testing, and cleaning candidates}
\label{subsec:persistent-bsg-validation-v2}

The two branches use the same validation pattern.  Their candidate
sets and bad-pair predicates are:
\[
\begin{array}{c|c|c}
\text{branch} & \text{candidate set \(X\)} & \mathsf{Bad}(x,w)\\ \hline
\text{high, candidate \(d\)}
&
X_d:=S\cap(S+d)
&
\mathbf 1[x+w\notin\widehat Q_{\rm H}]
\\[1mm]
\text{low, candidate \(a\) in layer \(\tau\)}
&
X_a:=\{x\in S:a+x\in\widehat Q_\tau\}
&
\mathbf 1[(x,w)\notin\mathcal G_\tau].
\end{array}
\]

We first test the density of \(X\) inside \(S\).  If this test passes,
we obtain exact uniform samples from \(X\) by fixed-length rejection
sampling from \(S\), and estimate the average bad-pair rate
\[
\omega_X:=
\Pr_{x,w\sim U_X}[\mathsf{Bad}(x,w)=1].
\]
The candidate is accepted when the empirical bad-pair rate is at most
\(3/32\).

\begin{lemma}[Candidate validation]
\label{lem:candidate-validation-v2}
Condition on all randomness revealed before a candidate is fixed.
With the fresh validation block sizes of
Appendix~\ref{subsec:app-validation-v2}, the following statements
hold, each with conditional error at most \(\zeta\).

For a high candidate \(d\), the degree test accepts at threshold
\(7h/8\).  It accepts with high probability when
\(|X_d|\geq hN\), while, except with probability \(\zeta\), every
candidate that passes satisfies
\[
|X_d|>\frac34hN.
\]

For a low candidate \(a\) associated with a retained layer of
certified density \(\delta_0\), the degree test accepts at threshold
\(7\delta_0/16\).  It accepts with high probability when
\(|X_a|\geq\delta_0N/2\), while, except with probability \(\zeta\),
every candidate that passes satisfies
\[
|X_a|>\frac{3\delta_0N}{8}.
\]

In either branch, conditional on the success of the exact
\(X\)-sampling blocks, the bad-pair test has the following gap:
\[
\omega_X\leq\frac1{16}
\quad\Longrightarrow\quad
\Pr[\text{reject}]\leq\zeta,
\]
whereas, except with probability \(\zeta\), every candidate that
passes satisfies \(\omega_X<1/8\).
\end{lemma}

\begin{proof}
After conditioning, all validation samples are fresh.  Apply
Lemma~\ref{lem:finite-gap-tests-v25} to the two degree tests and
Lemma~\ref{lem:finite-gap-tests-v25} to the complementary
good-pair indicators in the bad-pair test.  The
resulting constants are recorded in
Lemma~\ref{lem:finite-gap-tests-v25} and
Appendix~\ref{subsec:app-validation-v2}.
\end{proof}

Validation controls an average bad-pair rate.  The final output needs
a pointwise guarantee for every point that it accepts.

\begin{lemma}[Independent cleaning]
\label{lem:independent-cleaning-v2}
Let \(X\subseteq S\) be nonempty, and let
\(\mathsf{Bad}:X^2\to\{0,1\}\) be fixed with
\[
\mathbb E_{x,w\sim U_X}\mathsf{Bad}(x,w)\leq\frac18.
\]
Draw fresh independent \(W_1,\ldots,W_{m_{\rm clean}}\sim U_X\) and
define
\[
Y:=
\left\{
x\in X:
\frac1{m_{\rm clean}}
\sum_{i=1}^{m_{\rm clean}}\mathsf{Bad}(x,W_i)
\leq\frac5{16}
\right\}.
\]
With the block size of
Appendix~\ref{subsec:app-validation-v2}, except with probability
\(\rho_{\rm clean}\),
\begin{equation}
|Y|\geq\frac12|X|,
\qquad
\Pr_{W\sim U_X}[\mathsf{Bad}(y,W)=1]\leq\frac38
\quad(y\in Y).
\label{eq:cleaning-output-v2}
\end{equation}
\end{lemma}

\begin{proof}
Lemma~\ref{lem:conditional-concentration-v25}, applied to the family of bad
degree predicates indexed by \(x\in X\), gives a simultaneous
additive error at most \(1/16\).  Thus every point of true bad degree at most \(1/4\) enters
\(Y\), while every point in \(Y\) has true bad degree at most \(3/8\).
By Markov's inequality, at least half the points of \(X\) have true
bad degree at most \(1/4\).
\end{proof}

The output record stores only the data used by the final membership
predicate.  Validation samples are not stored.

For a high candidate \(d\), the record contains the representation
sample, \(d\), and the cleaning sample.  It defines
\[
Y_D:=
\left\{
x\in S:
x+d\in S,\quad
\frac1{m_{\rm clean}}\sum_i
\mathbf 1[x+W_i\notin\widehat Q_{\rm H}]
\leq\frac5{16}
\right\}.
\]

For a low candidate \(a\) in layer \(\tau\), the record contains the
representation sample, \(\tau\), the empirical density count from which
\(\delta_0\) is reconstructed using the public parameters, the codegree
sample, \(a\), and the cleaning sample.  It defines
\[
Y_D:=
\left\{
x\in S:
a+x\in\widehat Q_\tau,\quad
\frac1{m_{\rm clean}}\sum_i
\mathbf 1[(x,W_i)\notin\mathcal G_\tau]
\leq\frac5{16}
\right\}.
\]
In both cases, once the record is fixed, this is a deterministic
membership predicate.

\paragraph{Deterministic assembly of a returned record.}
Assume that the representation, density, codegree, validation, and cleaning
conclusions used along the realized execution are all correct.  If
Algorithm~\ref{alg:persistent-bsg-v2} returns a nonfailure record \(D\), then
its fixed set \(Y_D\) satisfies one of the two structural criteria in
Lemma~\ref{lem:structural-criteria-v25}.  In particular,
\[
|Y_D|\ge c_{\rm BSG}hN,
\qquad
|Y_D+Y_D|\le C_{\rm BSG}h^{-8}|Y_D|,
\]
and a high-branch output satisfies the stronger
\[
|Y_D+Y_D|\le C_{\rm BSG,H}h^{-6}|Y_D|.
\]
Indeed, in the high branch the approximation event gives
\(|\widehat Q_{\rm H}|\le 2^{13}h^{-2}N\); validation gives
\(|X_d|>3hN/4\) and average bad-pair rate below \(1/8\); and independent
cleaning gives \(|Y_D|\ge |X_d|/2\) with pointwise bad degree at most
\(3|X_d|/8\).  These are exactly criterion (H).  In the low branch, the
retained-layer certificate gives
\[
\delta_0\ge\frac78h,\qquad
\delta_0\tau^2\ge\frac{h^2}{2^{16}},\qquad
|\widehat Q_\tau|\le\frac{2\delta_0N}{h\tau},
\]
while validation and cleaning give
\[
|X_a|>\frac{3\delta_0N}{8},\qquad |Y_D|\ge\frac12|X_a|.
\]
The sampled codegree sandwich ensures that every graph edge has true
\(\widehat Q_\tau\)-codegree at least
\(2^{-12}\delta_0^2N\), so criterion (L) applies.

\subsubsection{Proof of the main theorem}
\label{subsec:persistent-bsg-correctness-v2}

Useful candidates retain noticeable proposal mass after the ideal
relations are replaced by the sampled approximations.

\begin{lemma}[Good candidates have noticeable proposal mass]
\label{lem:compiled-good-mass-v2}
Assume the representation event of
Lemma~\ref{lem:representation-layer-approx-v2}.

If the high-energy alternative of
Theorem~\ref{thm:weighted-deterministic-sifting-v2} holds, then the
energy-biased distribution assigns probability at least
\(7\alpha/2^8\) to differences \(d\) satisfying
\[
|X_d|\geq hN,\qquad
\Pr_{x,w\sim U_{X_d}}
[x+w\notin\widehat Q_{\rm H}]
\leq\frac1{16}.
\]

If the low-energy alternative holds, let \(\tau\) be the visible
threshold.  On the corresponding codegree event, the squared-degree
distribution for \(\widehat Q_\tau\) assigns probability at least
\(63/128\) to centres \(a\) satisfying
\[
|X_a|\geq\frac{\delta_\tau^{\rm true}N}{2}
\geq\frac{\delta_{\tau,0}N}{2},
\]
and whose average nonedge rate in \(\mathcal G_\tau\) is at most
\(1/16\).
\end{lemma}

\begin{proof}
For the high branch, the weak-pair identity from
Lemma~\ref{lem:searchable-weighted-certificates-v25}(H) and
\eqref{eq:compiled-unpopular-energy-v2} show that high fibres with
bad-pair rate above \(1/16\) carry squared-representation mass at most
\(h^2N^3/2^8\).  The high-energy alternative carries at least
\(h^2N^3/32\), leaving at least \(7h^2N^3/2^8\) on the stated good
fibres.  Since \(E(S)\leq N^3\), their energy-biased probability is
at least \(7\alpha/2^8\).

For the low branch, write \(\delta:=\delta_\tau^{\rm true}\) and
\(\delta_0:=\delta_{\tau,0}\).  Centres of degree below
\(\delta N/2\) carry at most half of the squared-degree mass.  By the
left-hand implication in
\eqref{eq:codegree-graph-sandwich-v2}, every graph nonedge has
codegree below \(2\kappa_0N\).  The double-counting argument of
Lemma~\ref{lem:searchable-weighted-certificates-v25}(L) shows that centres with
more than a \(1/16\) fraction of nonedges carry at most
\(32\kappa_0N^3\) squared-degree mass.  Since
\(\kappa_0=\delta_0^2/2^{12}\) and \(\delta_0\leq\delta\), this is at
most \(1/128\) of the total squared-degree mass.  The remaining mass
is at least \(1-1/2-1/128=63/128\).
\end{proof}

\begin{proof}[Proof of Theorem~\ref{thm:persistent-bsg-v2}]
The exact parameter choices are stated in
Appendix~\ref{subsec:app-parameter-allocation-v2}.  Let
\[
J_h:=|\mathcal T_h|,
\qquad
R_\star:=R_{\rm H}+J_hR_{\rm L},
\qquad
\rho_\wedge:=\min\{\rho_{\rm snd},\rho_{\rm cmp}\}.
\]

\paragraph{Soundness.}
No energy assumption is made.  Let \(\mathcal E_{\rm snd}\) be the
intersection of the representation and density events, all codegree
events, all validation-soundness events, and all cleaning events for
the finitely many possible slots.  Conditional on the complete past, each validation or cleaning sample
is fresh.  Lemma~\ref{lem:sequential-conditional-failure-v2} therefore
gives
\begin{align}
\Pr[\mathcal E_{\rm snd}^c]
&\leq
\rho_{\rm rep}
+\rho_{\rm den}
+J_h\rho_{\rm code}
+2R_\star\zeta
+R_\star\rho_{\rm clean}
\nonumber\\
&=
\frac{21}{200}\rho_\wedge
<\rho_{\rm snd}.
\label{eq:soundness-ledger-v2}
\end{align}
Proposal or rejection-block exhaustion is not included here: it
discards a candidate and cannot produce an output.

On \(\mathcal E_{\rm snd}\), every nonfailure output satisfies the
hypotheses of the deterministic assembly above.  This proves
\eqref{eq:persistent-bsg-soundness-v2}.

\paragraph{Completeness.}
Assume \eqref{eq:persistent-bsg-energy-v2}.  By
Theorem~\ref{thm:weighted-deterministic-sifting-v2}, the high- or
low-energy alternative holds.

In the high case, each successful proposal request returns an exact
energy-biased candidate, and
Lemma~\ref{lem:compiled-good-mass-v2} gives good-candidate
probability at least \(7\alpha/2^8\).  The choice
\[
R_{\rm H}:=
\left\lceil
64\alpha^{-1}\log\frac{40}{\rho_{\rm cmp}}
\right\rceil
\]
makes the probability of missing every good high candidate at most
\(\rho_{\rm cmp}/40\).

In the low case, the visible layer is retained on the global
approximation event.  Each successful proposal request for that layer
has the exact squared-degree law, and
Lemma~\ref{lem:compiled-good-mass-v2} gives good-centre probability at
least \(63/128\).  The choice
\[
R_{\rm L}:=
\left\lceil
3\log\frac{40J_h}{\rho_{\rm cmp}}
\right\rceil
\]
makes the probability of missing every good centre in the visible
layer at most \(\rho_{\rm cmp}/(40J_h)\), and hence at most
\(\rho_{\rm cmp}/40\).

Every good candidate satisfies the completeness premises of
Lemma~\ref{lem:candidate-validation-v2}.  If its finite rejection
blocks succeed, it passes validation and produces a cleaned record
except on the allocated validation and cleaning events.  The complete
failure accounting is
\[
\begin{array}{l|c}
\text{failure source} & \text{upper bound}\\ \hline
\text{global approximation, validation-soundness, and cleaning events}
& 21\rho_{\rm cmp}/200\\
\text{no good candidate proposed}
& \rho_{\rm cmp}/40\\
\text{proposal-block exhaustion}
& \rho_{\rm cmp}/100\\
\text{candidate-sampling exhaustion}
& \rho_{\rm cmp}/100\\
\text{cleaning-sampling exhaustion}
& \rho_{\rm cmp}/100\\
\text{validation-completeness failures}
& \rho_{\rm cmp}/50
\end{array}
\]
because
\[
2R_\star\zeta
=
\frac{\rho_\wedge}{50}
\leq\frac{\rho_{\rm cmp}}{50}.
\]
The sum is at most
\[
\frac{9}{50}\rho_{\rm cmp}<\rho_{\rm cmp},
\]
which proves \eqref{eq:persistent-bsg-completeness-v2}.

\paragraph{Complexity.}
Appendix~\ref{subsec:app-resource-accounting-v2} proves the displayed
sample, query, record-length, and per-query bounds.  Every loop and
rejection block has a deterministic length, so these are worst-case
rather than expected bounds.
\end{proof}

\subsubsection{Exact sampling from the output set}
\label{subsec:output-sampling-v2}

The record provides deterministic membership in \(Y_D\), but later
algorithms also need a prescribed finite block of exact uniform
samples from this same set.

\begin{proposition}[Exact finite sampling from the output]
\label{prop:output-exact-sampling-v2}
Fix a nonfailure record satisfying
\[
|Y_D|\geq c_{\rm BSG}\sqrt\alpha\,|S|.
\]
Let \(t\geq1\) be fixed before sampling begins, and let
\(\rho_{\rm samp}\in(0,1)\).  For each requested sample, use an
independent block of
\[
T:=
\left\lceil
\frac{1}{c_{\rm BSG}\sqrt\alpha}
\log\frac{t}{\rho_{\rm samp}}
\right\rceil
\]
uniform samples from \(S\), returning the first point accepted by the
membership predicate of \(Y_D\).  If any block is exhausted, return
\(\mathsf{FAIL}\).

The procedure fails with probability at most \(\rho_{\rm samp}\).
Conditioned on success, the returned tuple has law
\begin{equation}
U_{Y_D}^{\otimes t}.
\label{eq:output-product-law-v2}
\end{equation}
It uses at most \(tT\) samples from \(S\) and
\[
O\!\left(
t\alpha^{-9/2}\Lambda^3
\log\frac{et}{\rho_{\rm samp}}
\right)
\]
membership queries to \(S\).
\end{proposition}

\begin{proof}
Condition on the fixed record \(D\).  Each draw from \(S\) lands in
\(Y_D\) with probability at least
\(c_{\rm BSG}\sqrt\alpha\).  Thus one block is exhausted with
probability at most
\[
(1-c_{\rm BSG}\sqrt\alpha)^T
\leq
e^{-c_{\rm BSG}\sqrt\alpha T}
\leq
\frac{\rho_{\rm samp}}t.
\]
A union bound proves the failure estimate.  In a successful block,
the first accepted point is exactly uniform on \(Y_D\).  Independence
of the blocks and their block-local success events gives
\eqref{eq:output-product-law-v2}.  The query bound follows from the
per-query cost in Theorem~\ref{thm:persistent-bsg-v2}.
\end{proof}

All samples and membership queries in the proposition refer to the
same realized set \(Y_D\).  Later compositions invoke it on the
structural-good event and charge the soundness error of
Theorem~\ref{thm:persistent-bsg-v2} separately.

\subsection{Finite implementation and resource bounds}
\label{app:persistent-bsg-details-v2}

The following block lengths implement
Section~\ref{sec:persistent-bsg-v2} with deterministic loop bounds.
We account separately for approximation, proposals, validation, and cleaning.

Throughout, \(N:=|S|\), \(h:=\sqrt\alpha\), and
\[
\Lambda:=
1+n+\log\frac1\alpha+
\log\frac1{\rho_{\rm snd}}+
\log\frac1{\rho_{\rm cmp}}.
\]

\subsubsection{Fixed-length rejection sampling}
\label{subsec:app-fixed-rejection-v2}

All rejection procedures use
Lemma~\ref{lem:fixed-block-exact-sampling-v25}.  In particular, when a
fixed set \(X\subseteq S\) has density at least \(\beta\), assigning
each of \(m\) requested samples an independent block of length
\begin{equation}
T_X(\beta,m,\varepsilon):=
\left\lceil
\beta^{-1}\log\frac{m}{\varepsilon}
\right\rceil
\label{eq:app-generic-rejection-length-v2}
\end{equation}
makes the probability that any request fails at most \(\varepsilon\).
Conditioned on all blocks succeeding, the returned tuple has the exact
product law \(U_X^{\otimes m}\), by
Lemma~\ref{lem:fixed-block-exact-sampling-v25}.

\subsubsection{Parameter and error allocation}
\label{subsec:app-parameter-allocation-v2}

Let
\[
J_h:=|\mathcal T_h|
=
1+\left\lceil\log_2\frac{2^8}{h}\right\rceil.
\]
We use
\begin{equation}
R_{\rm H}:=
\left\lceil
64\alpha^{-1}\log\frac{40}{\rho_{\rm cmp}}
\right\rceil,
\qquad
R_{\rm L}:=
\left\lceil
3\log\frac{40J_h}{\rho_{\rm cmp}}
\right\rceil,
\qquad
R_\star:=R_{\rm H}+J_hR_{\rm L}.
\label{eq:app-candidate-counts-v2}
\end{equation}
The constants \(64\) and \(3\) dominate respectively
\(2^8/7\) and \(128/63\), the reciprocals of the proposal masses in
Lemma~\ref{lem:compiled-good-mass-v2} after removing the factors
\(\alpha\) and \(1\).

Put
\[
\rho_\wedge:=\min\{\rho_{\rm snd},\rho_{\rm cmp}\}.
\]
The global and slotwise error parameters are
\begin{align}
\rho_{\rm rep}
&:=\frac{\rho_\wedge}{40},
&
\rho_{\rm den}
&:=\frac{\rho_\wedge}{40},
&
\rho_{\rm code}
&:=\frac{\rho_\wedge}{40J_h},
\label{eq:app-global-errors-v2}\\
\zeta
&:=\frac{\rho_\wedge}{100R_\star},
&
\rho_{\rm clean}
&:=\frac{\rho_\wedge}{100R_\star},
\label{eq:app-slot-errors-v2}\\
\varepsilon_{\rm prop}
&:=\frac{\rho_{\rm cmp}}{100R_\star},
&
\varepsilon_{\rm pair}
&:=\frac{\rho_{\rm cmp}}{100R_\star},
&
\varepsilon_{\rm clean}
&:=\frac{\rho_{\rm cmp}}{100R_\star}.
\label{eq:app-exhaustion-errors-v2}
\end{align}

These choices give the soundness sum
\[
\rho_{\rm rep}+\rho_{\rm den}
+J_h\rho_{\rm code}
+2R_\star\zeta
+R_\star\rho_{\rm clean}
=
\frac{21}{200}\rho_\wedge.
\]
They also give the validation-completeness contribution
\(2R_\star\zeta=\rho_\wedge/50\), which is displayed separately in
the completeness table of the main proof.

\subsubsection{Global representation, density, and codegree blocks}
\label{subsec:app-global-sketches-v2}

Set
\[
\xi_{\rm rep}:=\frac{h^2}{2^{16}},
\qquad
\xi_{\rm den}:=\frac{h}{2^7}.
\]
It is enough to take
\begin{align}
m_{\rm rep}
&:=
\left\lceil
\frac{1}{2\xi_{\rm rep}^2}
\left(n\log 2+\log\frac{2}{\rho_{\rm rep}}\right)
\right\rceil,
\label{eq:app-mrep-v2}\\
m_{\rm den}
&:=
\left\lceil
\frac{1}{2\xi_{\rm den}^2}
\log\frac{2J_h}{\rho_{\rm den}}
\right\rceil.
\label{eq:app-mden-v2}
\end{align}
For every retained low layer, with
\(\kappa_0:=\delta_0^2/2^{12}\), take the layer-specific length
\begin{equation}
m_{\rm code}:=
\left\lceil
\frac{8}{\kappa_0^2}
\left(2n\log2+\log\frac{2}{\rho_{\rm code}}\right)
\right\rceil.
\label{eq:app-mcode-v2}
\end{equation}
Hoeffding's inequality gives respectively the uniform errors
\(\xi_{\rm rep}\), \(\xi_{\rm den}\), and \(\kappa_0/4\).
In the resource bounds, \(m_{\rm code}\) denotes the maximum of
these layer-specific lengths over retained layers; this maximum is
deterministically bounded because every retained layer has
\(\delta_0\geq7h/8\).

Every retained layer has \(\delta_0\geq7h/8\), so
\(\kappa_0\geq(49/64)h^2/2^{12}\).  Consequently,
\begin{equation}
m_{\rm rep}=O(\alpha^{-2}\Lambda),
\qquad
m_{\rm den}=O(\alpha^{-1}\Lambda),
\qquad
m_{\rm code}=O(\alpha^{-2}\Lambda).
\label{eq:app-global-block-asymptotics-v2}
\end{equation}

A membership query to any compiled difference set evaluates
\(\widehat p(q)\) and therefore uses
\begin{equation}
q_Q:=m_{\rm rep}=O(\alpha^{-2}\Lambda)
\label{eq:app-qQ-v2}
\end{equation}
membership queries to \(S\).  One adjacency query to a low-branch
graph evaluates two difference predicates for every codegree sample,
and hence uses
\begin{equation}
q_{\mathcal G}:=2m_{\rm code}q_Q
=O(\alpha^{-4}\Lambda^2)
\label{eq:app-qG-v2}
\end{equation}
membership queries to \(S\).

\subsubsection{Candidate proposal and local sampling blocks}
\label{subsec:app-candidate-blocks-v2}

Under the energy promise, a high proposal trial accepts with
probability at least \(\alpha\).  For every retained low layer,
\(\delta_\tau^{\rm true}\geq\delta_0\geq7h/8\), so a low proposal
trial accepts with probability at least
\((7h/8)^2=49\alpha/64\).  We therefore use the common lower bound
\[
\beta_{\rm prop}:=\frac{49}{64}\alpha
\]
and assign every proposal request a block of length
\begin{equation}
T_{\rm prop}:=
\left\lceil
\beta_{\rm prop}^{-1}
\log\frac1{\varepsilon_{\rm prop}}
\right\rceil.
\label{eq:app-Tprop-v2}
\end{equation}
On an arbitrary input this block still has deterministic length and,
when successful, returns the exact conditional proposal law.  The
lower bound on its success probability is used only in the
completeness proof.

A candidate that reaches pair validation or cleaning is sampled from
its associated set \(X\) through independent rejection blocks.  We use
the common nominal density
\[
\beta_X:=\frac h4.
\]
For a bad-pair test requesting \(2m_{\rm pair}\) points, give each
request an independent block of length
\begin{equation}
T_{\rm pair}:=
\left\lceil
\beta_X^{-1}
\log\frac{2m_{\rm pair}}{\varepsilon_{\rm pair}}
\right\rceil.
\label{eq:app-Tpair-v2}
\end{equation}
For the \(m_{\rm clean}\) cleaning points, use
\begin{equation}
T_{\rm clean}:=
\left\lceil
\beta_X^{-1}
\log\frac{m_{\rm clean}}{\varepsilon_{\rm clean}}
\right\rceil.
\label{eq:app-Tclean-v2}
\end{equation}

For an ideal high candidate, \(|X|\geq hN\); for an ideal low
candidate, \(|X|\geq\delta_0N/2\geq7hN/16\).  Thus
\(\beta_X=h/4\) is valid in the completeness proof.  In the
soundness proof, block exhaustion merely discards a candidate.  On the
degree-validation soundness event, every candidate that proceeds
beyond the degree test also has density greater than \(\beta_X\).

The resulting asymptotic bounds are
\begin{equation}
T_{\rm prop}=O(\alpha^{-1}\Lambda),
\qquad
T_{\rm pair},T_{\rm clean}
=O(\alpha^{-1/2}\Lambda).
\label{eq:app-rejection-asymptotics-v2}
\end{equation}

\subsubsection{Validation and cleaning block sizes}
\label{subsec:app-validation-v2}

Lemma~\ref{lem:finite-gap-tests-v25} gives the
explicit choices
\begin{align}
m_{\rm deg,H}
&:=
\left\lceil
128h^{-1}\log\frac{2}{\zeta}
\right\rceil,
\label{eq:app-mdegH-v2}\\
m_{\rm deg,L}(\delta_0)
&:=
\left\lceil
256\delta_0^{-1}\log\frac{2}{\zeta}
\right\rceil,
\label{eq:app-mdegL-v2}\\
m_{\rm pair}
&:=
\left\lceil
512\log\frac{2}{\zeta}
\right\rceil.
\label{eq:app-mpair-v2}
\end{align}

For cleaning, a uniform additive error \(1/16\) over the \(2^n\)
possible query points is obtained with
\begin{equation}
m_{\rm clean}:=
\left\lceil
128\left(
n\log2+\log\frac{2}{\rho_{\rm clean}}
\right)
\right\rceil.
\label{eq:app-mclean-v2}
\end{equation}
Since every retained layer has \(\delta_0\geq7h/8\),
\begin{equation}
m_{\rm deg,H},m_{\rm deg,L}
=O(\alpha^{-1/2}\Lambda),
\qquad
m_{\rm pair},m_{\rm clean}=O(\Lambda).
\label{eq:app-validation-asymptotics-v2}
\end{equation}

In the resource table below, \(m_{\rm deg,L}\) denotes the maximum
of \(m_{\rm deg,L}(\delta_0)\) over retained layers.  Since
\(\delta_0\geq7h/8\), this maximum is deterministic.

All validation statements are conditional on the randomness revealed
before the current candidate is fixed.  The conditional error bound is
uniform over that past, so summing over the deterministic number
\(R_\star\) of possible slots is valid.

\subsubsection{Failure arithmetic}
\label{subsec:app-failure-arithmetic-v2}

The soundness contribution is
\begin{align*}
&\rho_{\rm rep}
+\rho_{\rm den}
+J_h\rho_{\rm code}
+2R_\star\zeta
+R_\star\rho_{\rm clean}\\
&\qquad=
\frac{\rho_\wedge}{40}
+\frac{\rho_\wedge}{40}
+\frac{\rho_\wedge}{40}
+\frac{\rho_\wedge}{50}
+\frac{\rho_\wedge}{100}
=
\frac{21}{200}\rho_\wedge.
\end{align*}

For completeness, the candidate counts in
\eqref{eq:app-candidate-counts-v2} give
\begin{align*}
\Pr[\text{miss every good high candidate}]
&\leq
\exp\!\left(-\frac{7\alpha}{2^8}R_{\rm H}\right)
\leq\frac{\rho_{\rm cmp}}{40},\\
\Pr[\text{miss every good centre in the visible layer}]
&\leq
\exp\!\left(-\frac{63}{128}R_{\rm L}\right)
\leq\frac{\rho_{\rm cmp}}{40J_h}.
\end{align*}
The exhaustion contributions follow from
Lemma~\ref{lem:fixed-block-exact-sampling-v25} and
\eqref{eq:app-exhaustion-errors-v2}:
\[
R_\star\varepsilon_{\rm prop}
=
R_\star\varepsilon_{\rm pair}
=
R_\star\varepsilon_{\rm clean}
=
\frac{\rho_{\rm cmp}}{100}.
\]
Finally,
\[
2R_\star\zeta
=
\frac{\rho_\wedge}{50}
\leq\frac{\rho_{\rm cmp}}{50}.
\]
Thus the full completeness sum is bounded by
\[
\frac{21}{200}
+\frac1{40}
+\frac1{100}
+\frac1{100}
+\frac1{100}
+\frac1{50}
=
\frac9{50}.
\]

\subsubsection{Worst-case resource accounting}
\label{subsec:app-resource-accounting-v2}

The candidate and layer counts satisfy
\begin{equation}
J_h=O(\Lambda),\qquad
R_{\rm H}=O(\alpha^{-1}\Lambda),\qquad
R_{\rm L}=O(\Lambda),\qquad
R_\star=O(\alpha^{-1}\Lambda^2).
\label{eq:app-count-asymptotics-v2}
\end{equation}

The following table charges every possible slot, even when the
algorithm halts earlier.  A raw sample means one call to the uniform
sampler for \(S\).

\begin{table}[t]
\centering
\small
\caption{Worst-case resources for persistent BSG compilation.}
\label{tab:app-persistent-bsg-resources-v2}
\renewcommand{\arraystretch}{1.15}
\setlength{\tabcolsep}{3pt}
\begin{tabular}{
  >{\raggedright\arraybackslash}p{0.23\textwidth}
  >{\raggedright\arraybackslash}p{0.14\textwidth}
  >{\raggedright\arraybackslash}p{0.20\textwidth}
  >{\raggedright\arraybackslash}p{0.3\textwidth}
}
\toprule
stage & number of slots & raw \(S\)-samples & membership queries to \(S\)\\
\midrule
representation sample
& \(1\)
& \(m_{\rm rep}\)
& \(0\)
\\
layer-density estimates
& \(J_h\)
& \(2J_hm_{\rm den}\)
& \(J_hm_{\rm den}q_Q\)
\\
codegree samples
& at most \(J_h\)
& \(J_hm_{\rm code}\)
& \(0\)
\\
high proposals
& \(R_{\rm H}\)
& \(3R_{\rm H}T_{\rm prop}\)
& \(R_{\rm H}T_{\rm prop}\)
\\
low proposals
& \(J_hR_{\rm L}\)
& \(3J_hR_{\rm L}T_{\rm prop}\)
& \(2J_hR_{\rm L}T_{\rm prop}q_Q\)
\\
degree tests
& at most \(R_\star\)
& \(R_{\rm H}m_{\rm deg,H}
+J_hR_{\rm L}m_{\rm deg,L}\)
& \(R_{\rm H}m_{\rm deg,H}
+J_hR_{\rm L}m_{\rm deg,L}q_Q\)
\\
candidate pair blocks
& at most \(R_\star\)
& \(2R_\star m_{\rm pair}T_{\rm pair}\)
& \(2m_{\rm pair}T_{\rm pair}
(R_{\rm H}+J_hR_{\rm L}q_Q)\)
\\
bad-pair evaluation
& at most \(R_\star\)
& \(0\)
& \(R_{\rm H}m_{\rm pair}q_Q
+J_hR_{\rm L}m_{\rm pair}q_{\mathcal G}\)
\\
cleaning blocks
& at most \(R_\star\)
& \(R_\star m_{\rm clean}T_{\rm clean}\)
& \(m_{\rm clean}T_{\rm clean}
(R_{\rm H}+J_hR_{\rm L}q_Q)\)
\\
\bottomrule
\end{tabular}
\end{table}

Summing the raw-sample column and substituting
\eqref{eq:app-global-block-asymptotics-v2},
\eqref{eq:app-rejection-asymptotics-v2},
\eqref{eq:app-validation-asymptotics-v2}, and
\eqref{eq:app-count-asymptotics-v2} gives
\begin{equation}
S_{\rm compile}=O(\alpha^{-2}\Lambda^4).
\label{eq:app-total-samples-v2}
\end{equation}
The membership-query column is dominated by low-branch bad-pair
evaluation,
\[
J_hR_{\rm L}m_{\rm pair}q_{\mathcal G}
=
O(\alpha^{-4}\Lambda^5),
\]
and every other term is no larger.  Hence
\begin{equation}
Q_{\rm compile}=O(\alpha^{-4}\Lambda^5).
\label{eq:app-total-queries-v2}
\end{equation}

A high-output membership query evaluates at most
\(m_{\rm clean}\) compiled difference predicates, with cost
\[
O(m_{\rm clean}q_Q)=O(\alpha^{-2}\Lambda^2).
\]
A low-output query evaluates one compiled difference predicate and at
most \(m_{\rm clean}\) graph predicates, giving
\begin{equation}
q_{Y_D}
=
O(q_Q+m_{\rm clean}q_{\mathcal G})
=
O(\alpha^{-4}\Lambda^3).
\label{eq:app-output-query-v2}
\end{equation}

A high record stores one representation sample, one difference, and
one cleaning sample.  A low record additionally stores one threshold,
one empirical density count, one codegree sample, and one centre.  If that
count is \(k_\tau\in\{0,\ldots,m_{\rm den}\}\), then the derived value
\[
\delta_{\tau,0}=\max\left\{0,\frac{k_\tau}{m_{\rm den}}-\frac h{2^7}\right\}
\]
is reconstructed from \(k_\tau\) and the public parameters; it is not stored
as a separate rational number.  Thus the count requires
\(O(\log m_{\rm den})\) bits, while the threshold \(\tau\) is stored by
its dyadic index.  All later threshold comparisons use this reconstructed
expression under the exact-public-parameter arithmetic convention stated
at the start of Appendix~\ref{app:probability-tools}.  The representation and codegree samples are the dominant
components, and
\begin{equation}
|D|_{\rm bits}
=
O(nm_{\rm rep})
=
O(n\alpha^{-2}\Lambda).
\label{eq:app-record-size-v2}
\end{equation}
Accounting for vector arithmetic and predicate evaluation, the
ordinary running time is
\[
O\!\left(
\alpha^{-4}\Lambda^5\operatorname{poly}(n)
\right).
\]
All bounds in this appendix are deterministic upper bounds over every
random execution.

\subsection{Deduction of the main-text compiler theorem}
\label{subsec:deduce-main-compiler}

The detailed theorem above gives promise-free soundness, high-energy
completeness, a deterministic membership predicate once the record is
fixed, and the displayed worst-case compilation and per-query bounds.
Proposition~\ref{prop:output-exact-sampling-v2} states that, for any fixed
record satisfying the retained-size conclusion, independent rejection
blocks of length
\[
\left\lceil
(c_{\rm BSG}\sqrt\alpha)^{-1}\log\frac{t}{\zeta}
\right\rceil
\]
fail with total probability at most \(\zeta\), and conditioned on all
blocks succeeding return exactly \(U_{Y_D}^{\otimes t}\).  These are
precisely the clauses of Theorem~\ref{thm:v3-compiler}.

\section{Size-Oblivious Algorithmic PFR}
\label{app:size-oblivious-pfr}
\label{sec:size-oblivious-pfr-v2}

The persistent compiler of Appendix~\ref{app:persistent-bsg} produces a
fixed implicit set together with membership access and a prescribed finite
block of exact samples, but it does not reveal the cardinality of that set.
This appendix records the clean-input PFR interface needed by the outer
learners.  The mathematical PFR step is based on the algorithmic framework
of Castro--Silva et al.~\citep{castrosilva2026algorithmicpfr}; the wrapper
below makes the input size unnecessary, keeps every randomized stage within
a deterministic finite budget, and exposes an adaptive-invocation interface
compatible with Appendix~\ref{app:persistent-bsg}.

The construction is promise based: throughout this appendix the input
\(A\subseteq\mathbb F_2^n\) is nonempty and satisfies
\begin{equation}
|A+A|\le K|A|.
\label{eq:C-pfr-promise}
\end{equation}
The routine does not try to certify this promise.  Its role in the full
paper is downstream of a persistent BSG core whose small-doubling guarantee
has already been certified on the global good event.

\begin{theorem}[Size-oblivious algorithmic PFR]
\label{thm:size-oblivious-pfr-v2}
There is a nondecreasing polynomial \(P_{\rm PFR}\) and a nondecreasing
computable function \(F_{\rm PFR}\) with the following property.  Given
\(K\ge1\), \(\rho\in(0,1)\), deterministic membership access to a nonempty
\(A\subseteq\mathbb F_2^n\) satisfying \eqref{eq:C-pfr-promise}, and
independent exact uniform samples from \(A\), a randomized algorithm that is
not given \(|A|\) outputs, with probability at least \(1-\rho\), a basis for
a subspace \(V\le\mathbb F_2^n\) such that
\begin{equation}
|V|\le |A|,
\qquad
\mathcal N_V(A)\le P_{\rm PFR}(K).
\label{eq:C-pfr-output}
\end{equation}
The number of samples requested from \(A\) is fixed before they are drawn and
obeys
\begin{equation}
s_{\rm PFR}(n,\rho)=O\!\left(n+\log\frac1\rho\right).
\label{eq:C-pfr-sample-count}
\end{equation}
The membership-query complexity and ordinary running time are at most
\begin{equation}
F_{\rm PFR}(K)\operatorname{poly}\!\left(n,\log\frac1\rho\right),
\qquad
F_{\rm PFR}(K)=2^{O(K)}K^{O(\log K)}.
\label{eq:C-pfr-resource}
\end{equation}

The same statement holds conditionally after an arbitrary adaptive past:
if a past transcript fixes a nonempty set \(A\), its deterministic membership
predicate, and the promise \eqref{eq:C-pfr-promise}, then supplying a fresh
block with law \(U_A^{\otimes s_{\rm PFR}(n,\rho)}\) and fresh internal
randomness gives conditional failure probability at most \(\rho\).
\label{cor:pfr-adaptive-invocation-v2}
\end{theorem}

The polynomial in the theorem may be enlarged, if necessary, to the fixed
nondecreasing dominating polynomial \(P_{\rm PFR}\) used in the main text.
No later argument depends on its exact coefficients.

\subsection{Localization without knowing the input size}

A short sample block first replaces the possibly exponentially sparse
ambient instance by a dense instance inside an explicitly known subspace.
Draw \(t_{\rm loc}\) independent samples
\(X_1,\ldots,X_{t_{\rm loc}}\sim U_A\), let
\[
U:=\operatorname{span}(X_1,\ldots,X_{t_{\rm loc}}),
\qquad
A_0:=A\cap U,
\]
and write \(r:=\dim U\).

\begin{lemma}[Localization and safe size bracketing]
\label{lem:C-localization-size}
\label{lem:pfr-localization-v2}
\label{lem:pfr-size-scales-v2}
For every \(\rho_{\rm loc},\rho_{\rm size}\in(0,1)\), one can choose
\begin{equation}
t_{\rm loc}
=\left\lceil 8\left(n+\log\frac1{\rho_{\rm loc}}\right)\right\rceil
\label{eq:C-localization-count}
\end{equation}
so that, except with probability \(\rho_{\rm loc}\),
\begin{equation}
|A_0|\ge \frac12|A|,
\qquad
\operatorname{span}(A_0)=U,
\qquad
|A_0+A_0|\le L|A_0|,
\quad L:=2K.
\label{eq:C-localization-event}
\end{equation}
On this event, using \(2^{O(L)}\log(1/\rho_{\rm size})\) uniform samples from
the explicit subspace \(U\) and the same number of membership queries to
\(A\), the algorithm computes integers \(\underline N,\widehat N\) such
that, except with an additional probability \(\rho_{\rm size}\),
\begin{equation}
1\le \underline N\le |A_0|\le \widehat N<4|A_0|.
\label{eq:C-size-sandwich}
\end{equation}
If \(|A_0|\ge4\), then also \(\underline N\ge |A_0|/4\).
\end{lemma}

\begin{proof}
Let \(U_j=\operatorname{span}(X_1,\ldots,X_j)\).  As long as
\(|A\cap U_j|<|A|/2\), the next exact sample lies outside \(U_j\) with
conditional probability greater than \(1/2\), and every such event increases
the dimension by one.  Coupling these growth indicators from below by
independent Bernoulli\((1/2)\) variables and applying the conditional
Chernoff bound from Appendix~\ref{app:probability-tools} proves the first
claim for the choice \eqref{eq:C-localization-count}.  Since the sampled
points lie in \(A_0\) and span \(U\), we also have
\(\operatorname{span}(A_0)=U\); the doubling estimate follows from
\(|A_0+A_0|\le|A+A|\le K|A|\le2K|A_0|\).

For the size bracket we use the finite-field spanning estimate from the
clean algorithmic PFR framework~\citep[Theorem~5.1]{castrosilva2026algorithmicpfr}:
if \(B\subseteq W\), \(\operatorname{span}(B)=W\), and
\(|B+B|\le L|B|\), then
\begin{equation}
|W|\le \frac{2^{2L}}{2L}|B|.
\label{eq:C-spanning-bound}
\end{equation}
Applied to \(A_0\subseteq U\), this gives the density lower bound
\(\theta_0:=|A_0|/|U|\ge2^{-2L}\).  The learner can sample uniformly from
\(U\) using its basis, so a multiplicative estimate
\(\widehat\theta\) of \(\theta_0\) with relative error at most \(1/4\)
requires only \(2^{O(L)}\log(1/\rho_{\rm size})\) membership tests.  Set
\[
N_-:=\frac23\widehat\theta|U|,
\qquad
N_+:=\frac43\widehat\theta|U|,
\]
\[
\underline N:=\max\{1,\lfloor N_-\rfloor\},
\qquad
\widehat N:=
\begin{cases}
1,&N_+=0,\\
2^{\lceil\log_2N_+\rceil},&N_+>0.
\end{cases}
\]
On the concentration event,
\[
\frac12|A_0|\le N_-\le|A_0|
\le N_+\le\frac53|A_0|.
\]
Integer and dyadic rounding give \eqref{eq:C-size-sandwich}; if
\(|A_0|\ge4\), rounding the lower endpoint loses at most another factor two.
\end{proof}

The upper scale \(\widehat N\) is used only to choose a dense linear model;
the lower scale \(\underline N\) is used only to enforce the final
nontriviality condition \(|V|\le|A|\).

\subsection{A bounded dense model with deterministic image oracles}

Fix a coordinate isomorphism \(\chi_U:U\to\mathbb F_2^r\).  Let
\begin{equation}
m_0:=\left\lceil
\log_2\widehat N+4\log_2L+c_{\rm mod}
\right\rceil,
\qquad
b_L:=\lceil2L\rceil+4,
\qquad
m:=\max\{m_0,r-b_L\},
\label{eq:C-model-dimension}
\end{equation}
where \(c_{\rm mod}\) is a sufficiently large absolute constant.  If
\(m\ge r\), use the fixed injection
\(u\mapsto(\chi_U(u),0^{m-r})\).  If \(m<r\), draw five independent uniform
\(m\times r\) binary matrices and use the first full-row-rank matrix; if none
exists, the current model trial fails.

\begin{lemma}[Safe dense linear model and image access]
\label{lem:C-safe-model}
\label{lem:pfr-bounded-surjection-v2}
\label{lem:pfr-safe-model-v2}
\label{lem:pfr-image-oracles-v2}
Assume the events in Lemma~\ref{lem:C-localization-size}.  The construction
above has the following properties.
\begin{enumerate}
\item[(a)] If \(m<r\), model generation fails with probability below
\(1/32\), and conditioned on success the selected map
\(\pi:U\to\mathbb F_2^m\) is exactly uniform among surjective linear maps.
\item[(b)] The safeguard in \eqref{eq:C-model-dimension} is inactive on the
good size event, so \(m=m_0\) and
\begin{equation}
2^m\le C_{\rm mod}L^4|A_0|
\label{eq:C-model-density}
\end{equation}
for an absolute constant \(C_{\rm mod}\).  Conditional on successful model
generation, \(\pi\) is Freiman-good,
\begin{equation}
\ker\pi\cap(4A_0)=\{0\},
\label{eq:C-Freiman-good}
\end{equation}
with probability at least \(3/4\).
\item[(c)] For every generated model, whether or not it is Freiman-good, the
image set \(S_\pi:=\pi(A_0)\) has a deterministic membership oracle and a
deterministic canonical inverse oracle \(g_\pi:S_\pi\to\mathbb F_2^r\),
each using at most \(2^{b_L}=2^{O(K)}\) membership queries to \(A\) and
\(2^{O(K)}\operatorname{poly}(n)\) time.
\end{enumerate}
\end{lemma}

\begin{proof}
A uniform \(m\times r\) matrix with \(m<r\) has full row rank with
probability
\[
\prod_{i=0}^{m-1}(1-2^{i-r})>\frac12.
\]
Thus all five draws fail with probability below \(2^{-5}\).  Every full-rank
matrix has the same probability of being the first successful draw, so the
conditional distribution is uniform over surjections.

The spanning bound \eqref{eq:C-spanning-bound} gives
\(r\le\log_2|A_0|+2L\).  Since \(\widehat N\ge|A_0|\), the term
\(r-b_L\) is strictly below \(m_0\); hence \(m=m_0\).  The upper bracket
\(\widehat N<4|A_0|\) then gives \eqref{eq:C-model-density}.  For a uniform
surjection and any fixed nonzero \(z\in\mathbb F_2^r\),
\[
\Pr[\pi(z)=0]
=\frac{2^{r-m}-1}{2^r-1}
\le2^{-m}.
\]
Pl\"unnecke's inequality gives \(|4A_0|\le L^4|A_0|\), so a union bound and
the choice of \(c_{\rm mod}\) imply \eqref{eq:C-Freiman-good} with
probability at least \(3/4\).  If \(m\ge r\), the fixed injection is
Freiman-good deterministically.

Finally, \(m\ge r-b_L\) ensures \(\dim\ker\pi\le b_L\) on every transcript.
To answer a query at \(x\in\mathbb F_2^m\), solve \(\pi(u)=x\) in \(U\).
If the fiber is empty then \(x\notin S_\pi\); otherwise enumerate its at most
\(2^{b_L}\) points, test membership in \(A\), and return the first accepted
point in a fixed total order.  Composing that point with \(\chi_U\) gives
the canonical inverse value \(g_\pi(x)\).
\end{proof}

The safeguard \(m\ge r-b_L\) is important only off the main good event: it
prevents a bad size estimate from creating exponentially large fibers and
therefore makes every trial a finite worst-case computation.

\subsection{The imported restricted-homomorphism primitive}

The wrapper uses one nontrivial external algorithmic primitive: the
restricted homomorphism routine from the clean algorithmic PFR framework.
The theorem below is a normalized restatement of the restricted-homomorphism
lemma of Castro--Silva et al.~\citep[Lemma~5.10]{castrosilva2026algorithmicpfr};
we spell out the interface
because both the hypothesis and the conclusion are measured on the ambient
scale \(2^m\).

\begin{theorem}[Algorithmic restricted homomorphism; imported interface]
\label{thm:C-restricted-homomorphism}
\label{thm:pfr-rh-primitive-v2}
There are a polynomial \(P_{\rm RH}\) and a nondecreasing computable function
\(F_{\rm RH}(\mathcal K)=\mathcal K^{O(\log(2\mathcal K))}\) such that the
following holds.  Let \(S'\subseteq\mathbb F_2^m\) and
\(f:S'\to\mathbb F_2^r\).  If
\begin{equation}
\left|
\left\{
(x_1,x_2,x_3,x_4)\in(S')^4:
\begin{array}{l}
x_1+x_2=x_3+x_4,\\
f(x_1)+f(x_2)=f(x_3)+f(x_4)
\end{array}
\right\}
\right|
\ge \frac{2^{3m}}{\mathcal K},
\label{eq:C-RH-premise}
\end{equation}
then, given membership-query access to \(S'\) and value-query access to
\(f\), the routine returns with probability at least \(0.7\) a linear map
\(M:\mathbb F_2^m\to\mathbb F_2^r\) and \(v\in\mathbb F_2^r\) satisfying
\begin{equation}
|\{x\in S':f(x)=Mx+v\}|
\ge \frac{2^m}{P_{\rm RH}(\mathcal K)}.
\label{eq:C-RH-conclusion}
\end{equation}
Whenever the premise holds, its query and time bounds are
\(F_{\rm RH}(\mathcal K)\operatorname{poly}(m+r)\).  In the wrapper below,
we fix a deterministic resource cap \(B_{\rm RH}(\mathcal K,m,r)\) for the
imported implementation; after enlarging \(F_{\rm RH}\) if necessary, this cap
is at most \(F_{\rm RH}(\mathcal K)\operatorname{poly}(m+r)\) on the displayed
asymptotic scale.  Every invocation is forcibly stopped at this cap; a timeout, an explicit failure
output, or a malformed output is treated as a failed trial.  Hence every
wrapper execution has a finite worst-case budget even when the premise is
false, while on promised inputs the cap does not alter the imported guarantee.
\end{theorem}

On a Freiman-good model, \(g_\pi\) is the true inverse of
\(\pi|_{A_0}\).  Hence every additive quadruple in \(S_\pi\) is respected by
\(g_\pi\), and the number of respected quadruples equals \(E(A_0)\).  Since
\[
E(A_0)\ge\frac{|A_0|^4}{|A_0+A_0|}
\ge\frac{|A_0|^3}{L},
\]
while \eqref{eq:C-model-density} gives
\(2^m/|A_0|\le C_{\rm mod}L^4\), the premise
\eqref{eq:C-RH-premise} holds with
\begin{equation}
\mathcal K_\star(L):=C_{\rm RH}L^{13},
\qquad C_{\rm RH}:=C_{\rm mod}^3.
\label{eq:C-RH-parameter}
\end{equation}
Thus the imported primitive is invoked only at an \(L\)-dependent parameter.

\subsection{Fresh validation, covering, and safe truncation}

Replace \(P_{\rm RH}\) by a nondecreasing polynomial upper envelope with
\(P_{\rm RH}\ge1\) and set
\begin{equation}
\beta_L:=\frac{1}{2P_{\rm RH}(\mathcal K_\star(L))}.
\label{eq:C-beta}
\end{equation}
For a returned affine candidate \((M,v)\), define its absolute agreement
\begin{equation}
p_\pi(M,v):=
\Pr_{X\sim U_{\mathbb F_2^m}}
[X\in S_\pi\ \text{and}\ g_\pi(X)=MX+v].
\label{eq:C-agreement}
\end{equation}
A successful restricted-homomorphism call on a Freiman-good model has
\(p_\pi(M,v)\ge2\beta_L\).  The wrapper validates \((M,v)\) on a fresh
uniform block from \(\mathbb F_2^m\) and accepts when the empirical agreement
is at least \(3\beta_L/2\).

\begin{lemma}[Validated affine agreement yields a safe PFR output]
\label{lem:C-validated-output}
\label{lem:pfr-rh-premise-v2}
\label{lem:pfr-agreement-validation-v2}
\label{lem:pfr-agreement-cover-v2}
\label{lem:pfr-truncation-v2}
\label{lem:pfr-output-path-v2}
Assume the localization and size-bracketing events and the model-size bound
\eqref{eq:C-model-density}.
\begin{enumerate}
\item[(a)] For any conditional error budget \(\zeta\in(0,1)\), a fresh
validation block of
\begin{equation}
m_{\rm agr}
=\left\lceil16\beta_L^{-1}\log\frac2\zeta\right\rceil
\label{eq:C-validation-size}
\end{equation}
accepts every candidate with \(p_\pi(M,v)\ge2\beta_L\) except with
probability \(\zeta\), while a candidate with
\(p_\pi(M,v)\le\beta_L\) is accepted with probability at most \(\zeta\).
\item[(b)] Every candidate with \(p_\pi(M,v)>\beta_L\) determines a
subspace \(H\le U\) satisfying
\begin{equation}
\mathcal N_H(A)<\frac{L}{\beta_L}.
\label{eq:C-H-cover}
\end{equation}
This implication does not require the realized model to be Freiman-good.
\item[(c)] Let
\[
d_V:=\min\{\dim H,\lfloor\log_2\underline N\rfloor\}
\]
and choose any \(d_V\)-dimensional \(V\le H\).  Then
\begin{equation}
|V|\le|A|,
\qquad
[H:V]\le C_{\rm tr}L^4,
\label{eq:C-truncation}
\end{equation}
for an absolute constant \(C_{\rm tr}\), and consequently
\begin{equation}
\mathcal N_V(A)
\le
2C_{\rm tr}L^5P_{\rm RH}(C_{\rm RH}L^{13}).
\label{eq:C-final-cover}
\end{equation}
\end{enumerate}
\end{lemma}

\begin{proof}
Part (a) is the multiplicative gap test of
Lemma~\ref{lem:finite-gap-tests-v25} with
\((a,\theta,b)=(1,3/2,2)\) and scale \(q=\beta_L\); the constant in
\eqref{eq:C-validation-size} is sufficient for both tails.

For part (b), let
\[
C:=\{f_\pi(x):x\in S_\pi,\ g_\pi(x)=Mx+v\},
\]
where \(f_\pi(x)\) denotes the canonical support point represented by
\(g_\pi(x)\).  Distinct image points have distinct canonical representatives,
so \(|C|=p_\pi(M,v)2^m\).  Moreover, \(C\) lies in one affine coset of
\(H:=\chi_U^{-1}(\operatorname{Im}M)\).  Since
\(m\ge m_0\ge\log_2\widehat N\),
\[
|C|>\beta_L2^m
\ge\beta_L\widehat N
\ge\beta_L|A_0|
\ge\frac{\beta_L}{2}|A|.
\]
Also \(C+C\subseteq H\) and \(C\subseteq A\), hence
\[
|A+C|\le|A+A|\le K|A|
<\frac{2K}{\beta_L}|C|
=\frac{L}{\beta_L}|C|.
\]
Ruzsa's covering lemma therefore covers \(A\) by fewer than
\(L/\beta_L\) translates of \(C+C\subseteq H\), proving
\eqref{eq:C-H-cover}.  Notice that this argument uses only validated
agreement and the original small-doubling promise; Freiman-goodness is needed
for completeness, not for soundness of an accepted candidate.

For part (c), the definition gives
\(|V|\le\underline N\le|A_0|\le|A|\).  If \(V\ne H\), then
\(|V|>\underline N/2\), so
\[
[H:V]<\frac{2|H|}{\underline N}
\le\frac{2\cdot2^m}{\underline N}.
\]
When \(|A_0|\ge4\), Lemma~\ref{lem:C-localization-size} gives
\(\underline N\ge|A_0|/4\), and \eqref{eq:C-model-density} yields
\([H:V]=O(L^4)\).  When \(|A_0|<4\), the trivial lower bound
\(|V|\ge1\) and the same model-size estimate give the same conclusion after
increasing the absolute constant.  Refining each \(H\)-coset into
\([H:V]\) many \(V\)-cosets and substituting
\(\beta_L^{-1}=2P_{\rm RH}(C_{\rm RH}L^{13})\) proves
\eqref{eq:C-final-cover}.
\end{proof}

\subsection{Algorithm, correctness, and deterministic budgets}

The wrapper has one preparation stage followed by a fixed number of model
trials.  It returns the first validated subspace; different successful
trials need not produce comparable subspaces.

\begin{algorithm}[t]
\caption{\textsc{SizeObliviousPFR}}
\label{alg:size-oblivious-pfr-v2}
\begin{algorithmic}[1]
\Require membership access and exact uniform samples from \(A\); \(K\ge1\);
         failure parameter \(\rho\)
\State draw the prescribed localization sample block, form
       \(U=\operatorname{span}(X_1,\ldots,X_{t_{\rm loc}})\), and set
       \(A_0=A\cap U\)
\State estimate the density of \(A_0\) in the explicit subspace \(U\) and
       compute \(\underline N,\widehat N\)
\For{a prescribed number \(R_{\rm trial}\) of independent trials}
    \State generate a bounded-time linear model \(\pi:U\to\mathbb F_2^m\)
    \State expose deterministic membership/value oracles for
           \((S_\pi,g_\pi)\)
    \State invoke the restricted-homomorphism routine with resource cap
           \(B_{\rm RH}(\mathcal K_\star,m,r)\)
    \If{the call times out, returns failure, or does not return a well-formed
          pair \((M,v)\) of the prescribed dimensions}
        \State \textbf{continue} to the next model trial
    \EndIf
    \State validate the returned affine map using a fresh ambient block
    \If{validation passes}
        \State construct \(H\), truncate to \(V\le H\) using
               \(\underline N\), and \Return \(V\)
    \EndIf
\EndFor
\State \Return \(\mathsf{FAIL}\)
\end{algorithmic}
\end{algorithm}

\begin{proof}[Proof of Theorem~\ref{thm:size-oblivious-pfr-v2}]
Allocate
\[
\rho_{\rm loc}:=\frac\rho4,
\qquad
\rho_{\rm size}:=\frac\rho4,
\qquad
R_{\rm trial}:=
\left\lceil\frac{\log(4/\rho)}{\log(5/3)}\right\rceil,
\qquad
\zeta:=\frac{\rho}{8R_{\rm trial}}.
\]
Condition on successful localization and size bracketing.  In one fresh model
trial, bounded-time model generation succeeds with probability at least
\(31/32\); conditional on success, the model is Freiman-good with
probability at least \(3/4\).  On a Freiman-good model, the discussion after
Theorem~\ref{thm:C-restricted-homomorphism} verifies the restricted-
homomorphism premise with parameter \(C_{\rm RH}L^{13}\).  The prescribed
resource cap therefore does not truncate the promised execution, and the
imported routine returns a well-formed good affine candidate with probability
at least \(7/10\).  A good affine candidate
has true agreement at least \(2\beta_L\), and fresh validation accepts it
with probability at least \(1-\zeta\ge7/8\).  Therefore one trial produces a
good accepted candidate with probability at least
\begin{equation}
\frac{31}{32}\cdot\frac34\cdot\frac7{10}\cdot\frac78
=\frac{4557}{10240}>\frac25.
\label{eq:C-one-trial-success}
\end{equation}
The probability that all \(R_{\rm trial}\) trials miss is at most
\((3/5)^{R_{\rm trial}}\le\rho/4\).

Soundness of accepted candidates uses a different event.  A candidate may be
chosen adaptively from the entire preceding execution, but its validation
block is fresh.  By Lemma~\ref{lem:finite-gap-tests-v25} and sequential
conditional accounting from Lemma~\ref{lem:sequential-conditional-failure-v2},
the probability that any accepted candidate has true agreement at most
\(\beta_L\) is at most
\(R_{\rm trial}\zeta=\rho/8\).  On the localization, size-bracketing, and
validation-soundness events, every returned subspace therefore satisfies
Lemma~\ref{lem:C-validated-output}.  The total failure probability is at most
\[
\frac\rho4+\frac\rho4+\frac\rho4+\frac\rho8
=\frac{7\rho}{8}<\rho.
\]
Taking
\begin{equation}
P_{\rm PFR}(K)
:=2C_{\rm tr}(2K)^5
P_{\rm RH}\!\left(C_{\rm RH}(2K)^{13}\right)
\label{eq:C-PFR-polynomial}
\end{equation}
proves \eqref{eq:C-pfr-output}; enlarging this polynomial if necessary makes
it agree with the dominating convention of the main text.

Only the localization stage requests samples from the input set \(A\), so
\eqref{eq:C-pfr-sample-count} is fixed before the block is drawn.  The size
estimate samples from the explicitly known subspace \(U\) and therefore
requires only membership queries to \(A\).  Every image-model operation costs
at most \(2^{O(K)}\) membership queries to \(A\).  There are
\(O(\log(1/\rho))\) model trials; every imported call is capped at
\(B_{\rm RH}(\mathcal K,m,r)\), with
\(\mathcal K=C_{\rm RH}(2K)^{13}\), and this cap is bounded by
\(F_{\rm RH}(\mathcal K)\operatorname{poly}(m+r)\) for the fixed imported
implementation.  Thus bad models cannot create unbounded executions.
Validation has
polynomial dependence on \(K\).  Thus the query and time bound has the form
\eqref{eq:C-pfr-resource}.  More explicitly, for a sufficiently large
absolute constant \(C\), one may take
\[
F_{\rm PFR}(K)
=2^{CK}\left(
1+F_{\rm RH}(C_{\rm RH}(2K)^{13})
+P_{\rm RH}(C_{\rm RH}(2K)^{13})
\right)
=2^{O(K)}K^{O(\log K)}.
\]

Finally condition on an arbitrary adaptive past transcript that fixes the
input predicate and supplies a fresh product block
\(U_A^{\otimes s_{\rm PFR}(n,\rho)}\).  After conditioning, every random
choice made above has exactly the standalone law, so the same proof gives
conditional failure probability at most \(\rho\).  This is precisely the
form needed when the input \(A\) is a persistent core generated earlier in
the outer learner.  Proposition~\ref{prop:output-exact-sampling-v2} supplies
such a fresh exact block from the fixed core using finite rejection budgets.
\end{proof}

\subsection{Resource ledger and interface summary}
\label{app:pfr-details-v2}

For completeness, the deterministic budgets can be read directly from the
proof.  With \(L=2K\), let
\[
\mathcal K_\star=C_{\rm RH}L^{13},
\qquad
\beta_L=\frac{1}{2P_{\rm RH}(\mathcal K_\star)}.
\]
The input-sample block has size
\[
t_{\rm loc}
=\left\lceil8\left(n+\log\frac4\rho\right)\right\rceil.
\]
The localized density is at least \(2^{-2L}\), so
\[
m_{\rm size}
=O\!\left(2^{2L}\log\frac1\rho\right)
\]
uniform points of \(U\) suffice for size bracketing.  Every model trial uses
exactly five random matrices when \(m<r\), at most one resource-capped
restricted-homomorphism call, and (only when that call returns a well-formed
affine candidate) a validation block of length
\[
m_{\rm agr}
=O\!\left(
P_{\rm RH}(C_{\rm RH}L^{13})
\log\frac{R_{\rm trial}}\rho
\right).
\]
The number of trials is \(R_{\rm trial}=O(\log(1/\rho))\).  Every query to the
image model expands to at most \(2^{b_L}=2^{O(K)}\) membership queries to the
original input.  All these budgets are fixed before the corresponding fresh
random blocks are drawn.

The interface used later can therefore be summarized as
\[
\boxed{
\begin{array}{c}
\text{fixed membership predicate for a nonempty }A
\\[1mm]
+\ U_A^{\otimes s_{\rm PFR}(n,\rho)}\text{ supplied exactly}
\\[1mm]
+\ |A+A|\le K|A|
\end{array}}
\quad\Longrightarrow\quad
\boxed{
\begin{array}{c}
|V|\le|A|,
\\[1mm]
\mathcal N_V(A)\le P_{\rm PFR}(K)
\end{array}}
\]
with conditional failure probability at most \(\rho\) even after an
arbitrary adaptive past.  This is the only PFR interface required by the
moderate- and arbitrary-corruption learners.

\section{Structural Lifting from an Observed Core}
\label{app:structural-lifting}

A cover learned on a core \(Y\subseteq B\) transfers to a compatible
support \(A\) whenever the core contains enough clean mass. The two
deterministic bounds below provide this transfer and enforce the size cap.
They apply to every compatible support after the random core is fixed.

For \(\gamma\in[0,1]\) and \(\eta\in[0,1)\), define
\begin{equation}
\Delta(\gamma,\eta)
:=\gamma-\max\{\gamma,1-\gamma\}\eta.
\label{eq:D-purity-margin}
\end{equation}

\begin{lemma}[Clean mass forced by observed size]
\label{lem:D-clean-mass}
\label{lem:robust-clean-mass-v2}
Suppose \(Y\subseteq B\subseteq\mathbb F_2^n\),
\(|Y|\ge\gamma|B|\), and
\(|A\triangle B|\le\eta|A|\).  Then
\begin{equation}
|A\cap Y|
\ge
[\Delta(\gamma,\eta)]_+|A|,
\qquad [x]_+:=\max\{x,0\}.
\label{eq:D-clean-mass}
\end{equation}
In particular, when \(0<\gamma\le1/2\),
\begin{equation}
\Delta(\gamma,\eta)=\gamma-(1-\gamma)\eta.
\label{eq:D-small-gamma}
\end{equation}
\end{lemma}

\begin{proof}
Write
\[
u:=\frac{|A\setminus B|}{|A|},
\qquad
v:=\frac{|B\setminus A|}{|A|}.
\]
Then \(u,v\ge0\), \(u+v\le\eta\), and
\(|B|=(1-u+v)|A|\).  Because \(Y\subseteq B\),
\begin{align*}
|A\cap Y|
&\ge |Y|-|B\setminus A|\\
&\ge \gamma|B|-v|A|\\
&=\bigl(\gamma-\gamma u-(1-\gamma)v\bigr)|A|\\
&\ge \Delta(\gamma,\eta)|A|.
\end{align*}
Combining with the trivial lower bound zero proves \eqref{eq:D-clean-mass}.
If \(\gamma\le1/2\), then \(1-\gamma\ge\gamma\), giving
\eqref{eq:D-small-gamma}.
\end{proof}

The next lemma is the common structural transfer used in both algorithmic
regimes.  Its hypotheses are phrased directly in terms of clean mass; the
moderate-corruption learner obtains that mass from Lemma~\ref{lem:D-clean-mass},
while the arbitrary-corruption learner obtains it by averaging across the
peeled cores.

\begin{lemma}[Clean-core structural lifting with size-oblivious truncation]
\label{lem:D-structural-lifting}
\label{thm:one-core-lifting-v2}
\label{lem:list-serving-entry-v2}
Let \(A,Y,B\subseteq\mathbb F_2^n\), with \(A\ne\varnothing\) and
\(Y\subseteq B\), and let \(H\le\mathbb F_2^n\).  Suppose
\begin{equation}
|A+A|\le K|A|,
\qquad
|A\cap Y|\ge\delta|A|,
\qquad
|H|\le|Y|,
\qquad
\mathcal N_H(Y)\le P,
\label{eq:D-lifting-hypotheses}
\end{equation}
for some \(\delta>0\), and assume only the observable-size comparison
\begin{equation}
|B|<2|A|.
\label{eq:D-B-size}
\end{equation}
Define \(V\le H\) by the fixed rule
\begin{equation}
V:=H \quad\text{if }H=\{0\},
\qquad
V:=\operatorname{span}(h_1,\ldots,h_{r-1})
\quad\text{otherwise},
\label{eq:D-fixed-truncation}
\end{equation}
where \(h_1,\ldots,h_r\) is the basis of \(H\) returned by the PFR routine.
Then
\begin{equation}
|V|\le|A|,
\qquad
\mathcal N_V(A)\le\frac{2KP}{\delta}.
\label{eq:D-lifting-conclusion}
\end{equation}
The construction of \(V\) uses neither \(A\), \(|A|\), \(\eta\), nor
\(\delta\).
\end{lemma}

\begin{proof}
Set \(C:=A\cap Y\), so \(|C|\ge\delta|A|\).  Cover \(Y\) by
\(m:=\mathcal N_H(Y)\le P\) distinct \(H\)-cosets.  One coset contains a
subset \(C'\subseteq C\) of size at least
\(|C|/m\ge\delta|A|/m\).  Since \(C'\subseteq A\),
\[
|A+C'|
\le|A+A|
\le K|A|
\le\frac{Km}{\delta}|C'|.
\]
Ruzsa's covering lemma therefore yields
\begin{equation}
\mathcal N_H(A)
\le\frac{Km}{\delta}
\le\frac{KP}{\delta}.
\label{eq:D-H-cover}
\end{equation}

It remains only to enforce the nontriviality condition in the cost
\(q_A\).  From \(|H|\le|Y|\le|B|<2|A|\), if \(H=\{0\}\) then
\(|H|\le|A|\) already.  Otherwise the rule \eqref{eq:D-fixed-truncation}
produces an index-two subspace \(V\le H\), so
\[
|V|=\frac{|H|}{2}<|A|.
\]
Every \(H\)-coset splits into two \(V\)-cosets, hence
\[
\mathcal N_V(A)
\le2\mathcal N_H(A)
\le\frac{2KP}{\delta}.
\]
The rule depends only on the returned basis of \(H\), proving the final
claim.
\end{proof}

Two specializations will be used repeatedly.  First, if
\(|A\triangle B|\le\eta|A|\), then \(|B|\le(1+\eta)|A|<2|A|\), so
\eqref{eq:D-B-size} is automatic.  Second, if a core satisfies
\(|Y|\ge\gamma|B|\) and \(\Delta(\gamma,\eta)>0\), then
Lemma~\ref{lem:D-clean-mass} and Lemma~\ref{lem:D-structural-lifting} give
\begin{equation}
q_A(V)=\mathcal N_V(A)
\le
\frac{2KP}{\Delta(\gamma,\eta)}.
\label{eq:D-one-core-specialization}
\end{equation}
This is the one-core transfer used in the moderate-corruption theorem.  In
the list theorem, the analysis instead proves
\(|A\cap Y_j|\ge\delta_A|A|\) for a serving index \(j=j(A)\) and applies
\eqref{eq:D-lifting-conclusion} directly.  No size-only purity argument is
needed there.

\section{Intermediate List Construction under Arbitrary Corruption}
\label{app:list-implementation}
\label{app:list-implementation-v2}

We construct candidates for the arbitrary-corruption learner by peeling
persistent cores, applying PFR to each core, and lifting its cover to latent
supports. The compiler, PFR, and lifting guarantees are supplied by
Appendices~\ref{app:persistent-bsg}--\ref{app:structural-lifting}.
One observation-dependent success event will serve the entire identification set.

Throughout, fix \(K\ge1\), \(\eta\in[0,1)\), and
\[
\varepsilon:=1-\eta,
\qquad
\mathcal I:=\mathcal I_{K,\eta}(B)\ne\varnothing.
\]
For every \(A\in\mathcal I\), write \(C_A:=A\cap B\).  The shared
insertion--deletion budget implies
\begin{equation}
|C_A|\ge \varepsilon|A|,
\qquad
|B|\le(1+\eta)|C_A|,
\qquad
|C_A|>\frac{|B|}{2}.
\label{eq:E-shared-budget}
\end{equation}
The last inequality is the strict-majority property from the main text.

Set
\begin{equation}
\alpha_\star:=\frac{\varepsilon}{128K},
\qquad
\gamma_\star:=c_{\rm BSG}\sqrt{\alpha_\star},
\qquad
L_\star:=C_{\rm BSG}\alpha_\star^{-4},
\qquad
P_\star:=P_{\rm PFR}(L_\star).
\label{eq:E-structural-parameters}
\end{equation}
Thus
\[
\gamma_\star=\Theta\!\left(\sqrt{\varepsilon/K}\right),
\qquad
L_\star=O(K^4\varepsilon^{-4}).
\]
The deterministic peel cap is
\begin{equation}
M_\star
:=
\left\lceil
\frac{\ln4}{-\ln(1-\gamma_\star)}
\right\rceil
=
O\!\left(\sqrt{K/\varepsilon}\right).
\label{eq:E-peel-cap}
\end{equation}
Only \(K,\eta\), and the absolute compiler constants enter these
parameters.

\begin{proposition}[Intermediate universal structural list]
\label{prop:v3-intermediate-list}
There exist absolute constants \(C_{\rm list},C_{\rm list,L}>0\) and a
computable function \(G_{\rm list}\) such that the following holds.  Let
\(\varepsilon=1-\eta\), suppose
\(\mathcal I_{K,\eta}(B)\neq\varnothing\), and let \(\rho\in(0,1)\).
Given \(K,\eta,\rho\) and \(\mathsf{SampMQ}(B)\), a randomized algorithm
always terminates, returning either \(\mathsf{FAIL}\) or a finite list
\(\mathcal L(B)=\{V_1,\ldots,V_m\}\).  With probability at least
\(1-\rho\), the returned list is nonempty and simultaneously satisfies
\[
 1\le m\le C_{\rm list}\sqrt{K/\varepsilon},
\]
and, for every \(A\in\mathcal I_{K,\eta}(B)\),
\[
 \min_{V\in\mathcal L(B)}q_A(V)
 \le
 C_{\rm list}K^{3/2}\varepsilon^{-3/2}
 P_{\rm PFR}\!\left(C_{\rm list,L}K^4\varepsilon^{-4}\right).
\]
The number of exact samples from \(B\) is polynomial in
\(K,\varepsilon^{-1},n,\log(1/\rho)\), while membership queries and
ordinary running time are at most
\[
G_{\rm list}(K,\varepsilon^{-1})
\left(1+n+\log\frac1\rho\right)^{O(\sqrt{K/\varepsilon})}.
\]
Appendix~\ref{app:list-to-common} selects one common output from this list.
\end{proposition}

\subsection{The relative residual certificate}

Measure the clean mass remaining in a residual relative to the
\emph{initial observed clean mass} \(C_A=A\cap B\). This gives a
constant-fraction stopping target even when \(1-\eta\) is small.

\begin{lemma}[Retained observed mass forces residual energy]
\label{lem:E-residual-energy}
Let \(R\subseteq B\).  If some \(A\in\mathcal I\) satisfies
\begin{equation}
|A\cap R|\ge\frac12|C_A|,
\label{eq:E-residual-premise}
\end{equation}
then
\begin{equation}
\frac{E(R)}{|R|^3}
\ge
\frac{\varepsilon}{16K(1+\eta)^3}
\ge
\frac{\varepsilon}{128K}
=
\alpha_\star.
\label{eq:E-residual-energy}
\end{equation}
Consequently, if \(E(R)<\alpha_\star|R|^3\), then simultaneously for every
\(A\in\mathcal I\),
\begin{equation}
|A\cap R|<\frac12|A\cap B|.
\label{eq:E-low-energy-terminal}
\end{equation}
\end{lemma}

\begin{proof}
Fix a witness \(A\) to \eqref{eq:E-residual-premise}.  Since
\(A\cap R\subseteq A\),
\[
E(R)
\ge E(A\cap R)
\ge
\frac{|A\cap R|^4}{|(A\cap R)+(A\cap R)|}
\ge
\frac{|C_A|^4}{16K|A|}.
\]
Also \(R\subseteq B\), so by \eqref{eq:E-shared-budget},
\[
|R|^3
\le |B|^3
\le (1+\eta)^3|C_A|^3.
\]
Therefore
\[
\frac{E(R)}{|R|^3}
\ge
\frac{|C_A|}{16K(1+\eta)^3|A|}
\ge
\frac{\varepsilon}{16K(1+\eta)^3}.
\]
Since \(1+\eta<2\), the latter quantity is at least
\(\varepsilon/(128K)\).  The final statement is the contrapositive, applied
separately to every \(A\in\mathcal I\).
\end{proof}

The contrapositive is uniform over the identification set: a low-energy
residual retains less than half of the initial clean overlap for every
compatible explanation.

\subsection{Finite residual access}

At round \(j\), previously compiled cores \(Y_0,\ldots,Y_{j-1}\) determine
the residual
\begin{equation}
R_0:=B,
\qquad
R_j:=B\setminus\bigcup_{i<j}Y_i.
\label{eq:E-residual-definition}
\end{equation}
Conditional on the complete past, \(R_j\) is one fixed deterministic set.
Its predicate is
\begin{equation}
\mathbf 1_{R_j}(x)
=
\mathbf 1_B(x)
\prod_{i<j}\bigl(1-\mathbf 1_{Y_i}(x)\bigr).
\label{eq:E-residual-predicate}
\end{equation}

We use the fixed constants
\[
\xi_{\rm res}:=\frac1{32},
\qquad
\theta_{\rm res}:=\frac3{16},
\qquad
\beta_{\rm res}:=\frac5{32}.
\]
A residual-density estimate with additive error \(\xi_{\rm res}\) therefore
creates a strict gap between stopping and continuation.

\begin{proposition}[Residual test and exact compiler input]
\label{prop:E-residual-access}
Fix a past transcript that determines \(R\subseteq B\), a requested sample
count \(s\ge1\), and a conditional failure budget \(\rho_0\in(0,1)\).
There is a finite procedure using fresh samples from \(B\) and the fixed
predicate of \(R\) with the following properties.

Draw
\begin{equation}
m_{\rm res}
:=
\left\lceil
\frac{1}{2\xi_{\rm res}^2}
\log\frac{2}{\rho_0}
\right\rceil
\label{eq:E-residual-test-count}
\end{equation}
samples from \(B\), let \(\widehat p_R\) be the empirical fraction lying in
\(R\), and return \(\mathsf{STOP}\) when
\(\widehat p_R\le\theta_{\rm res}\).  Except with conditional probability
\(\rho_0\),
\begin{equation}
\mathsf{STOP}
\Longrightarrow
\frac{|R|}{|B|}\le\frac7{32},
\qquad
\mathsf{CONTINUE}
\Longrightarrow
\frac{|R|}{|B|}>\frac5{32}.
\label{eq:E-density-gap}
\end{equation}
In particular, on this event a \(\mathsf{STOP}\) decision implies
\begin{equation}
|A\cap R|<\frac12|A\cap B|
\qquad
\text{for every }A\in\mathcal I.
\label{eq:E-density-stop-terminal}
\end{equation}

On the density-estimation good event in \eqref{eq:E-density-gap}, and
conditional on \(\mathsf{CONTINUE}\), generate each requested point by an
independent rejection block from \(U_B\) of length
\begin{equation}
T_{\rm res}
:=
\left\lceil
\beta_{\rm res}^{-1}
\log\frac{s}{\rho_0}
\right\rceil.
\label{eq:E-residual-bridge-length}
\end{equation}
The probability that any of the \(s\) blocks is exhausted is at most
\(\rho_0\); conditioned on all blocks succeeding, their accepted values have
the exact product law
\begin{equation}
U_R^{\otimes s}.
\label{eq:E-residual-product-law}
\end{equation}
\end{proposition}

\begin{proof}
Conditional on the past, the indicators
\(\mathbf 1_R(X)\) for fresh \(X\sim U_B\) are i.i.d. Bernoulli with mean
\(p_R=|R|/|B|\).  Hoeffding's inequality and
\eqref{eq:E-residual-test-count} give
\(|\widehat p_R-p_R|\le1/32\) except with probability \(\rho_0\).  The
threshold \(3/16=6/32\) then yields \eqref{eq:E-density-gap}.

If \(\mathsf{STOP}\), then
\[
|A\cap R|
\le |R|
\le\frac7{32}|B|
<
\frac14|B|
<
\frac12|A\cap B|
\]
for every compatible \(A\), where the last inequality uses
\eqref{eq:E-shared-budget}.

On the density-estimation good event, a continue branch has
\(p_R>\beta_{\rm res}\).  Conditional on that event and the realized
\(\mathsf{CONTINUE}\) decision, the fixed-block exact rejection lemma,
Lemma~\ref{lem:fixed-block-exact-sampling-v25}, with
acceptance set \(R\), gives block-exhaustion probability at most
\(\rho_0/s\) and exact \(U_R\) law for the first accepted point conditional
on success.  Independence of the blocks gives
\eqref{eq:E-residual-product-law}.
\end{proof}

\subsection{Peeling and the uniform terminal invariant}

Let
\begin{equation}
\rho_0:=\frac{\rho}{10M_\star}.
\label{eq:E-slot-budget}
\end{equation}
Let \(s_{\rm BSG}^\star\) be a deterministic upper bound on the number of
uniform samples requested by Theorem~\ref{thm:persistent-bsg-v2} at
parameters \((\alpha_\star,\rho_0,\rho_0)\), and let
\[
s_{\rm PFR}^\star:=s_{\rm PFR}(n,\rho_0)
\]
be the fixed sample count in Theorem~\ref{thm:size-oblivious-pfr-v2}.

For a reached round \(j\), condition on the complete past and on a successful
residual sample bridge, so that the compiler receives exact product samples
from the now-fixed residual \(R_j\).  Let \(D_j\) denote its output and define
the two one-sided compiler failure events
\begin{align}
\mathcal F_{{\rm snd},j}
&:=\bigl\{D_j\neq\mathsf{FAIL}\ \text{and the structural conclusion of
Theorem~\ref{thm:persistent-bsg-v2} fails}\bigr\},
\label{eq:E-bsg-soundness-failure}\\
\mathcal F_{{\rm cmp},j}
&:=\bigl\{E(R_j)\ge\alpha_\star|R_j|^3
\ \text{and}\ D_j=\mathsf{FAIL}\bigr\}.
\label{eq:E-bsg-completeness-failure}
\end{align}
By the promise-free soundness and high-energy completeness clauses of
Theorem~\ref{thm:persistent-bsg-v2}, each event has conditional probability
at most \(\rho_0\), uniformly over the realized past.  Thus, outside
\(\mathcal F_{{\rm snd},j}\), every nonfailure record is structurally valid;
outside \(\mathcal F_{{\rm cmp},j}\), an internal compiler
\(\mathsf{FAIL}\) implies \(E(R_j)<\alpha_\star|R_j|^3\).

The learner starts with \(R_0=B\).  In a reached round \(j<M_\star\), it
performs the following operations using fresh random blocks.

\begin{enumerate}
\item Apply Proposition~\ref{prop:E-residual-access} to \(R_j\) with
      \(s=s_{\rm BSG}^\star\).  On \(\mathsf{STOP}\), terminate normally.
      Exhaustion of a residual rejection block is an algorithmic
      \(\mathsf{FAIL}\).
\item Run the persistent compiler on the exact \(R_j\)-sample block with
      parameter \(\alpha_\star\) and soundness/completeness budgets
      \((\rho_0,\rho_0)\).  If the compiler returns its internal
      \(\mathsf{FAIL}\), terminate normally.  Otherwise let \(Y_j\subseteq
      R_j\) be the fixed output set.
\item Outside \(\mathcal F_{{\rm snd},j}\),
      \begin{equation}
      |Y_j|\ge\gamma_\star|R_j|,
      \qquad
      |Y_j+Y_j|\le L_\star|Y_j|.
      \label{eq:E-valid-core}
      \end{equation}
      Since a reached continue branch has
      \(|R_j|>(5/32)|B|\), such a core obeys
      \begin{equation}
      \frac{|Y_j|}{|B|}
      >
      \beta_Y,
      \qquad
      \beta_Y:=\frac{5\gamma_\star}{32}.
      \label{eq:E-core-density}
      \end{equation}
      Generate the \(s_{\rm PFR}^\star\) samples required by the PFR
      routine directly from \(B\), using independent \(Y_j\)-rejection
      blocks of length
      \begin{equation}
      T_Y
      :=
      \left\lceil
      \beta_Y^{-1}
      \log\frac{s_{\rm PFR}^\star}{\rho_0}
      \right\rceil.
      \label{eq:E-core-bridge-length}
      \end{equation}
      Exhaustion is an algorithmic \(\mathsf{FAIL}\); conditional on
      success, the block has exact law
      \(U_{Y_j}^{\otimes s_{\rm PFR}^\star}\).
\item Invoke Theorem~\ref{thm:size-oblivious-pfr-v2} on the fixed set
      \(Y_j\), with doubling parameter \(L_\star\) and error budget
      \(\rho_0\).  A PFR failure is an algorithmic \(\mathsf{FAIL}\).
      Otherwise let \(H_j\) be the returned subspace and define \(V_j\)
      from \(H_j\) by the fixed truncation rule
      \eqref{eq:D-fixed-truncation}.
\item Store \(Y_j,V_j\) and set \(R_{j+1}:=R_j\setminus Y_j\).
\end{enumerate}

The distinction between the two kinds of failure is important.  A compiler
\(\mathsf{FAIL}\) is a \emph{normal stop}: on its completeness good event it
certifies that the current residual does not meet the energy threshold.
Rejection-block exhaustion and PFR failure carry no such structural
meaning and cause the outer learner to return \(\mathsf{FAIL}\).

\begin{lemma}[Uniform terminal invariant and serving core]
\label{lem:E-terminal-serving}
Assume the following observation-dependent good events hold in every
reached round: the density conclusion of
Proposition~\ref{prop:E-residual-access}; success of the residual and core
sample bridges; neither \(\mathcal F_{{\rm snd},j}\) nor
\(\mathcal F_{{\rm cmp},j}\) occurs; and the PFR output is valid.  Suppose
the outer learner does not return algorithmic \(\mathsf{FAIL}\).

If it produces \(m\ge1\) cores, then \(m\le M_\star\), the cores are
pairwise disjoint, and its final residual \(R_{\rm fin}\) satisfies
simultaneously for every \(A\in\mathcal I\),
\begin{equation}
|A\cap R_{\rm fin}|
<
\frac12|A\cap B|.
\label{eq:E-terminal-invariant}
\end{equation}
Consequently,
\begin{equation}
\sum_{j<m}|A\cap Y_j|
>
\frac12|A\cap B|
\ge
\frac{\varepsilon}{2}|A|
\qquad
(A\in\mathcal I),
\label{eq:E-total-captured-mass}
\end{equation}
and for every \(A\in\mathcal I\) there exists an index \(j=j(A)<m\) such
that
\begin{equation}
|A\cap Y_j|
\ge
\frac{\varepsilon}{2M_\star}|A|.
\label{eq:E-serving-core}
\end{equation}
Moreover, on the same good events the process cannot terminate normally
before producing its first core.
\end{lemma}

\begin{proof}
Because \(Y_j\subseteq R_j\) and
\(|Y_j|\ge\gamma_\star|R_j|\),
\[
|R_{j+1}|
=
|R_j|-|Y_j|
\le
(1-\gamma_\star)|R_j|.
\]
Thus the stored cores are pairwise disjoint and, after \(M_\star\) valid
peels,
\[
|R_{M_\star}|
\le
(1-\gamma_\star)^{M_\star}|B|
\le
\frac14|B|.
\]
By \eqref{eq:E-shared-budget},
\[
\frac14|B|
<
\frac12|A\cap B|
\]
for every \(A\in\mathcal I\).  Hence reaching the deterministic peel cap
implies \eqref{eq:E-terminal-invariant}.

If instead the residual-density test returns \(\mathsf{STOP}\), then
\eqref{eq:E-density-stop-terminal} gives the same conclusion.  Finally,
suppose the persistent compiler returns its internal \(\mathsf{FAIL}\).
Outside \(\mathcal F_{{\rm cmp},j}\), an internal compiler
\(\mathsf{FAIL}\) at the final reached round can happen only when
\[
E(R_{\rm fin})<\alpha_\star|R_{\rm fin}|^3.
\]
Lemma~\ref{lem:E-residual-energy} then gives
\eqref{eq:E-terminal-invariant}.  These exhaust the normal stopping
mechanisms.

Since
\[
B\setminus R_{\rm fin}
=
\bigsqcup_{j<m}Y_j,
\]
the terminal invariant gives
\[
\sum_{j<m}|A\cap Y_j|
=
|A\cap B|-|A\cap R_{\rm fin}|
>
\frac12|A\cap B|
\ge
\frac{\varepsilon}{2}|A|,
\]
which is \eqref{eq:E-total-captured-mass}.  Averaging over
\(m\le M_\star\) cores yields \eqref{eq:E-serving-core}.

It remains to rule out a normal stop at \(R_0=B\).  For any
\(A\in\mathcal I\),
\[
|A\cap R_0|=|A\cap B|>\frac12|A\cap B|,
\]
so a density stop is impossible on its good event, and
Lemma~\ref{lem:E-residual-energy} gives
\(E(B)\ge\alpha_\star|B|^3\).  Hence a compiler \(\mathsf{FAIL}\) is also impossible outside the
corresponding one-sided event \(\mathcal F_{{\rm cmp},0}\).
Therefore at least one valid core is produced.
\end{proof}

For the serving index in \eqref{eq:E-serving-core}, the PFR output satisfies
\[
|H_j|\le|Y_j|,
\qquad
\mathcal N_{H_j}(Y_j)\le P_\star.
\]
Applying Lemma~\ref{lem:D-structural-lifting} with
\[
\delta=\frac{\varepsilon}{2M_\star}
\]
gives
\begin{equation}
q_A(V_j)
=
\mathcal N_{V_j}(A)
\le
\frac{4KM_\star}{\varepsilon}P_\star.
\label{eq:E-cover-before-substitution}
\end{equation}
Substituting \eqref{eq:E-structural-parameters} and
\eqref{eq:E-peel-cap} yields
\begin{equation}
q_A(V_j)
\le
O(K^{3/2}\varepsilon^{-3/2})
P_{\rm PFR}\!\left(O(K^4\varepsilon^{-4})\right),
\label{eq:E-final-cover}
\end{equation}
uniformly over \(A\in\mathcal I\), with the serving index allowed to depend
on \(A\).

\subsection{Failure accounting and resources}

All probability statements are organized chronologically.  At the
beginning of a reached round, the past transcript fixes the residual
predicate.  Conditional on the density-test good event and successful
residual bridge, the compiler receives an exact product block from that
fixed residual.  Conditional on a valid compiled core and successful direct
core bridge, the PFR routine receives an exact product block from that fixed
core.  Thus every invoked guarantee can be applied conditionally on the
entire past.

There are at most six charged events per reached round: residual-density
estimation, residual-bridge exhaustion, the one-sided events
\(\mathcal F_{{\rm snd},j}\) and \(\mathcal F_{{\rm cmp},j}\), core-bridge
exhaustion, and PFR failure.  With the slot budget
\eqref{eq:E-slot-budget}, sequential conditional failure accounting,
Lemma~\ref{lem:sequential-conditional-failure-v2}, gives
\begin{equation}
\Pr[\mathcal E_{\rm global}^c]
\le
6M_\star\rho_0
=
\frac{3}{5}\rho
<
\rho.
\label{eq:E-global-failure}
\end{equation}
On \(\mathcal E_{\rm global}\), all hypotheses of
Lemma~\ref{lem:E-terminal-serving} hold simultaneously.

For explicit complexity bounds, put
\begin{equation}
\Lambda_E
:=
1+n+\log\frac{K}{\varepsilon}
+\log\frac{M_\star}{\rho}.
\label{eq:E-Lambda}
\end{equation}

\begin{proposition}[Finite sample bound and recursive oracle cost]
\label{prop:E-resources}
The outer list learner can be implemented with
\begin{equation}
O\!\left(
\left(\frac{K}{\varepsilon}\right)^{5/2}
\Lambda_E^5
\right)
\label{eq:E-sample-bound}
\end{equation}
uniform samples from \(B\).

There is a nondecreasing computable function
\(G_{\rm list}(K,\varepsilon^{-1})\) such that its membership-query
complexity and ordinary running time are both at most
\begin{equation}
G_{\rm list}(K,\varepsilon^{-1})
\left(1+n+\log\frac1\rho\right)^{O(M_\star)}
=
G_{\rm list}(K,\varepsilon^{-1})
\left(1+n+\log\frac1\rho\right)^{
O(\sqrt{K/\varepsilon})}.
\label{eq:E-XP-bound}
\end{equation}
All rejection blocks and all sample counts are deterministic once
\(K,\eta,\rho,n\) are fixed.
\end{proposition}

\begin{proof}
At one reached round, the density test uses
\(O(\log(1/\rho_0))\) samples from \(B\).  By
Theorem~\ref{thm:persistent-bsg-v2}, the compiler requests
\[
s_{\rm BSG}^\star
=
O(\alpha_\star^{-2}\Lambda_E^4)
=
O\!\left((K/\varepsilon)^2\Lambda_E^4\right)
\]
uniform samples from its residual input.  Proposition
\ref{prop:E-residual-access} realizes those requests from \(B\) with blocks
of length
\[
T_{\rm res}
=
O\!\left(\log\frac{s_{\rm BSG}^\star}{\rho_0}\right)
=
O(\Lambda_E).
\]
Hence the residual bridge uses
\[
O\!\left((K/\varepsilon)^2\Lambda_E^5\right)
\]
samples from \(B\) per round.

The PFR routine requests
\[
s_{\rm PFR}^\star
=
O\!\left(n+\log\frac1{\rho_0}\right)
=
O(\Lambda_E)
\]
samples from \(Y_j\).  Since \(\beta_Y=\Theta(\gamma_\star)\),
\eqref{eq:E-core-bridge-length} gives
\[
T_Y
=
O\!\left(
\sqrt{K/\varepsilon}\,\Lambda_E
\right).
\]
Thus the direct core bridge costs at most
\(O(\sqrt{K/\varepsilon}\,\Lambda_E^2)\) samples from \(B\) per round.
Multiplying all per-round costs by
\(M_\star=O(\sqrt{K/\varepsilon})\) proves the conservative bound
\eqref{eq:E-sample-bound}.

For queries, let \(Q_j\) be the worst-case number of
\(\mathsf{MQ}_B\) calls needed to evaluate the fixed residual predicate
\(\mathbf1_{R_j}(x)\) once.  The compiler theorem gives a bound
\[
C_\star
=
O(\alpha_\star^{-4}\Lambda_E^3)
\]
on the number of membership queries to \(R_i\) needed for one evaluation of
a stored \(Y_i\)-predicate.  Therefore
\begin{equation}
Q_0=1,
\qquad
Q_j
\le
1+C_\star\sum_{i<j}Q_i
\le
(1+C_\star)^j.
\label{eq:E-residual-query-recurrence}
\end{equation}
The density tests and residual rejection bridges evaluate
\(\mathbf1_{R_j}\); compiler compilation makes only a
parameter-bounded number of such calls; evaluating \(Y_j\) costs at most
\(C_\star Q_j\); and the PFR invocation makes
\(F_{\rm PFR}(L_\star)\operatorname{poly}(\Lambda_E)\) calls to that fixed
\(Y_j\)-predicate.  Since \(\alpha_\star^{-1}\), \(L_\star\), and all
nonpolynomial dependence of the PFR routine are functions only of
\(K,\varepsilon^{-1}\), summing the reached rounds and absorbing these
parameter-only factors into \(G_{\rm list}\) gives
\[
G_{\rm list}(K,\varepsilon^{-1})
\Lambda_E^{O(M_\star)}.
\]
Replacing \(\Lambda_E\) by a parameter-only factor times
\(1+n+\log(1/\rho)\) proves \eqref{eq:E-XP-bound}.  The same recurrence
bounds ordinary running time because each oracle evaluation and every
linear-algebra operation used by the compiler/PFR routines has the
corresponding finite worst-case cost.
\end{proof}

\subsection{Assembly of the intermediate list guarantee}

\begin{proof}[Proof of Proposition~\ref{prop:v3-intermediate-list}]
Run the peeling procedure above.  It always terminates: there are at most
\(M_\star\) reached rounds, and every randomized subroutine and rejection
bridge has a predetermined finite budget.  Outside the event
\(\mathcal E_{\rm global}\), the algorithm may return \(\mathsf{FAIL}\).

Work now on \(\mathcal E_{\rm global}\), which has probability at least
\(1-\rho\) by \eqref{eq:E-global-failure}.  Because
\(\mathcal I\ne\varnothing\), Lemma~\ref{lem:E-terminal-serving} shows that
the returned list is nonempty and has
\[
1\le m\le M_\star
=
O\!\left(\sqrt{K/\varepsilon}\right).
\]
For every \(A\in\mathcal I\), the same lemma supplies a serving core
\(Y_{j(A)}\), and \eqref{eq:E-final-cover} together with the deterministic
truncation in Lemma~\ref{lem:D-structural-lifting} gives
\[
q_A(V_{j(A)})
\le
O(K^{3/2}\varepsilon^{-3/2})
P_{\rm PFR}\!\left(O(K^4\varepsilon^{-4})\right).
\]
The event \(\mathcal E_{\rm global}\), the realized list, and all stored
predicates depend only on \(B\) and the learner's randomness.  The argument
choosing \(j(A)\) is deterministic after that event is fixed.  Hence the
same realized list satisfies the guarantee simultaneously for every
\(A\in\mathcal I\), which is the required
\(\forall A\,\exists j(A)\) quantifier order.

Finally, Proposition~\ref{prop:E-resources} gives the polynomial sample
bound and the XP membership-query/time bound stated in
Proposition~\ref{prop:v3-intermediate-list}.
\end{proof}

\section{From an Intermediate List to One Common Hypothesis}
\label{app:list-to-common}

The list from Appendix~\ref{app:list-implementation} supplies feasible
descriptions for the majority-coverage rule. Fresh samples let us apply
that rule to the fixed list and prove Lemma~\ref{lem:v3-list-common}.

\begin{proof}[Proof of Lemma~\ref{lem:v3-list-common}]
Condition on an arbitrary adaptive past that fixes an explicit list
\(V_1,\ldots,V_m\) satisfying the serving property in the lemma.  Put
\(b=\varepsilon/8\).  Draw a fresh sample block from \(U_B\).  For each
\(V_i\), group sampled points by their \(V_i\)-coset, let \(Z_i\) be the
union of the at most \(Q\) most frequent observed cosets, and write
\(\widehat P(Z_i)\) for their empirical mass.  Keep indices with
\[
 \widehat P(Z_i)\ge \frac12+b
\]
and return, with fixed tie-breaking, a kept entry of minimum dimension.
If no entry is kept, return failure.

For an elementary finite-class implementation, consider all unions of at
most \(Q\) cosets of any listed subspace.  This class has at most
\(m2^{nQ}\) nonempty descriptions: use \(Q\) ambient representatives,
repeating representatives when fewer cosets are needed.  Hence
\[
 s=\left\lceil
 \frac{32}{\varepsilon^2}
 \left(nQ\ln2+\ln\frac{2m}{\zeta}\right)
 \right\rceil
\]
samples suffice, by Hoeffding and a union bound, for one event of
probability at least \(1-\zeta\) on which
\begin{equation}
\label{eq:v3-post-uniform}
 |\widehat P(Z)-U_B(Z)|\le b
\end{equation}
for every such union.  The event includes the data-dependent top-frequency
choices \(Z_i\).

Fix this event.  For every compatible \(A\), choose a serving entry
\(V_{i(A)}\).  Since \(\mathcal N_{V_{i(A)}}(A)\le Q\), some union \(Z_A\)
of at most \(Q\) \(V_{i(A)}\)-cosets contains \(A\).  The shared budget gives
\[
 U_B(Z_A)\ge \frac{|A\cap B|}{|B|}
 \ge \frac{1}{1+\eta}
 =\frac{1}{2-\varepsilon}
 \ge \frac12+\frac{\varepsilon}{4}.
\]
The empirically most frequent \(Q\) cosets maximize empirical coverage for
that fixed subspace, so
\[
 \widehat P(Z_{i(A)})\ge \widehat P(Z_A)
 \ge \frac12+\frac{\varepsilon}{8}.
\]
Thus every compatible support has at least one retained serving entry.  The
minimum-dimension choice therefore yields
\[
 |V_\star|\le |V_{i(A)}|\le |A|
 \qquad\text{simultaneously for all }A\in\mathcal I_{K,\eta}(B).
\]

The selected union satisfies \(U_B(Z_\star)\ge1/2\) by
\eqref{eq:v3-post-uniform}.  For every compatible \(A\), the majority-capture
calculation from Section~\ref{sec:v3-statistical-proof} gives
\[
 |A\cap B\cap Z_\star|\ge \frac{\varepsilon}{2}|A|.
\]
One of the at most \(Q\) cosets in \(Z_\star\) therefore contains a set
\(C\subseteq A\cap B\) of size at least
\(\varepsilon|A|/(2Q)\).  Because \(C+C\subseteq V_\star\) and
\(|A+C|\le K|A|\), the maximal-disjoint-translate covering argument yields
\[
 \mathcal N_{V_\star}(A)
 \le \frac{|A+C|}{|C|}
 \le \frac{2KQ}{\varepsilon}.
\]
The same concentration event works for all compatible supports; no union
bound over \(A\) is taken.  This proves the lemma.
\end{proof}

\paragraph{Dimension-free sample refinement.}
For a fixed \(V_i\), unions of at most \(Q\) of its cosets have VC dimension
at most \(Q\): a shattered set must occupy distinct cosets, and more than
\(Q\) points cannot be shattered.  Standard uniform convergence followed by
a union bound over the \(m\) fixed subspaces therefore permits
\[
 s=O\!\left(
 \varepsilon^{-2}\bigl[Q\log(2/\varepsilon)+\log(2m/\zeta)\bigr]
 \right).
\]
The finite-description bound above remains available when one wants a fully
elementary count.

\begin{proof}[Proof of the arbitrary-corruption clause of Theorem~\ref{thm:v3-query-common}]
Run Proposition~\ref{prop:v3-intermediate-list} with failure budget
\(\rho/2\).  Abort on its failure or if its list exceeds the prescribed
parameter-only cap, which makes every downstream execution finite.  On its
good event, apply Lemma~\ref{lem:v3-list-common} with
\[
 Q=\left\lceil
 C_{\rm list}K^{3/2}\varepsilon^{-3/2}
 P_{\rm PFR}\!\left(C_{\rm list,L}K^4\varepsilon^{-4}\right)
 \right\rceil
\]
and \(\zeta=\rho/2\).  The two failure probabilities sum to at most \(\rho\),
and the lemma gives
\[
 q_A(V_\star)
 \le
 O\!\left(
 K^{5/2}\varepsilon^{-5/2}
 P_{\rm PFR}(O(K^4\varepsilon^{-4}))
 \right)
 \qquad\forall A\in\mathcal I_{K,\eta}(B).
\]
The postprocessor makes no membership queries.  Its additional sample and
linear-algebra costs are polynomial in the structural parameters and the
list size, so the query/time recurrence from
Proposition~\ref{prop:E-resources} retains the XP form
\eqref{eq:v3-general-common-xp}.  This proves the full-range clause.
\end{proof}
\section{Common Learning under Moderate Corruption}
\label{app:moderate-common}
\label{app:moderate-common-v2}

Under moderate corruption, a single persistent core is large enough to
contain clean mass from every compatible support. We prove this size
guarantee and then apply deterministic lifting on one observation-level
success event, establishing the first part of
Theorem~\ref{thm:v3-query-common}.

Fix
\begin{equation}
\alpha_K:=\frac{1}{2K},
\qquad
\gamma_K:=c_{\rm BSG}\sqrt{\alpha_K}
=\frac{c_{\rm BSG}}{\sqrt{2K}},
\qquad
L_K:=C_{\rm BSG}\alpha_K^{-4}
=16C_{\rm BSG}K^4.
\label{eq:F-parameters}
\end{equation}
Choose once and for all an absolute constant
\begin{equation}
0<c_{\rm rob}
\le
\min\left\{
\frac1{32},
\frac{c_{\rm BSG}}{4\sqrt2}
\right\}.
\label{eq:F-crob}
\end{equation}
The numerical values are not important; only the two constant margins
below are used.

\begin{lemma}[Observation energy and uniform clean mass]
\label{lem:F-observation-core}
Let
\[
\bar\eta:=\frac{c_{\rm rob}}{\sqrt K},
\qquad
\mathcal I_{\rm rob}:=\mathcal I_{K,\bar\eta}(B),
\]
and suppose \(\mathcal I_{\rm rob}\ne\varnothing\).  Then \(B\ne\varnothing\)
and
\begin{equation}
E(B)\ge\alpha_K|B|^3.
\label{eq:F-observation-energy}
\end{equation}
Moreover, every set \(Y\subseteq B\) satisfying
\(|Y|\ge\gamma_K|B|\) obeys, simultaneously for all
\(A\in\mathcal I_{\rm rob}\),
\begin{equation}
|A\cap Y|
\ge
\delta_K|A|,
\qquad
\delta_K
\ge
\frac{3c_{\rm BSG}}{4\sqrt{2K}}
=
\Omega(K^{-1/2}).
\label{eq:F-uniform-clean-mass}
\end{equation}
\end{lemma}

\begin{proof}
Fix \(A\in\mathcal I_{\rm rob}\) and put \(C:=A\cap B\).  Since
\(|A\triangle B|\le\bar\eta|A|\),
\[
|C|\ge(1-\bar\eta)|A|,
\qquad
|B|\le(1+\bar\eta)|A|.
\]
Also \(C+C\subseteq A+A\), hence
\[
E(B)
\ge E(C)
\ge\frac{|C|^4}{|C+C|}
\ge
\frac{(1-\bar\eta)^4}{K}|A|^3.
\]
After normalizing by \(|B|^3\),
\begin{equation}
\frac{E(B)}{|B|^3}
\ge
\frac{(1-\bar\eta)^4}{(1+\bar\eta)^3K}.
\label{eq:F-energy-ratio}
\end{equation}
By \eqref{eq:F-crob}, \(\bar\eta\le1/32\), and the constant factor on
the right is larger than \(1/2\).  This proves
\eqref{eq:F-observation-energy}.

For the second claim, apply Lemma~\ref{lem:D-clean-mass} with
\(\gamma=\gamma_K\) and \(\eta=\bar\eta\).  Since
\(\gamma_K<1/2\),
\[
\Delta(\gamma_K,\bar\eta)
=
\gamma_K-(1-\gamma_K)\bar\eta
\ge
\gamma_K-\bar\eta.
\]
Using \eqref{eq:F-crob},
\[
\gamma_K-\bar\eta
\ge
\frac{c_{\rm BSG}}{\sqrt{2K}}
-
\frac{c_{\rm BSG}}{4\sqrt{2K}}
=
\frac{3c_{\rm BSG}}{4\sqrt{2K}}.
\]
The same lower bound holds for each compatible \(A\), proving
\eqref{eq:F-uniform-clean-mass}.
\end{proof}

\subsection{One observation-level event}

Run Theorem~\ref{thm:persistent-bsg-v2} on \(S=B\) with energy parameter
\(\alpha_K\) and soundness/completeness budgets \(\rho/4\).  By
Lemma~\ref{lem:F-observation-core}, the high-energy premise holds.  On the
intersection of the compiler soundness and completeness events, the
compiler returns a fixed record \(D\) defining \(Y_D\subseteq B\) with
\begin{equation}
|Y_D|\ge\gamma_K|B|,
\qquad
|Y_D+Y_D|\le L_K|Y_D|.
\label{eq:F-core-structure}
\end{equation}

Let
\[
s^\star_{\rm PFR}:=s_{\rm PFR}(n,\rho/4).
\]
Use Proposition~\ref{prop:output-exact-sampling-v2} to generate this
prescribed number of exact uniform samples from the same fixed \(Y_D\), with
bridge-failure budget \(\rho/4\).  Conditional on bridge success, the block
has law
\[
U_{Y_D}^{\otimes s^\star_{\rm PFR}}.
\]
Invoke Theorem~\ref{thm:size-oblivious-pfr-v2} on this fixed implicit set
with doubling parameter \(L_K\) and failure budget \(\rho/4\).  On its good
event the returned subspace \(H\) satisfies
\begin{equation}
|H|\le|Y_D|,
\qquad
\mathcal N_H(Y_D)\le P_{\rm PFR}(L_K).
\label{eq:F-PFR-output}
\end{equation}
Finally define \(V\) from the returned basis of \(H\) by the deterministic
truncation rule \eqref{eq:D-fixed-truncation}.

All random objects in this construction depend only on \(B\), \(K\), the
failure budget, and the learner's randomness.  In particular, no
distinguished \(A\in\mathcal I_{\rm rob}\) is selected.

\subsection{Proof of the moderate-corruption theorem}

\begin{proof}[Moderate-corruption clause of Theorem~\ref{thm:v3-query-common}]
Let \(\mathcal E_{\rm rob}\) be the intersection of the compiler
soundness/completeness events, success of the finite \(Y_D\)-sample bridge,
and the PFR success event.  Chronological conditioning and the allocated
budgets give
\[
\Pr[\mathcal E_{\rm rob}]\ge1-\rho.
\]

Fix a realization in \(\mathcal E_{\rm rob}\).  Equations
\eqref{eq:F-core-structure} and \eqref{eq:F-PFR-output} are now
deterministic statements about one fixed \(Y_D,H,V\).  Lemma
\ref{lem:F-observation-core} gives, for every
\(A\in\mathcal I_{\rm rob}\),
\[
|A\cap Y_D|\ge\delta_K|A|.
\]
Because \(|A\triangle B|\le\bar\eta|A|\) with \(\bar\eta<1\), the size
condition \(|B|<2|A|\) required by Lemma~\ref{lem:D-structural-lifting} is
automatic.  Apply that lemma with
\[
\delta=\delta_K,
\qquad
P=P_{\rm PFR}(L_K).
\]
The same realized \(V\) then satisfies, simultaneously for all
\(A\in\mathcal I_{\rm rob}\),
\begin{align}
q_A(V)
&=
\mathcal N_V(A)\\
&\le
\frac{2K}{\delta_K}P_{\rm PFR}(L_K)\\
&\le
C K^{3/2}
P_{\rm PFR}(16C_{\rm BSG}K^4)
=:P_{\rm rob}(K)
\label{eq:F-final-polynomial}
\end{align}
for an absolute constant \(C\).  Enlarging \(P_{\rm rob}\), if necessary,
makes it a nondecreasing polynomial.  Since the good event was fixed before
quantifying over \(A\), this proves exactly
\[
\Pr\!\left[
\forall A\in\mathcal I_{\rm rob}:q_A(V)\le P_{\rm rob}(K)
\right]\ge1-\rho,
\]
with no union bound over the identification set.

It remains to record complexity.  At \(\alpha_K=1/(2K)\), the persistent
compiler uses
\[
F_1(K)\operatorname{poly}\!\left(n,\log\frac1\rho\right)
\]
samples, queries, and time for a computable \(F_1\).  The retained fraction
\(\gamma_K=\Theta(K^{-1/2})\) makes the finite output-sampling bridge cost
only an additional
\[
F_2(K)\operatorname{poly}\!\left(n,\log\frac1\rho\right)
\]
samples and predicate evaluations.  The PFR call has doubling parameter
\(L_K=O(K^4)\), so Theorem~\ref{thm:size-oblivious-pfr-v2} contributes
\[
F_{\rm PFR}(O(K^4))
\operatorname{poly}\!\left(n,\log\frac1\rho\right).
\]
One evaluation of the fixed \(Y_D\)-predicate has parameter-dependent
worst-case cost by Theorem~\ref{thm:persistent-bsg-v2}.  Absorbing all
\(K\)-only factors into a computable \(F_{\rm rob}(K)\) proves the FPT
sample, query, and time bound in the theorem.
\end{proof}

\section{Sample-Only Common Learning: Full Proof}
\label{app:statistical-common}

This appendix supplies the detailed proof of Theorem~\ref{thm:v3-statistical}.
Let \(p=p(K)\) be the integer-valued clean PFR covering bound fixed in the
main text.

\begin{lemma}[Shared budget and majority capture]
\label{lem:v3-shared-budget}
For every \(A\in\mathcal I_{K,\eta}(B)\),
\[
 |B|\le(1+\eta)|A\cap B|.
\]
If \(Z\subseteq G\) satisfies \(|Z\cap B|\ge|B|/2\), then
\[
 |A\cap B\cap Z|\ge\frac{\varepsilon}{2}|A|.
\]
\end{lemma}
\begin{proof}
Set \(a=|A|\), \(c=|A\cap B|\), \(d=|A\setminus B|\), and
\(o=|B\setminus A|\).  From \(d+o\le\eta(c+d)\) we get
\(o\le\eta c-\varepsilon d\le\eta c\), proving the first assertion.  For
the second,
\[
 |A\cap B\cap Z|
 \ge |B\cap Z|-o
 \ge \frac{a-d+o}{2}-o
 =\frac{a-d-o}{2}
 \ge\frac{\varepsilon a}{2}.
\]
\end{proof}

\begin{lemma}[VC dimension of parallel affine unions]
\label{lem:v3-vc-cover-class}
The family of unions of at most \(p\) parallel affine subspaces of
\(\mathbb F_2^n\) has VC dimension
\(O((n+1)p\log(2p))\).
\end{lemma}
\begin{proof}
First consider all affine subspaces, allowing the empty set.  Any \(n+2\)
distinct points are affinely dependent, so some point lies in the affine
span of the others.  No affine subspace can contain all the other points
while excluding that point.  Thus the affine-subspace class has VC
dimension at most \(v=n+1\).  By Sauer's lemma its number of traces on
\(t\ge v\) points is at most \((et/v)^v\).  Unions of \(p\) arbitrary
affine subspaces have at most \((et/v)^{vp}\) traces, and parallel unions
form a subclass.  If \(t\) points are shattered, then
\(2^t\le(et/v)^{vp}\).  Writing \(x=t/(vp)\) gives
\(2^x\le epx\), which forces \(x=O(\log(2p))\).
\end{proof}

\begin{proof}[Proof of Theorem~\ref{thm:v3-statistical}]
Let
\[
 \mathcal Z_p
 =\left\{\bigcup_{j=1}^{p}(z_j+V):V\le G,\ z_j\in G\right\},
\]
where repeated representatives allow fewer than \(p\) cosets.  Set
\(\delta=\varepsilon/8\).  Draw the sample block prescribed in
\eqref{eq:v3-vc-samples}; standard VC uniform convergence and
Lemma~\ref{lem:v3-vc-cover-class} give, with probability at least \(1-\rho\),
\begin{equation}
\label{eq:v3-app-stat-uniform}
 \sup_{Z\in\mathcal Z_p}|\widehat P(Z)-U_B(Z)|\le\delta.
\end{equation}
On this event, among all candidate descriptions with
\(\widehat P(Z)\ge1/2+\delta\), choose one whose underlying subspace has
minimum dimension, with fixed tie-breaking.

For every compatible \(A\), clean PFR supplies a subspace \(V_A\) and a
union \(Z_A\) of at most \(p\) of its cosets such that
\[
 A\subseteq Z_A,
 \qquad |V_A|\le|A|.
\]
By Lemma~\ref{lem:v3-shared-budget},
\[
 U_B(Z_A)\ge\frac{|A\cap B|}{|B|}
 \ge\frac1{1+\eta}
 =\frac1{2-\varepsilon}
 \ge\frac12+\frac{\varepsilon}{4}.
\]
Hence every such witness is empirically feasible on
\eqref{eq:v3-app-stat-uniform}.  The minimum-dimension rule gives
\[
 |\widehat V|\le|V_A|\le|A|
 \qquad\forall A\in\mathcal I_{K,\eta}(B).
\]
For the selected union \(\widehat Z\), the same uniform event gives
\(U_B(\widehat Z)\ge1/2\).  Lemma~\ref{lem:v3-shared-budget} therefore yields
\[
 |A\cap B\cap\widehat Z|
 \ge\frac{\varepsilon}{2}|A|
 \qquad\forall A\in\mathcal I_{K,\eta}(B).
\]
One of the at most \(p\) \(\widehat V\)-cosets contains
\(C\subseteq A\cap B\) with \(|C|\ge\varepsilon|A|/(2p)\).  Since
\(C+C\subseteq\widehat V\) and \(|A+C|\le K|A|\), a maximal family of
disjoint translates \(t+C\), \(t\in A\), has size at most
\(|A+C|/|C|\) and maximality implies that the corresponding translates of
\(C+C\) cover \(A\).  Thus
\[
 \mathcal N_{\widehat V}(A)
 \le\frac{|A+C|}{|C|}
 \le\frac{2Kp}{\varepsilon}.
\]
The same sample event works for every compatible support, establishing the
joint guarantee.

For an elementary alternative to the VC bound, every subspace is the row
span of some \(n\times n\) binary matrix and the \(p\) representatives use
\(np\) bits.  Thus there are at most \(2^{n^2+np}\) descriptions, and
Hoeffding plus a union bound gives the sufficient count
\[
 O\!\left(\frac{n^2+np+\log(1/\rho)}{\varepsilon^2}\right).
\]
Exhaustive enumeration over these descriptions gives the finite
exponential-time implementation.

For sharpness, take \(G=\mathbb F_2^t\), \(B=\{0\}\), and
\(\varepsilon=2^{-t}\).  Both \(A_0=\{0\}\) and \(A_1=G\) are compatible
at \(K=1\).  Any common output must satisfy \(|V|\le|A_0|=1\), so
\(V=\{0\}\), and its cost on \(A_1\) is \(|G|=\varepsilon^{-1}\).  Conversely
\(V=\{0\}\) costs at most \(|G|\) on every subset of \(G\).  This proves
the displayed equality in the theorem.
\end{proof}

\paragraph{Population interpretation.}
The proof also gives a deterministic population rule over candidate
descriptions: among all \((V,z_1,\ldots,z_p)\) whose induced set
\(Z=\bigcup_j(z_j+V)\) satisfies \(U_B(Z)\ge1/2\), select a description
minimizing \(\dim V\).  The same argument gives cost
\(2Kp/\varepsilon\).  The sample-only learner is exactly a
uniform-convergence implementation of this population structural ERM.

\section{Conditional Sample-Only Hardness via Learning Parities with Noise}
\label{app:lpn-hardness}

We prove Theorem~\ref{thm:v3-lpn} by reducing search-LPN at the stated
noise rate to coarse quotient-cover learning. Blocking forces a bounded-cost
output to reveal a candidate secret, and a frozen-table coupling realizes
the required fixed-support observation model.

\paragraph{Search-LPN convention.}
For a dimension \(d\) and fixed \(0<\theta<1/2\), a secret
\(s\in\mathbb F_2^d\) is uniform, and the learner receives independent
samples
\[
 (X,\langle s,X\rangle+E),
 \qquad X\sim U_{\mathbb F_2^d},
 \quad E\sim\operatorname{Bernoulli}(\theta).
\]
This is the standard noisy-parity problem studied, for example, by Blum, Kalai, and Wasserman~\citep{blum2003noisetolerant}.  The assumption used here is that no randomized polynomial-time algorithm
recovers the uniform secret with constant success probability for all
sufficiently large \(d\).  The reduction below in fact succeeds uniformly
for every fixed secret, so this average-case convention is sufficient.

Fix the theorem's constant quality budget \(Q\), let
\[
 \ell=\lfloor\log_2Q\rfloor,
 \qquad k=\ell+1,
 \qquad D=kd,
 \qquad n=k(d+1),
\]
and define, for \(s\in\mathbb F_2^d\),
\begin{equation}
\label{eq:v3-app-blocked-support}
 A_s=
 \{(x_1,\ldots,x_k,
 \langle s,x_1\rangle,\ldots,\langle s,x_k\rangle):
 x_i\in\mathbb F_2^d\}.
\end{equation}
It is a \(D\)-dimensional linear subspace and hence has doubling one.

\begin{lemma}[Decoding a coarse subspace output]
\label{lem:v3-block-decoder}
From a basis of any \(V\le\mathbb F_2^n\) satisfying \(q_{A_s}(V)\le Q\),
one can construct in \(\operatorname{poly}(n,k,2^{k+\ell})\) time a list
of at most \(k2^{k+\ell}\le2kQ^2\) vectors of \(\mathbb F_2^d\) containing
\(s\).  On arbitrary inputs, discard any \(V\) with
\(\dim V\notin[D-\ell,D]\) before enumeration.
\end{lemma}
\begin{proof}
Because \(A_s\) is a subspace,
\[
 \mathcal N_V(A_s)=2^{D-\dim(A_s\cap V)}.
\]
The assumed cost gives \(\dim(A_s\cap V)\ge D-\ell\), while the size cap in
\(q_{A_s}\) gives \(\dim V\le D\).  Hence
\(\dim V^\perp\le k+\ell\) and
\[
 \dim(A_s^\perp\cap V^\perp)
 =n-\dim(A_s+V)
 \ge k-\ell=1.
\]
Directly from \eqref{eq:v3-app-blocked-support},
\[
 A_s^\perp
 =\{(a_1s,\ldots,a_ks,a_1,\ldots,a_k):a\in\mathbb F_2^k\}.
\]
Enumerate the at most \(2^{k+\ell}\) vectors
\(w=(b_1,\ldots,b_k,a_1,\ldots,a_k)\in V^\perp\), and insert \(b_j\) in
the candidate list whenever \(a_j=1\).  A nonzero vector in the displayed
intersection has some \(a_j=1\), and then \(b_j=s\).
\end{proof}

\begin{lemma}[Frozen-table coupling]
\label{lem:v3-frozen-table}
Set \(\theta=\eta_0/(8k)\).  For every complete input tuple
\(x=(x_1,\ldots,x_k)\in\mathbb F_2^D\), independently draw and freeze
\(E(x)\in\mathbb F_2^k\) with independent Bernoulli\((\theta)\) coordinates,
and let \(B_E\) be the resulting noisy graph.  Then
\[
 |B_E|=|A_s|=2^D,
\]
\[
 \Pr_E[|A_s\triangle B_E|>\eta_0|A_s|]
 \le \exp(-9\eta_0^2 2^D/32),
\]
and for any adaptive sample-only algorithm making at most \(T\) requests,
the transcript law under fresh grouped LPN samples and the law obtained by
first drawing \(E\) and then sampling independently from \(U_{B_E}\) have
total variation distance at most
\[
 \binom{T}{2}2^{-D}.
\]
The frozen table is not conditioned on satisfying the support promise.
\end{lemma}
\begin{proof}
Different complete inputs produce different graph points.  The two graphs
differ at input \(x\) exactly when \(E(x)\neq0\), and each changed input
contributes two points to the symmetric difference.  The changed-input
indicators are independent with mean
\[
 a=1-(1-\theta)^k\le k\theta=\eta_0/8.
\]
The support promise holds whenever their empirical mean over \(2^D\) inputs
is at most \(\eta_0/2\).  Hoeffding with gap \(3\eta_0/8\) gives the stated
invalid-table probability.

For the coupling, use the same independent uniform complete input tuples in
both experiments.  On the first visit to a tuple, use its fresh LPN noise to
reveal the previously unexposed frozen-table entry.  Until a complete tuple
repeats, the two streams coincide.  A repetition among \(T\) tuples has
probability at most \({T\choose2}/2^D\).  Unexposed entries can then be filled
independently, giving exactly the advance-sampled table distribution.  Pad
adaptively stopped executions to \(T\) requests and couple internal
randomness as well.
\end{proof}

\begin{lemma}[Uniqueness at doubling one through the critical radius]
\label{lem:v3-k1-uniqueness}
If \(0\le\eta\le1/3\), then \(\mathcal I_{1,\eta}(B)\) contains at most one
support.
\end{lemma}
\begin{proof}
Every nonempty doubling-one set in \(\mathbb F_2^n\) is an affine subspace.
For distinct affine subspaces \(A,A'\) of the same size \(a\), their
intersection has size at most \(a/2\), so
\(|A\triangle A'|\ge a=(|A|+|A'|)/2\).  If their sizes differ, assume
\(|A'|=b\ge2a=2|A|\); then
\[
 |A\triangle A'|\ge b-a\ge(a+b)/3.
\]
Thus simultaneous compatibility is impossible for \(\eta<1/3\).  At
\(\eta=1/3\), equality could only occur for a nested pair of sizes \(a\)
and \(2a\).  Their distance is \(a\), but the two integer discrepancy
budgets sum to
\(\lfloor a/3\rfloor+\lfloor2a/3\rfloor=a-1\), since \(a\) is a power of
two.  The triangle inequality again gives a contradiction.
\end{proof}

\begin{proof}[Proof of Theorem~\ref{thm:v3-lpn}]
Assume for contradiction a polynomial-time sample-only learner with the
stated guarantee.  Given search-LPN samples for secret \(s\in\mathbb F_2^d\),
group them in blocks of \(k\), thereby producing samples in the ambient
space of \eqref{eq:v3-app-blocked-support}.  If size information is allowed,
supply \(\dim A_s=kd\) and \(|A_s|=2^{kd}\).  Let \(T(n)\) be a polynomial
upper bound on the learner's sample requests.

In the frozen-table experiment, every table satisfying the support budget is
a legitimate fixed instance for the learner.  Average the learner's
constant success guarantee over such tables without conditioning the
simulator on validity.  Lemma~\ref{lem:v3-frozen-table} shows that the
LPN-driven simulation returns a serving subspace with probability at least
\[
 c_0-
 \exp(-9\eta_0^2 2^{kd}/32)
 -{T(n)\choose2}2^{-kd},
\]
where \(c_0>1/2\) can be obtained, if necessary, by running the learner an
absolute constant number of independent times and merging the returned
structural lists.  If one run succeeds with probability \(c>0\), then after
\(r=O(1)\) runs the merged list contains a serving entry with probability
\(1-(1-c)^r\); its length remains polynomial, and this amplification does
not require identifying which entry serves the latent support.

For every output subspace, run Lemma~\ref{lem:v3-block-decoder}.  A
polynomial-size structural list therefore yields a polynomial-size union of
candidate secrets because \(Q\) and \(k\) are fixed.  Validate all candidates
on one fresh LPN block.  The true secret has mismatch probability \(\theta\),
whereas every false candidate has mismatch probability exactly \(1/2\).
The constant gap \(1/2-\theta\), Hoeffding's inequality, and a union bound
over the polynomial candidate set recover the true secret with constant
probability for all sufficiently large \(d\).  This contradicts the stated
search-LPN assumption.

Finally, \(\eta_0<1/3\) and Lemma~\ref{lem:v3-k1-uniqueness} show that every
valid observation used in the theorem has at most one compatible
K=1 support.  The reduction therefore does not rely on latent
non-identifiability.
\end{proof}

\paragraph{Assumption boundary.}
The reduction invokes search-LPN at
\(\theta=\eta_0/(8k)\), as specified in Theorem~\ref{thm:v3-lpn}.
Using an assumption at another rate would require a noise-rate reduction.

\section{Exact Affine-Support Recovery at \(K=1\)}
\label{app:exact-k1}
The following postprocessing result uses Fourier learning after finding an
explicit affine flat in which the latent support has noticeable density.
The output is the entire affine subspace, specified by an offset and a
basis.

\begin{lemma}[Heavy-Fourier primitive]\label{lem:v3-fourier-primitive}
For a fixed $f:\F^r\to\{0,1\}$, define
$\widehat f(\xi)=\mathbb E_z f(z)(-1)^{\langle\xi,z\rangle}$.
Given value-query access to $f$, a threshold $0<\tau<1$, and a failure
budget $0<\zeta<1$, one can in
\[
\operatorname{poly}(r,\tau^{-1},\log(1/\zeta))
\]
time and queries produce a list containing every $\xi$ with
$|\widehat f(\xi)|\ge\tau$ and containing at most $O(\tau^{-2})$ entries,
except with probability $\zeta$.
The routine can have deterministic finite budgets on every execution.
\end{lemma}
This is the standard Goldreich--Levin interface \citep{goldreich1989hardcore}. For finite
budgets, estimate squared Fourier prefix masses to error $\tau^2/4$, retain
estimates above $\tau^2/2$, and cap each level at $\lceil4/\tau^2\rceil$.
Parseval controls the retained count and heavy prefixes survive. For a
length-$j$ prefix $a$, the bounded unbiased estimator is
$f(u,w)f(v,w)(-1)^{\langle a,u+v\rangle}$, for independent uniform
$u,v\in\F^j$, $w\in\F^{r-j}$. Fixed estimation slots and the cap bound
all executions. Dimension zero is evaluated directly.

\begin{theorem}[From a coarse list to exact recovery]\label{thm:v3-exact-decoder}
Let $A\subseteq\F^n$ be a nonempty affine subspace and let
$|A\triangle B|\le\eta|A|$ with a supplied $0\le\eta<1/3$. Suppose one
is given $m\ge1$ explicitly based subspaces $V_1,\ldots,V_m$, at least one
of which satisfies $q_A(V_i)\le Q$, where $Q\ge1$ is supplied. The other
entries need not satisfy any guarantee. Given exact uniform samples and
membership queries to $B$, a randomized postprocessor recovers $A$ exactly
with probability at least $1-\rho$, in
\[
 \operatorname{poly}\left(n,m,Q,(1-3\eta)^{-1},\log(1/\rho)\right)
\]
samples, membership queries, and time. The same guarantee holds after an
adaptive past fixing the list, using fresh randomness thereafter.
\end{theorem}

\begin{proof}
Write $a=|A|$, $C=A\cap B$, and $H=A-A$. All stages below have prescribed
finite budgets. Set $\ell=\lfloor\log_2Q\rfloor$,
$\kappa=1/2-\eta>1/6$, and
\[
 t=\left\lceil\frac{8(\ell+\ln4)}\kappa\right\rceil,
 R=\left\lceil2(1+\eta)^{t+1}\ln\frac4\rho\right\rceil.
\]
For each $V_i$, make $R$ independent proposals using $t+1$ fresh samples
$X_0,\ldots,X_t$ from $B$, and form the affine flat
\[
 W=X_0+\operatorname{span}(V_i,X_1-X_0,\ldots,X_t-X_0).
\]
All proposals are explicit affine flats, so membership and uniform
sampling within them are elementary linear-algebra operations.

\paragraph{A useful proposal occurs with noticeable probability.}
Fix a serving entry $V$. The image of $A$ in $G/V$ is an affine subspace
of dimension at most $\ell$. Conditional on all proposal samples lying in
$C$, they are independent uniform samples from $C$. Any proper affine
subspace of this image has a preimage in $A$ of size at most $a/2$.
Since $|C|\ge(1-\eta)a$, a new conditional sample lies outside a currently
proper affine span with probability at least
\[
 \frac{|C|-a/2}{|C|}\ge\frac12-\eta=\kappa.
\]
There are at most $\ell$ affine-dimension increases to make after the
first sample.  Until the image is spanned, each new conditional sample
increases its affine dimension with conditional probability at least
$\kappa$.  Couple each not-yet-spanned step with an independent
$\operatorname{Bernoulli}(\kappa)$ success that is dominated by the actual
increase event.  Then failure to accumulate the required $\ell$ increases
is contained in the event
$\{\operatorname{Binomial}(t,\kappa)<\ell\}$.  Since
\[
 t\kappa\ge 8(\ell+\ln4),
\]
a lower-tail Chernoff bound gives
\[
 \Pr[\text{fewer than }\ell\text{ increases}]
 \le \Pr[\operatorname{Binomial}(t,\kappa)<\ell]
 <\frac12
\]
(with the case $\ell=0$ immediate).  Thus the projected samples span the
entire image with probability at least $1/2$.  The probability all samples
are clean is at least $(1+\eta)^{-(t+1)}$. Thus a single proposal equals
$A+V$ with probability
at least $\frac12(1+\eta)^{-(t+1)}$. Such a proposal contains $A$ and has
\[
 |W|=|A+V|=\N_V(A)|V|\le Qa.
\]
The probability of missing all such proposals for the serving entry is
at most $\rho/4$. Since $\eta<1/3$ and $\kappa>1/6$, $R$ is polynomial
in $Q$ and $\log(1/\rho)$ with absolute exponents.

\paragraph{A sample test certifies containment.}
For any proposed affine flat $W$, let $b_W=\Prb_{U_B}[X\in W]$.
If $A\subseteq W$, then
\[
 b_W\ge u:=\frac1{1+\eta}.
\]
If $A\not\subseteq W$, then $|A\cap W|\le a/2$. Writing
$d=|A\setminus B|$, $o=|B\setminus A|$, one obtains
\[
 b_W\le\frac{a/2+o}{a-d+o}\le
 v:=\frac1{2(1-\eta)}.
\]
The second inequality follows from $d+o\le\eta a$; maximizing the ratio
sets $d=\eta a-o$, after which the ratio is nonincreasing in $o$.
The gap is
\[
 u-v=\frac{1-3\eta}{2(1-\eta^2)}>0.
\]
Estimate $b_W$ with fresh samples from $B$ to accuracy $(u-v)/4$, accepting
only if the estimate is at least $(u+v)/2$. On the simultaneous test event,
every accepted flat contains $A$, and every flat containing $A$ is accepted.

\paragraph{A second test certifies density.}
Put $\beta=1/Q$. Using uniform points of the explicit flat $W$ and
membership queries to $B$, estimate
$\mu=|B\cap W|/|W|$ to accuracy $\beta/16$. Retain a flat passing the
containment test only if this estimate is at least $\beta/2$.
A useful proposal with $A\subseteq W$ and $|W|\le Qa$ has
$\mu\ge(1-\eta)\beta>2\beta/3$, so it passes. Every retained flat has
$\mu\ge7\beta/16$; since it contains $A$, its latent density
$p=a/|W|$ obeys
\[
 p\ge\frac{\mu}{1+\eta}>\frac1{4Q}=:p_{\min}.
\]
This test handles oversized nonserving entries in the input list. Allocate
a total failure budget $\rho/4$ to all containment and density estimates
over the at most $mR$ proposals. Take the first flat passing both tests;
if there is none, return failure.

\paragraph{A relative Fourier threshold recovers all affine equations.}
Parameterize the retained flat by an affine bijection
$\psi:\F^r\to W$, and let
$A'=\psi^{-1}(A)=b+H'$ and $f=\mathbf1_B\circ\psi$. Then
\[
 \widehat{\mathbf1_{A'}}(\xi)=
 \begin{cases}p(-1)^{\langle\xi,b\rangle},&\xi\in(H')^\perp,\\
 0,&\xi\notin(H')^\perp.
 \end{cases}
\]
The support discrepancy gives
$|\widehat f(\xi)-\widehat{\mathbf1_{A'}}(\xi)|\le\eta p$ for every
$\xi$, while $(1-\eta)p\le\mu\le(1+\eta)p$. Consequently
\begin{equation}\label{eq:relative-fourier-gap}
 \begin{split}
 \xi\in(H')^\perp&\ \Longrightarrow\
 |\widehat f(\xi)|-\mu/2\ge(1-3\eta)p/2,\\
 \xi\notin(H')^\perp&\ \Longrightarrow\
 \mu/2-|\widehat f(\xi)|\ge(1-3\eta)p/2.
 \end{split}
\end{equation}
Run Lemma~\ref{lem:v3-fourier-primitive} with threshold $\tau=1/(8Q)$ and
error budget $\rho/4$. Every frequency in $(H')^\perp$ has magnitude
at least $(1-\eta)p>1/(6Q)>\tau$, so all are listed. Add the zero
frequency explicitly. On a fresh block, estimate $\mu$ and all listed
coefficients to accuracy $e=(1-3\eta)/(64Q)$, with total error budget
$\rho/4$. Retain exactly those whose estimated absolute coefficient is
greater than half the estimated $\mu$. The gap in
\eqref{eq:relative-fourier-gap} is at least $8e$, whereas estimation moves
this comparison by at most $3e/2$. The retained set is therefore exactly
$(H')^\perp$.  For a retained nonzero coefficient, the latent coefficient
has magnitude $p$ and the support discrepancy perturbs it by at most
$\eta p<p$; hence the observed coefficient already has the same sign as
$\widehat{\mathbf1_{A'}}(\xi)$.  Its magnitude is at least
$(1-\eta)p>1/(6Q)$, whereas the final estimation error is at most
$e<1/(64Q)$, so the empirical coefficient has that same sign as well.
Choose a basis $\xi_1,\ldots,\xi_c$ of this set and solve
\[
 \langle\xi_j,z\rangle=b_j,
 \qquad (-1)^{b_j}=\operatorname{sign}\bigl(\widetilde{\widehat f}(\xi_j)\bigr),
\]
where the tilde denotes the empirical coefficient estimate.
The solution affine space is exactly $A'$, giving $A=\psi(A')$.
Inconsistent systems on unsuccessful executions are discarded.

The four failure contributions are useful-proposal absence, proposal-test
errors, Fourier-list failure, and coefficient-estimation errors, each at
most $\rho/4$. The containment and density tests use fixed error margins;
Fourier search uses the fixed lower density $p_{\min}$ rather than an
unverified estimated reciprocal. All stages therefore have the displayed
polynomial worst-case budgets on every transcript. Conditioning on an
earlier fixed list changes none of these arguments.
\end{proof}

\begin{corollary}[Exact recovery for every fixed radius below $1/3$]
\label{cor:v3-exact-query}
For $K=1$, $0\le\eta<1/3$, and nonempty $\I_{1,\eta}(B)$, exact recovery of
its unique affine support is possible from $\mathsf{SampMQ}(B)$ with
probability $1-\rho$, using
\[
 \operatorname{poly}\bigl(n,(1-3\eta)^{-1},\log(1/\rho)\bigr)
\]
samples, queries, and time.
\end{corollary}
\begin{proof}
Invoke the arbitrary-corruption clause of
Theorem~\ref{thm:v3-query-common} at fixed parameters $K=1$ and radius
$1/3$, with error budget $\rho/2$. The true $A$ is compatible at that
larger radius. At these fixed parameters the returned common subspace has
an absolute serving-quality bound, and the XP resource bound is polynomial
with an absolute exponent. Feed the singleton list containing this subspace
to Theorem~\ref{thm:v3-exact-decoder}, using the supplied radius
$\eta<1/3$ and remaining error budget $\rho/2$.
\end{proof}

\subsection{The critical radius: uniqueness with exponential access cost}
\begin{theorem}[Unconditional critical-radius access lower bound]
\label{thm:v3-critical-access}
At $K=1$ and $\eta=1/3$, every nonempty identification set is a singleton,
but a randomized classical algorithm that exactly recovers its member with
probability at least $2/3$ on every instance requires $\Omega(2^n)$ total
uniform-sample and membership-query calls in the worst case. This bound
holds without any restriction on internal computation. Enumerating $B$
with $2^n$ membership queries suffices with unbounded subsequent computation.
\end{theorem}
\begin{proof}
Uniqueness is Lemma~\ref{lem:v3-k1-uniqueness}. Let $a=2^{2t}=3r+1$, choose nested
linear subspaces $U\subset H$ of sizes $a,2a$, and fix any
$T\subset H\setminus U$ of size $r=(a-1)/3$. Put
\[
 B_0=U\cup T,\qquad Z=H\setminus B_0,\qquad
 B_z=B_0\cup\{z\}\quad(z\in Z).
\]
For $B_0$, the unique latent support is $U$, because
$|B_0\triangle U|=r\le a/3$. For each $B_z$, it is $H$, because
$|B_z\triangle H|=a-r-1\le2a/3$.

Even disclose $U,H,T$ to the learner. Compare the $B_0$ experiment with
the experiment that first chooses $z$ uniformly in $Z$ and then provides
access to $B_z$. Suppose there are at most $q$ membership queries and
$s$ sample requests. Couple internal randomness and all replies until
either a membership query hits $z$ or a sample returns $z$. Under the
$B_0$ transcript, queried points are independent of the initially uniform
hidden $z$, so the former event has probability at most $q/|Z|$. A sample
from $B_z$ equals $z$ with probability $1/(|B_0|+1)$ and otherwise is
uniform in $B_0$, so the latter event has probability at most
$s/(|B_0|+1)$. Consequently the transcript total variation distance is
at most
\[
 \frac q{|Z|}+\frac s{|B_0|+1}.
\]
Exact recovery with probability $2/3$ on both experiments distinguishes
them with a probability gap at least $1/3$. Hence this displayed quantity
is at least $1/3$. Since $|Z|=(2a+1)/3$ and $|B_0|+1=(4a+2)/3$,
the total access budget is $\Omega(a)$.
For $n=2t+1$, $a=2^{n-1}$. Padding this known construction by at most one
zero coordinate gives the same $\Omega(2^n)$ worst-case bound in the
remaining dimensions. The argument applies to adaptive requests by using
the hypothetical $B_0$ execution to define the queried set and padding
to deterministic call caps. A cap on the total number of calls also
gives such individual caps and the same asymptotic lower bound.

For the upper bound, query every ambient point to determine $B$, enumerate
affine subspaces, and select the unique compatible one. Computation is
unrestricted in this upper bound.
\end{proof}
\begin{remark}[Role of latent cardinality]
The critical construction uses different latent sizes $a$ and $2a$.
It applies to the present unknown-cardinality model, and does not
prove the same lower bound when the true cardinality is supplied.
This restriction differs from the LPN reduction, which remains valid
even when latent cardinality is known.
\end{remark}

\begin{theorem}[Non-identifiability above $1/3$]\label{thm:v3-above-critical}
For every $\eta>1/3$, there are two different doubling-one supports
compatible with the same observation at radius $\eta$. Thus exact
recovery cannot succeed with probability greater than $1/2$ on every
instance, even with full knowledge of $B$.
\end{theorem}
\begin{proof}
Take nested $U\subset H$ of sizes $a,2a$, with $a$ a sufficiently large
power of two. Choose an integer $j$ nearest to $a/3$, a set
$T\subset H\setminus U$ of size $j$, and $B=U\cup T$. The two relative
discrepancies are $j/a$ and $(a-j)/(2a)$, both tending to $1/3$ as
$a\to\infty$. For large enough $a$ they are below the supplied $\eta$.
The two instances give identical observations but require different
exact outputs, so their two success probabilities sum to at most one.
\end{proof}

\begin{proof}[Proof of Proposition~\ref{prop:v3-exact-boundary}]
Item 1 is Corollary~\ref{cor:v3-exact-query}. Item 2 is
Theorem~\ref{thm:v3-critical-access}, with uniqueness supplied by
Lemma~\ref{lem:v3-k1-uniqueness}. Item 3 is
Theorem~\ref{thm:v3-above-critical}. At the fixed critical parameters,
Theorem~\ref{thm:v3-query-common} has constant structural parameters and
therefore polynomial running time, giving the coarse common-output contrast
stated in the proposition.
\end{proof}

\paragraph{Exact versus coarse recovery at the same critical radius.}
At $K=1$ and $\eta=1/3$,
Theorem~\ref{thm:v3-query-common} computes one common
bounded-cost subspace in polynomial time, whereas
Theorem~\ref{thm:v3-critical-access} requires exponentially
many accesses for exact recovery when latent cardinality
is unknown. This contrast does not rely on the LPN
assumption.

\subsection{Exact-support lists and the value of latent cardinality}

\begin{lemma}[At most two affine explanations below one half]
\label{lem:v3-two-explanations}
For $0\le\eta<1/2$, the class $\I_{1,\eta}(B)$ contains at most two
members, and at most one member of any given cardinality.
\end{lemma}
\begin{proof}
Two different affine subspaces of the same size $a$ intersect in at most
$a/2$ points. Their symmetric difference is at least $a$, whereas
compatibility with $B$ would bound it by $2\eta a<a$.
Every compatible size lies in
\[
 \frac{|B|}{1+\eta}\le |A|\le\frac{|B|}{1-\eta}.
\]
The ratio of the endpoints is less than $3$. As affine sizes are powers
of two, this interval contains at most two possible sizes. Combining the
two observations proves the assertion.
\end{proof}

\begin{theorem}[Polynomial-time exact-support lists of size two]
\label{thm:v3-two-support-list}
Let $0\le\eta<1/2$, $\delta=1/2-\eta$, and
$\bar\eta=\eta+\delta/2$. Given $\mathsf{SampMQ}(B)$ with
$\I_{1,\eta}(B)\ne\varnothing$, one can in
\[
 \operatorname{poly}\bigl(n,\delta^{-1},\log(1/\rho)\bigr)
\]
time and oracle calls return a list $\mathcal C$ of at most two explicitly
represented affine subspaces such that, with probability $1-\rho$,
\[
 \I_{1,\eta}(B)\subseteq\mathcal C
 \subseteq\I_{1,\bar\eta}(B).
\]
The algorithm has deterministic finite budgets, returning failure on
exceptional executions if necessary. The polynomial has absolute
exponents, although they can depend on the absolute quality constant
from the fixed-parameter invocation of Theorem~\ref{thm:v3-query-common}.
\end{theorem}
\begin{proof}
Invoke the arbitrary-corruption clause of
Theorem~\ref{thm:v3-query-common} at $K=1$, radius $1/2$, and failure budget
$\rho/4$. At these fixed parameters it returns one common subspace $V_0$
with an absolute integer serving-quality bound $Q_0$ for every support in
$\I_{1,1/2}(B)$, and its XP resource bound is polynomial with an absolute
exponent. Treat $\{V_0\}$ as a singleton list. Reject a failed invocation
before any downstream work, so the prescribed slot bounds apply on every
execution. All subsequent randomness is fresh, conditional on this fixed
subspace.

\paragraph{Propose a containing flat for every true explanation.}
Put $\ell=\lfloor\log_2 Q_0\rfloor$ and
\[
 p_\star=(1+\eta)^{-(\ell+1)}\delta^\ell,
 \qquad R=\left\lceil p_\star^{-1}\ln\frac8\rho\right\rceil.
\]
For the common entry $V_0$, make $R$ proposals from $\ell+1$ fresh samples:
\[
 W=X_0+\operatorname{span}(V_0,X_1-X_0,\ldots,X_\ell-X_0).
\]
For a fixed true $A$, the common $V_0$ serves it, and all these samples are in
$C=A\cap B$ with probability at least $(1+\eta)^{-(\ell+1)}$.
Conditional on this event they are independent and uniform in $C$.
Any proper affine subspace of the image of $A$ modulo $V_0$ contains at
most $|A|/2$ preimage points. Until that image is spanned, a fresh clean
sample increases the affine dimension with probability at least
$1/2-\eta=\delta$. Its dimension is at most $\ell$, so the chance of
spanning it with these samples is at least $\delta^\ell$. Consequently
a proposal equals $A+V_0$, contains $A$, and has size at most $Q_0|A|$
with probability at least $p_\star$. By Lemma~\ref{lem:v3-two-explanations},
there are at most two true explanations. The chance of missing a useful
proposal for either is at most $2\exp(-Rp_\star)\le\rho/4$.

\paragraph{Guess the latent size and decode within each flat.}
For every proposed $W$, try all $a=2^d$, $0\le d\le n$, satisfying
$a\le |W|\le Q_0a$. Here $a$ is a guess, not supplied side information.
Let $p=a/|W|\ge1/Q_0$, and parameterize $W$ by an affine bijection $\psi$.
Use the heavy-Fourier primitive on $f=\mathbf1_B\circ\psi$ with threshold
$1/(4Q_0)$, and estimate its listed coefficients to accuracy
$\delta/(4Q_0)$. Retain estimated magnitudes above $p/2$.
Take a basis of the retained frequencies, solve the affine equations
specified by their estimated signs, and check all retained equations.
Discard inconsistent outputs and outputs whose cardinality differs from
the guess $a$. Map any remaining affine solution back to the ambient space.
This produces at most one candidate per flat and size guess.

For a useful flat and the true guess, writing $\psi^{-1}(A)=b+H'$, one has
\[
 \xi\in(H')^\perp\ \Longrightarrow\ |\widehat f(\xi)|\ge(1-\eta)p,
 \qquad
 \xi\notin(H')^\perp\ \Longrightarrow\ |\widehat f(\xi)|\le\eta p.
\]
The threshold $p/2$ has margin at least $\delta p\ge\delta/Q_0$,
and every required frequency has magnitude at least $p/2>1/(4Q_0)$.
Thus all required frequencies are listed and correctly selected, their
signs are correct, and this branch produces exactly $A$.
Allocate total failure probability $\rho/4$ to all Fourier calls and
coefficient estimates over their prescribed finite slots. Incorrect size
guesses and noncontaining flats need not have any Fourier guarantee
beyond producing an explicit candidate of the guessed size.

\paragraph{Validate with a small radius slack.}
For each produced candidate $D$, independently estimate
\[
 u_D=\Prb_{U_D}[X\in B],\qquad
 v_D=\Prb_{U_B}[X\in D]
\]
to additive accuracy $e=\delta/128$. The first estimate uses membership
queries to $B$ and uniform points of the explicit affine space $D$;
the second uses samples from $B$. Allocate total error probability
$\rho/4$ across these estimates. Reject if $\widehat v_D<1/2$.
Otherwise compute
\[
 \widehat r_D=1+\widehat u_D/\widehat v_D-2\widehat u_D,
\]
and accept only if $\widehat r_D\le\eta+\delta/4$.
When $v_D>0$, the exact identity is
\[
 \frac{|D\triangle B|}{|D|}
 =1+\frac{u_D}{v_D}-2u_D.
\]
On the simultaneous estimation event, the gate guarantees
$v_D\ge1/2-e$. Since $\delta\le1/2$, elementary bounds on the ratio give
\[
 \left|\widehat r_D-\frac{|D\triangle B|}{|D|}\right|
 \le 10e<\delta/8.
\]
Every true $A$ passes: its $v_A\ge1/(1+\eta)>2/3$, and its discrepancy
ratio is at most $\eta$. Every accepted candidate has discrepancy ratio
at most $\eta+3\delta/8<\bar\eta<1/2$.
Deduplicate the accepted affine subspaces. Lemma~\ref{lem:v3-two-explanations}
at radius $\bar\eta$ bounds their number by two. If more than two survive
on an exceptional execution, return failure rather than an oversized list.
The four error allocations sum to $\rho$.

The constants $Q_0,\ell$ are absolute. Therefore
$R=O(\delta^{-\ell}\log(1/\rho))$, all Fourier thresholds are constant,
and all estimation margins are inverse polynomial in $\delta^{-1}$.
The number of size guesses is at most $n+1$. Together with the original
fixed-parameter invocation, this proves the claimed worst-case budgets.
\end{proof}

\begin{corollary}[Known cardinality moves exact recovery to one half]
\label{cor:v3-known-size}
Under the same structural algorithm input, if the true cardinality
$a=|A|$ is supplied, then exact recovery is possible for every
$0\le\eta<1/2$ in
$\operatorname{poly}(n,(1-2\eta)^{-1},\log(1/\rho))$ time and oracle
calls. Restrict the algorithm's guesses and outputs to size $a$;
Lemma~\ref{lem:v3-two-explanations} ensures that every accepted candidate,
including those valid only at the relaxed radius, equals the true $A$.
\end{corollary}

\begin{theorem}[Exponential exact-support ambiguity at one half]
\label{thm:v3-half-ambiguity}
At $K=1$, $\eta=1/2$, even when the latent cardinality is known to be two,
one observation has $2^n-1$ compatible latent supports. Hence a list
covering all explanations must have exponential length. More generally,
a randomized list of length at most $L$ that includes each possible true
support with probability at least $1-\rho$ must satisfy
$L\ge(1-\rho)(2^n-1)$.
\end{theorem}
\begin{proof}
Take $B=\{0\}$. For every nonzero $v\in\F^n$, the linear subspace
$A_v=\{0,v\}$ has size two and discrepancy $|A_v\triangle B|=1$.
These $2^n-1$ supports are distinct. Summing the inclusion probabilities
gives the randomized-list lower bound. The same construction applies at
any larger radius.
\end{proof}

\paragraph{What the two-output result says at the critical point.}
At $\eta=1/3$ there is a unique exactly compatible support, yet the
algorithm can return two candidates because its validity radius is slightly
larger. It includes the exact support but need not identify which entry
has discrepancy at most the original $1/3$ budget. This is consistent with
Theorem~\ref{thm:v3-critical-access}: constructing two exact-support candidates is
polynomial time, whereas selecting the unique correct entry can require
exponential oracle access. For any fixed $\eta\in(1/3,1/2)$, the nested
construction from Theorem~\ref{thm:v3-above-critical} also shows that two
candidates are sometimes information-theoretically necessary.

\section{Ambiguity Beyond Latent Cardinality}
\label{app:ambiguity}
\label{app:ambiguity-v2}

The sharp example in Theorem~\ref{thm:v3-statistical} uses the compatible supports
\(\{0\}\) and \(G\), whose cardinalities are very different.  This appendix
shows that the linear corruption cost is not merely a consequence of
uncertainty about latent size.  A prescribed family of equally large
subspaces can force the same scale through independent unseen directions.
The qualifier ``prescribed family'' is important: for the observation below,
\(B\) itself also belongs to the full identification set and can impose a
smaller admissible subspace size than the equal-size family alone.

\begin{theorem}[Equal-size directional ambiguity]
\label{thm:equal-size-directional-ambiguity}
\label{thm:equal-size-directional-ambiguity-v2}
Let
\[
G=U\oplus X_1\oplus\cdots\oplus X_r
\]
be an internal direct sum over \(\mathbb F_2\), where
\(d:=\dim U\), \(\dim X_i=t\ge1\), and \(r\ge1\).  Define
\[
H_i:=U\oplus X_i,
\qquad
B:=U.
\]
Then every \(H_i\) has doubling one, all \(H_i\) have the same
cardinality, and
\[
|H_i\triangle B|
=(1-2^{-t})|H_i|.
\]
Moreover,
\begin{equation}
\min_{\substack{V\le G\\ \dim V\le d+t}}
\max_{1\le i\le r}\mathcal N_V(H_i)
=
2^{\,t-\lfloor t/r\rfloor}.
\label{eq:G-equal-size-tradeoff}
\end{equation}
In particular, if \(r>t\), then the optimum equals
\[
2^t=(1-\eta)^{-1},
\qquad
\eta:=1-2^{-t},
\]
even though all prescribed latent explanations have identical size.
\end{theorem}

\begin{proof}
The first assertions are immediate: each \(H_i\) is a subspace of
cardinality \(2^{d+t}\), while \(B=U\subseteq H_i\) has cardinality
\(2^d\), so
\[
\frac{|H_i\triangle B|}{|H_i|}
=1-\frac{|U|}{|H_i|}
=1-2^{-t}.
\]

Fix a subspace \(V\le G\) with \(\dim V\le d+t\).  Put
\[
J_i:=V\cap H_i,
\qquad
u:=\dim(V\cap U).
\]
The quotients \(J_i/(V\cap U)\) form an internal direct sum.  Indeed, if
\(x_i\in J_i\) and \(\sum_i x_i\in U\), projection onto the independent
coordinate \(X_i\) forces the \(X_i\)-component of each \(x_i\) to vanish,
so every \(x_i\in U\).  Consequently,
\begin{equation}
\sum_{i=1}^r \dim J_i-(r-1)u
=
\dim\!\left(\sum_{i=1}^r J_i\right)
\le \dim V
\le d+t.
\label{eq:G-intersection-dimension}
\end{equation}
Set
\[
c_i:=d+t-\dim J_i.
\]
Then \eqref{eq:G-intersection-dimension} gives
\begin{align*}
\sum_{i=1}^r c_i
&=r(d+t)-\sum_i\dim J_i\\
&\ge (r-1)(d+t-u)
\ge (r-1)t.
\end{align*}
Hence
\[
\max_i c_i
\ge
\left\lceil\frac{(r-1)t}{r}\right\rceil
=
t-\left\lfloor\frac tr\right\rfloor.
\]
Because \(H_i\) is a subspace,
\[
\mathcal N_V(H_i)
=|\pi_V(H_i)|
=2^{\dim H_i-\dim(V\cap H_i)}
=2^{c_i}.
\]
This proves the lower bound in \eqref{eq:G-equal-size-tradeoff}.

For the matching construction, choose integers
\(s_1,\ldots,s_r\) summing to \(t\), each equal to either
\(\lfloor t/r\rfloor\) or \(\lceil t/r\rceil\).  Choose subspaces
\(X_i'\le X_i\) with \(\dim X_i'=s_i\), and put
\[
V:=U\oplus\bigoplus_{i=1}^r X_i'.
\]
Then \(\dim V=d+t\), while
\[
V\cap H_i=U\oplus X_i',
\qquad
\mathcal N_V(H_i)=2^{t-s_i}.
\]
Therefore
\[
\max_i\mathcal N_V(H_i)
=2^{t-\min_i s_i}
=2^{\,t-\lfloor t/r\rfloor},
\]
which proves equality.
\end{proof}

Theorem~\ref{thm:equal-size-directional-ambiguity} isolates a different
source of ambiguity from the two-scale example used for the exact common
radius.  Here latent cardinality is fixed; what remains hidden is how the
support extends along mutually independent quotient directions.  When
\(r>t\), no admissible \((d+t)\)-dimensional common subspace can absorb a
positive-dimensional piece of every direction, and the full quotient cost
\(2^t\) is forced on at least one member of the family.

\section{A One-Core Size-Only Mechanism Barrier}
\label{app:one-core-barrier}
\label{app:one-core-barrier-v2}

The moderate-corruption learner certifies clean mass using the observed
size of one extracted core. We show that the \(K^{-1/2}\) scale is intrinsic
to this size-only proof architecture. The bound concerns that mechanism;
other recovery algorithms may use additional information.

The mechanism has two numerical ingredients.  First, at normalized energy
\(1/K\), one cannot universally force a bounded-small-doubling core larger
than order \(K^{-1/2}\) of the observation.  Second, if clean mass is inferred
only from a core's retained fraction, the elementary purity inequality used
in Appendix~\ref{app:structural-lifting} is sharp.

\begin{proposition}[Sharpness of size-only purity]
\label{prop:H-size-only-purity}
Fix \(0<\gamma\le1/2\) and \(\eta\in[0,1)\).  For arbitrarily large finite
sets \(A\), there exist \(B\) and \(Y\subseteq B\) such that
\[
|A\triangle B|\le\eta|A|,
\qquad
|Y|\ge\gamma|B|,
\]
and
\begin{equation}
\frac{|A\cap Y|}{|A|}
=
\bigl[\gamma-(1-\gamma)\eta\bigr]_+ + o_{|A|\to\infty}(1).
\label{eq:H-purity-sharp}
\end{equation}
Consequently, size information of this form alone can force a positive
universal clean fraction only when
\begin{equation}
\eta<\frac{\gamma}{1-\gamma}.
\label{eq:H-purity-positive}
\end{equation}
\end{proposition}

\begin{proof}
Let \(|A|=N\).  Choose a set \(D\) disjoint from \(A\) with
\(|D|=\lfloor\eta N\rfloor\), and put \(B:=A\cup D\).  Thus all corruption
consists of insertions and \(|A\triangle B|\le\eta|A|\).  Choose
\(Y\subseteq B\) of cardinality \(\lceil\gamma|B|\rceil\), taking points
from \(D\) before points from \(A\).  Then
\[
|A\cap Y|
=
\bigl[\lceil\gamma|B|\rceil-|D|\bigr]_+.
\]
Since \(|B|=(1+\eta+o(1))N\), division by \(N\) gives
\[
\frac{|A\cap Y|}{N}
=
\bigl[\gamma(1+\eta)-\eta\bigr]_++o(1)
=
\bigl[\gamma-(1-\gamma)\eta\bigr]_++o(1).
\]
The bracket is positive exactly under \eqref{eq:H-purity-positive}.
\end{proof}

The following construction is a disjoint affine-block
variant of the coordinate-axis example of
\citet{reiherschoen2024bsg}. Combining it with the sharp
purity bound gives the size-only lifting barrier below.

\begin{lemma}[Affine-block core-size obstruction]
\label{lem:H-affine-block-obstruction}
Let \(r\ge2\), let \(L\ge1\), and let \(\delta>0\).  There is a finite
nonempty set \(S\subseteq\mathbb F_2^n\) such that
\begin{equation}
E(S)\ge r^{-2}|S|^3,
\label{eq:H-affine-energy}
\end{equation}
and every nonempty \(Y\subseteq S\) satisfying
\[
|Y+Y|\le L|Y|
\]
obeys
\begin{equation}
|Y|\le\frac{1+\delta}{r}|S|.
\label{eq:H-affine-core-upper}
\end{equation}
At the same time, \(S\) contains a subset of size \(|S|/r\) with doubling
one.
\end{lemma}

\begin{proof}
Choose a power of two \(m\) so large that \(m\ge2L/\delta\).  Let
\[
G=E\oplus W_1\oplus\cdots\oplus W_r,
\]
where \(\dim E=r\), \(|W_i|=m\), and \(e_1,\ldots,e_r\) is a basis of
\(E\).  Define pairwise disjoint affine blocks
\[
S_i:=e_i+W_i,
\qquad
S:=\bigcup_{i=1}^r S_i.
\]
Then \(|S|=rm\), \(S_i+S_i=W_i\), and for \(i\ne j\),
\[
S_i+S_j=e_i+e_j+W_i+W_j
\]
has cardinality \(m^2\).  The cross-block sumsets belonging to distinct
unordered pairs \(\{i,j\}\) are disjoint.

We first compute enough energy.  At zero, the representation count is
\(r_S(0)=rm\).  For every nonzero \(d\in W_i\), there are exactly \(m\)
ordered representations inside \(S_i\), and for every
\(d\in S_i+S_j\), \(i<j\), there are exactly two ordered cross-block
representations.  Hence
\begin{align*}
E(S)
&=(rm)^2+r(m-1)m^2+4\binom r2m^2\\
&=rm^3+3r(r-1)m^2
\ge rm^3
= r^{-2}|S|^3,
\end{align*}
which proves \eqref{eq:H-affine-energy}.

Now let \(Y\subseteq S\), write
\[
a_i:=|Y\cap S_i|,
\qquad
M:=|Y|=\sum_i a_i.
\]
For \(i<j\), direct-sum uniqueness makes addition injective on
\((Y\cap S_i)\times(Y\cap S_j)\), and the resulting cross-block sumsets are
disjoint across unordered pairs.  Therefore
\begin{align*}
|Y+Y|
&\ge\sum_{i<j}a_ia_j
=\frac12\left(M^2-\sum_i a_i^2\right)\\
&\ge\frac12 M(M-m),
\end{align*}
because \(a_i\le m\) implies \(\sum_i a_i^2\le mM\).  If
\(|Y+Y|\le LM\) and \(M>0\), then \(M\le m+2L\).  Thus
\[
\frac{|Y|}{|S|}
\le
\frac1r\left(1+\frac{2L}{m}\right)
\le\frac{1+\delta}{r},
\]
which is \eqref{eq:H-affine-core-upper}.  Finally, every full block \(S_i\)
has size \(|S|/r\) and doubling one.
\end{proof}

\begin{theorem}[One-core size-only mechanism barrier]
\label{thm:H-one-core-barrier}
Let \(L:[1,\infty)\to[1,\infty)\) be any prescribed small-doubling
budget.  For every integer \(r\ge2\), every \(\delta>0\), and
\(K:=r^2\), there is a nonempty set \(S\subseteq\mathbb F_2^n\) with
\[
E(S)\ge K^{-1}|S|^3
\]
such that every nonempty \(Y\subseteq S\) satisfying
\[
|Y+Y|\le L(K)|Y|
\]
has
\begin{equation}
|Y|\le(1+\delta)K^{-1/2}|S|.
\label{eq:H-core-scale}
\end{equation}
Consequently, consider any proof architecture that
\begin{enumerate}
\item extracts one bounded-small-doubling core from an energy-
      \(1/K\) observation using only a universal retained-fraction guarantee
      \(|Y|\ge\gamma(K)|B|\); and
\item certifies clean latent mass in that same preselected core using only
      \(Y\subseteq B\), its relative size, and
      \(|A\triangle B|\le\eta|A|\).
\end{enumerate}
Then along \(K=r^2\) one necessarily has
\begin{equation}
\gamma(K)\le K^{-1/2},
\label{eq:H-universal-gamma}
\end{equation}
and this size-only route can certify a positive universal clean fraction
only below
\begin{equation}
\eta<\frac{1}{\sqrt K-1}.
\label{eq:H-contamination-barrier}
\end{equation}
\end{theorem}

\begin{proof}
Apply Lemma~\ref{lem:H-affine-block-obstruction} with
\(L=L(K)\).  Because \(\delta>0\) is arbitrary, no universally valid
retained fraction can exceed \(1/r=K^{-1/2}\); this proves
\eqref{eq:H-universal-gamma}.

For \(r\ge2\), \(\gamma(K)\le1/r\le1/2\).  By
Proposition~\ref{prop:H-size-only-purity}, a positive clean-mass guarantee
based only on retained size requires
\[
\eta<\frac{\gamma(K)}{1-\gamma(K)}.
\]
Since \(x/(1-x)\) is increasing on \([0,1)\),
\[
\frac{\gamma(K)}{1-\gamma(K)}
\le
\frac{1/r}{1-1/r}
=
\frac1{r-1}
=
\frac1{\sqrt K-1},
\]
which proves \eqref{eq:H-contamination-barrier}.
\end{proof}

The arbitrary-corruption learner instead peels disjoint cores and uses
their total captured clean mass. Averaging then finds a serving core for
each support, without requiring one preselected core to serve every support
on the basis of its size alone.

\section{Detailed Relation to Prior Work}
\label{app:related-work}

\subsection{Latent identifiability and common structural targets}

Latent-variable learning asks which features of an unobserved object can be
determined from observations. Nonlinear ICA and related models use auxiliary
structure to obtain identifiability~\citep{hyvarinen2019nonlinearica,khemakhem2020vaeica};
impossibility results for unsupervised disentanglement emphasize the role of
inductive bias~\citep{locatello2019disentanglement}. Under partial observation,
Yao et al.~\citep{yao2024multiview} and Xu et al.~\citep{xu2024sparsity} study
shared or sparse latent information. Syrota et al.~\citep{syrota2025metric}
study relations that remain identifiable when latent coordinates do not.

Here the latent object is a finite support, and the target is a subspace
whose covering guarantee holds across all compatible supports. This is a
validity requirement over an identification set; several subspaces may meet
it. The exact-recovery contrast in Appendix~\ref{app:exact-k1} also applies
when that set is a singleton: a common coarse cover can be computed with
polynomially many oracle calls while exact reconstruction requires
exponentially many.

\subsection{Robust, list-decodable, and finite-field subspace learning}

List-decodable learning permits a small inlier fraction and returns a short
list of hypotheses~\citep{charikar2017untrusted}. This approach has led to
robust mean and subspace recovery guarantees~\citep{diakonikolas2018listmean,raghavendra2020listsubspace,bakshi2021listsubspace}.
Our shared insertion--deletion budget makes a strict majority of $B$ genuine
under every compatible explanation. This property supports a stronger final
requirement: one subspace with a uniform covering guarantee. The query learner
uses a list to generate candidates, then selects a common output by fresh
samples.

Chen, De, and Vijayaraghavan~\citep{chen2021mixture} study mixtures of two
finite-field subspaces and identify LPN in the nested case. Our reduction
addresses quotient covers, for which one parity graph has a secret-independent
two-coset solution. Blocking makes a coarse output reveal a short secret list.
The frozen-table coupling then realizes one fixed observed support and
separately controls its symmetric-difference budget. The reduction allows
polynomial-size output lists and supplied latent cardinality.

\subsection{Query access and persistent noise}

Query learning treats access beyond passive examples as a computational
resource~\citep{angluin1988queries}. Local membership queries restrict that
resource~\citep{awasthi2013localmq}. Noisy parity learning~\citep{blum2003noisetolerant},
persistent-noise query models~\citep{jackson1999queriesnoise}, and parity
learning with queries~\citep{feldman2007parities} provide relevant precedents.
Feldman and Shah~\citep{feldman2009faulty} obtain conditional comparisons
between faulty-teacher access models using noisy-parity hardness.

We compare samples and membership queries for the same quotient-covering
loss on a fixed corrupted support. The statistical theorem establishes that
samples contain enough information; the matched comparison in
Section~\ref{sec:v3-main-results} isolates the computational contribution of
queries under the stated LPN assumption. Persistent-noise analyses already
use fixed random functions and collision couplings. Here the table is
indexed by complete blocked inputs, and the support promise is checked
separately without conditioning the simulator on its validity.

De, Nadimpalli, and Servedio~\citep{de2022sumset} study indicator-oracle
access to additive sets and approximate a relaxed sumset-size functional.
Our algorithms construct structural hypotheses and reusable implicit cores;
the cores must support adaptive membership evaluation and exact sampling
from the same realized set.

\subsection{Algorithmic additive structure}

The polynomial Freiman--Ruzsa theorem supplies the clean covering
guarantee~\citep{gowers2025marton}. Recent algorithms make clean PFR
constructive~\citep{arunachalam2026classicalpfr,castrosilva2026algorithmicpfr}.
Appendix~\ref{app:size-oblivious-pfr} adapts the restricted-homomorphism
primitive of Castro--Silva et al. to an input of unknown size, deterministic
resource caps, and conditional invocation after an adaptive past.

\citet{reiherschoen2024bsg} establish the essentially optimal
$\Theta(\sqrt{\alpha})$ retained-mass scale through
high-/low-multiplicity arguments and relation cleaning.
\citet{tulsiani2014quadratic} give randomized access
procedures for implicit BSG-related sets, with
set-sandwich guarantees controlling testing noise.
For explicit sets, algorithmic BSG appears in structured
3SUM~\citep{chan2015clustered3sum}.
\citet{song2026algorithmicbsg} reaches the optimal output-size
scale for integer sets using weighted shift sampling,
a shared codegree classifier, and sampled relation cleaning.

Our construction shares these structural ingredients.
It works through sample-and-query access to a potentially
exponentially large support and returns a finite record
defining one fixed subset. Its guarantee combines the sharp
retained-mass scale, a bound on invalid nonfailure outputs
without an energy promise, high-energy completeness, and
exact product sampling conditional on finite-block success.
These are the properties used by residual peeling and the
subsequent PFR calls. Appendix~\ref{app:persistent-bsg}
proves them.

\subsection{Exact affine recovery}

The exact decoder uses Goldreich--Levin heavy-Fourier
search~\citep{goldreich1989hardcore}. Related coding and learning results
include local list decoding of binary Reed--Muller codes~\citep{gopalan2008reedmuller},
sample learning of Fourier-sparse Boolean functions~\citep{haviv2015fouriersparse},
and property testing of affine-subspace indicators~\citep{bshouty2024affine}.
After localization, the affine-support indicator has degree equal to its
local codimension, which is bounded by the coarse covering quality.

The support-relative corruption model requires several steps before the
Fourier primitive can be applied: finding a containing flat, certifying its
latent density, and setting an observable coefficient threshold without
knowing the latent size. Candidate validation handles that unknown size.
The critical $1/3$ result measures the access cost of exact reconstruction in
this model. Its lower-bound pair has different latent cardinalities; the
known-cardinality refinement extends exact recovery below $1/2$.
These thresholds use affine rigidity at $K=1$. For general small doubling,
the main target remains a common structural cover.

\section{Algorithmic and Access-Model Questions}
\label{app:further-directions}

The statistical and query learners reach the same type of
common structural output through different candidate families.
This leaves two quantitative questions: how much covering
quality is lost by constructing candidates efficiently, and
how much computation is needed to represent those candidates.
A third question concerns which forms of access can support
the construction.

\subsection{Matching the statistical covering quality}

The statistical learner achieves cost
$2Kp(K)/\varepsilon$ by searching over clean PFR descriptions.
The general query learner constructs candidates through
structured-core extraction and then applies the common-output
filter. Its larger bound reflects both the quality of the
core covers and the transfer from a serving core to the
latent support.

These losses suggest two possible routes to improvement.
One is to strengthen the candidate-generation guarantee,
so that each compatible support has a serving entry with
smaller quotient cost. Another is to use more of the
relationship between the candidates when selecting a common
output. The present filter uses their individual coverage
scores and dimensions; it does not combine their algebraic
directions.

The resulting question is whether query access can attain
the statistical covering bound with comparable sample
complexity. This question can be studied separately from
the running-time dependence: better covering quality would
already strengthen the result even if the algorithm
remained XP.

\subsection{Representing residuals without recursive expansion}

The current XP exponent comes from implementing the
successive residuals
$R_j=B\setminus\bigcup_{i<j}Y_i$.
The number of peeling rounds depends only on the structural
parameters, but evaluating a core predicate invokes the
residual oracle from its construction stage. Expanding
these calls through earlier stages produces the
parameter-dependent polynomial degree.

An FPT implementation would need a representation whose
evaluation cost has a parameter-independent degree in $n$.
Possible approaches include sharing intermediate predicate
evaluations, compiling several peeling stages jointly,
or constructing a candidate family without sequential
residual evaluation. For an approach based on reuse,
the central quantity is the number of distinct derived
query points across the adaptive computation.

Any such representation must preserve the two properties
used downstream: membership refers to one fixed realized
core, and sampling supplies that same core to the PFR
routine. Improving this interface could therefore reduce
the computational cost while leaving the outer covering
and selection arguments intact.

\subsection{Learning with weaker access}

The query construction tests additive relations at points
formed from previous samples, such as $u+v+w$.
This suggests studying access restrictions in terms of
which constructed points can be tested. A restricted-query
learner would need to generate enough useful additive
relations within its allowed query set to support candidate
extraction.

Nonuniform sampling changes a different part of the
argument. The shared corruption budget controls
cardinalities, and exact uniform sampling turns those
bounds into probability-mass guarantees. For a general
sampling distribution on $B$, this correspondence needs
an additional assumption. A natural starting point is
a distribution whose density relative to $U_B$ is bounded
above and below, with the goal of tracking how the
selection thresholds and core guarantees depend on
those bounds.

Approximate membership also interacts with the sampling
interface: accepting points through an inaccurate
predicate can change the distribution supplied to the
next stage. A useful extension would specify the oracle
error model and control both the structural validity of
the returned core and the distributional error in its
downstream samples. These questions connect weaker
access directly to the two mechanisms used by the
learner: candidate generation and common-output selection.

\end{document}